\documentclass[11pt]{article}

\makeatletter
\newdimen\bibindent
\renewenvironment{thebibliography}[1]
  {\par\bigskip 
   \section*{\refname} 
   \@mkboth{\MakeUppercase{\refname}}{\MakeUppercase{\refname}}
   \list{\@biblabel{\arabic{enumi}}}%
        {\settowidth\labelwidth{\@biblabel{#1}}%
         \leftmargin\labelwidth
         \advance\leftmargin\labelsep
         \advance\leftmargin\bibindent
         \itemindent -\bibindent
         \listparindent \itemindent
         \parsep \z@
         \usecounter{enumi}%
         \let\p@enumi\@empty
         }%
     \renewcommand\newblock{\hskip .11em \@plus.33em \@minus.07em}%
     \sloppy\clubpenalty4000\widowpenalty4000%
     \frenchspacing\footnotesize}
  {\def\@noitemerr
     {\@latex@warning{Empty `thebibliography' environment}}%
   \endlist}
\makeatother

\makeatletter
\g@addto@macro\bfseries{\boldmath}
\makeatother

\usepackage{tikz}
 	 \usetikzlibrary{arrows,backgrounds}

\tikzstyle{shaded}=[fill=red!10!blue!20!gray!30!white]
\tikzstyle{unshaded}=[fill=white]
\tikzstyle{empty box}=[circle, draw, thick, fill=white, opaque, inner sep=2mm]
\tikzstyle{annular}=[scale=.7, inner sep=1mm, baseline]
\tikzstyle{rectangular}=[scale=.75, inner sep=1mm, baseline=-.1cm]
\usetikzlibrary{calc}

\newcommand{\sline}[5]{
\draw[color=white,line width = 4pt] (#1,#2) -- (#3,#4);
\draw[color=#5,ultra thick] (#1,#2) -- (#3,#4);
}

\newcommand{\aline}[5]{
\draw[color=#5,ultra thick] (#1,#2) -- (#3,#4);
}

\usepackage{bbm}
\usepackage{graphicx}
\usepackage{amsmath}
\usepackage{amssymb, amsthm}
\usepackage[cmtip,all]{xy}
\usepackage{mathrsfs}
\usepackage{amsmath,amscd}
\usepackage[colorlinks]{hyperref}
\usepackage{mathabx,epsfig}
\usepackage{longtable}
\usepackage{enumitem}

\hypersetup{
allcolors={blue}
}

\usepackage{changepage}

\newcommand{\lmatr}[1]{\ensuremath{\left[\begin{array}{#1}}}
\newcommand{\rmatr}{\ensuremath{\end{array}\right]}}
\newcommand{\lmatrp}[1]{\ensuremath{\left(\begin{array}{#1}}}
\newcommand{\rmatrp}{\ensuremath{\end{array}\right)}}
\newcommand{\lmatrd}[1]{\ensuremath{\left|\begin{array}{#1}}}
\newcommand{\rmatrd}{\ensuremath{\end{array}\right|}}

\newcommand{\df}{\ensuremath{\mathrel{\mathop:}=}}

\newcommand{\id}{\mathrm{id}}

\newcommand{\diag}{\mathrm{diag}}

\newcommand{\supp}{\mathrm{supp}}

\theoremstyle{definition}
\newtheorem{thm}{Theorem} [section]

\newtheorem{fact}[thm]{Fact}

\newtheorem{lem}[thm]{Lemma}
\newtheorem{rem}[thm]{Remark}

\usepackage[margin=1in]{geometry}

\newcommand{\lrsquiggly}{\xymatrix{{}\ar@{<~>}[r]&{}}}
\newcommand{\rsquiggly}{\xymatrix{{}\ar@{~>}[r]&{}}}

\allowdisplaybreaks

\newcommand{\ModB}{\ensuremath{\mathrm{Mod}_{\mathcal{B}}}}

\newcommand{\TLJcat}{\ensuremath{\mathcal{T}\!\mathcal{L}\mathcal{J}}}

\numberwithin{equation}{section}

\begin{document}

\title{Realizing the braided Temperley--Lieb--Jones\\ C*-tensor categories as 
Hilbert C*-modules}

\author{Andreas N\ae s Aaserud and David E.\ Evans}

\maketitle

\begin{abstract}
We associate to each Temperley--Lieb--Jones C*-tensor category 
$\TLJcat(\delta)$ with parameter $\delta$ in the discrete range 
$\{2\cos(\pi/(k+2))\,:\,k=1,2,\ldots\}\cup\{2\}$ a certain C*-algebra 
$\mathcal{B}$ of compact operators. We use the unitary braiding on 
$\TLJcat(\delta)$ to equip the category $\ModB$ of (right) Hilbert 
$\mathcal{B}$-modules with the structure of a braided C*-tensor category. We 
show that $\TLJcat(\delta)$ is equivalent, as a braided C*-tensor category, to 
the full subcategory $\ModB^f$ of $\ModB$ whose objects are those modules which 
admit a finite orthonormal basis. Finally, we indicate how these considerations 
generalize to arbitrary finitely generated rigid braided C*-tensor categories.
\end{abstract}

\section{Introduction}\label{sec:Intro}

In the present paper, the authors recast the Temperley--Lieb--Jones C*-tensor 
category $\TLJcat(\delta)$ with parameter $\delta$ in Jones' discrete range 
$\{2\cos(\pi/(k+2))\,:\,k = 1,2,\ldots\}\cup \{2\}$ (cf.\ \cite{J}) as a 
C*-tensor category of (right) Hilbert C*-modules, drawing inspiration from the 
work of Erlijman--Wenzl \cite{EW}, Hartglass--Penneys \cite{HP1} and Yuan 
\cite{Yuan}, among others.

\subsection{Background}

Temperley--Lieb algebras first appeared in the work of Temperley and Lieb 
\cite{TL} on Potts and ice-type models in statistical mechanics, in which they 
were defined in terms of generators and relations. These relations reappeared 
in the work of Jones \cite{J}, in which (quotients of) Temperley--Lieb algebras 
manifested as subalgebras of higher relative commutants of (von Neumann) 
subfactors (see also \cite{GHJ}). A description of the Temperley--Lieb algebras 
in terms of what are now known as Temperley--Lieb diagrams first appeared in 
the work of Kauffman \cite{Kau} (see also \cite{JR}), who was studying a knot 
invariant introduced by Jones \cite{J85}.
Later, it was realized that a diagrammatic description could be given for 
tensor categories (cf.\ e.g.\ \cite{Tur}) and standard invariants of subfactors 
(cf.\ Jones' introduction of subfactor planar algebras \cite{J2} based on the 
work of Popa \cite{Po}). In particular, diagrammatic Temperley--Lieb--Jones 
C*-tensor categories were considered (cf.\ e.g.\ \cite{Yam}, \cite{Coo}, 
\cite{EP12}), which can be viewed as arising from the Temperley--Lieb--Jones 
factor planar algebras (cf.\ \cite{J2}; see also \cite{MPS}, \cite{BHP}). When 
the parameter $\delta$ is confined to $\{2\cos(\pi/(k+2))\,:\,k=1,2,\ldots\}$, 
the associated Temperley--Lieb--Jones C*-tensor categories $\TLJcat(\delta)$ 
are known to describe (up to equivalence) categories that have appeared in 
various contexts, including
\begin{itemize}
    \item representations of affine Lie algebras and vertex operator algebras 
    arising from $\mathrm{SU}(2)$ Wess--Zumino--Witten models at finite levels 
    $k=1,2,\ldots$ in 2D conformal field theory (cf.\ e.g.\ \cite{HL} and the 
    references therein);
    \item representations of the loop group $L\mathrm{SU}(2)$ at finite levels $k=1,2,\ldots$ (cf.\ \cite{PS}, \cite{Wa2});
    \item representations of quantum $\mathrm{SU}(2)$ at certain roots of unity (cf.\ \cite{We2}).
\end{itemize}
We refer the reader to \cite{He} for an overview and further references. It should also be mentioned that $\TLJcat(\delta)$ can be recovered as the C*-tensor category of $M$-bimodules arising from certain subfactors $(N\subset M)$ (cf.\ \cite{Xu}; see also Remark 8.2 in \cite{PV}). A special feature of the C*-tensor categories $\TLJcat(\delta)$ with $\delta\in \{2\cos(\pi/(k+2))\,:\,k = 1,2,\ldots\}\cup \{2\}$ is the presence of a unitary braiding (cf.\ e.g.\ \cite{MPS}), which we will use extensively in the present paper.

\subsection{Motivation}

Ultimately, our goal of describing $\TLJcat(\delta)$ in terms of Hilbert C*-modules is motivated by a connection with $K$-theory (cf.\ e.g.\ \cite{Bl}, \cite{RLL}, \cite{HiR}), namely the theorem of Freed, Hopkins and Teleman (cf.\ \cite{FHT1,FHT2,FHT3}) describing the fusion ring of the category of level $k$ representations of the loop group $L\mathrm{SU}(2)$ in terms of twisted equivariant $K$-theory.
Related to this, we observed in \cite{AE} that the $K_0$-group of certain 
approximately finite-dimensional (AF) C*-algebras has a ring structure that is 
closely related to the fusion ring of $\TLJcat(\delta)$. For example, the 
$K_0$-group of the inductive limit $\mathrm{TLJ}_\infty(\delta) = \lim_{n} 
\mathrm{TLJ}_n(\delta)$ of Temperley--Lieb--Jones C*-algebras, whose Bratteli 
diagram is given in \cite{J}, is a localization of the fusion ring of 
$\TLJcat(\delta)$. The present paper is a result of our efforts to lift such a 
ring structure in $K_0$-theory to a tensor product structure on an underlying 
category of modules. We found it natural to use the framework of Hilbert 
C*-modules, which generalize both Hilbert spaces and vector bundles and find 
uses in diverse areas of mathematics, including $K$-theory, Kasparov's 
$K\!K$-theory, and quantum groups (cf.\ e.g.\ \cite{Lance}, \cite{Bl}).

\subsection{Related work}

Given a (small) rigid C*-tensor category $\mathcal{C}$, Yuan in \cite{Yuan} constructed a unital C*-algebra $\mathcal{A}$ and a fully faithful monoidal *-functor from $\mathcal{C}$ into the category ${}_{\mathcal{A}}\mathrm{Mod}_{\mathcal{A}}$ of finite type Hilbert C*-bimodules over $\mathcal{A}$, the tensor product in ${}_{\mathcal{A}}\mathrm{Mod}_{\mathcal{A}}$ being given by interior tensor product. A variant of Yuan's construction yields a fully faithful monoidal *-functor from $\TLJcat(\delta)$ into ${}_{\mathcal{A}}\mathrm{Mod}_{\mathcal{A}}$, where $\mathcal{A}$ is the unital AF-algebra whose Bratteli diagram arises from the fusion graph of $f^{(0)}\oplus f^{(1)}$ (in the notation of section \ref{sec:JW}). For example, when $\delta = 2\cos(\pi/5)$, this diagram is
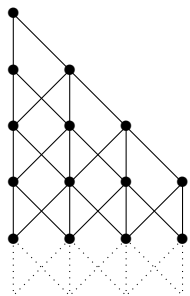
\begin{center}
\begin{tikzpicture}[scale=0.75]
      \draw (0,4) -- (0,0);
      \draw (1,3) -- (1,0);
      \draw (2,2) -- (2,0);
      \draw (3,1) -- (3,0);

      \draw (0,4) -- (3,1);
      \draw (0,3) -- (3,0);
      \draw (0,2) -- (2,0);
      \draw (0,1) -- (1,0);

      \draw (0,2) -- (1,3);
      \draw (0,1) -- (1,2);
      \draw (0,0) -- (2,2);
      \draw (1,0) -- (2,1);
      \draw (2,0) -- (3,1);

      \draw[dotted] (1,0) -- (0,-1);
      \draw[dotted] (1,0) -- (2,-1);
      \draw[dotted] (3,0) -- (2,-1);
      \draw[dotted] (0,0) -- (1,-1) -- (2,0) -- (3,-1);
      \draw[dotted] (0,0) -- (0,-1);
      \draw[dotted] (1,0) -- (1,-1);
      \draw[dotted] (2,0) -- (2,-1);
      \draw[dotted] (3,0) -- (3,-1);

      \node() at (0,4) {$\bullet$};
      \node() at (0,3) {$\bullet$};
      \node() at (0,2) {$\bullet$};
      \node() at (0,1) {$\bullet$};
      \node() at (0,0) {$\bullet$};

      \node() at (1,3) {$\bullet$};
      \node() at (1,2) {$\bullet$};
      \node() at (1,1) {$\bullet$};
      \node() at (1,0) {$\bullet$};

      \node() at (2,2) {$\bullet$};
      \node() at (2,1) {$\bullet$};
      \node() at (2,0) {$\bullet$};

      \node() at (3,1) {$\bullet$};
      \node() at (3,0) {$\bullet$};
    \end{tikzpicture}
\end{center}
In the present paper, we make use of Yuan's formalism in defining certain Hilbert spaces and bounded operators. In turn, Yuan was influenced by earlier realizations of C*-tensor categories in terms of bimodules over von Neumann algebras (for which we refer to the citations in \cite{Yuan}).

On the other hand, based on the work of Guionnet, Jones and 
Shlyakhtenko \cite{GJS}, Hartglass and Penneys in \cite{HP1}  
associated 
a C*-algebra 
$\mathcal{B}$ along with a Hilbert C*-bimodule $\mathcal{X}$ over 
$\mathcal{B}$ to an arbitrary factor planar algebra $\mathcal{P}_\bullet$. They 
then fed this bimodule into a construction due to Pimsner (cf.\ 
\cite{Pi}) in order to associate Cuntz and Toeplitz type algebras to planar 
algebras. When $\mathcal{P}_\bullet$ is the Temperley--Lieb--Jones planar 
algebra 
with parameter $\delta$, $K_0(\mathcal{B})$ is isomorphic to the fusion ring of 
$\TLJcat(\delta)$. This led us to consider modules over a variant of the 
C*-algebra $\mathcal{B}$.

It should also be mentioned that the tensor product that is defined in the 
present paper is related to a tensor product of modules over Temperley--Lieb 
algebras with varying numbers of strands that was introduced in \cite{RS, RS2} 
and studied further in \cite{GV}, \cite{B}, \cite{GS}. Moreover, the definition 
of the modified version of the C*-algebra $\mathcal{B}$ of Hartglass and 
Penneys that we use is influenced by the notion of dilute Temperley--Lieb 
algebras, which originated in \cite{Grimm}, \cite{BS}.

\subsection{Structure of the paper}

Section \ref{sec:Intro} is this introduction. In section \ref{sec:Prelim}, we 
cover well-known preliminary material on Hilbert space operators, Hilbert 
C*-modules, C*-tensor categories and the Temperley--Lieb--Jones C*-tensor 
categories $\TLJcat(\delta)$ with $\delta\in \{2\cos(\pi/(k+2))\,:\,k = 
1,2,\ldots\}\cup \{2\}$.

Our contribution starts in section \ref{sec:BandPhi}. Using the formalism of 
Yuan and the notion of dilute Temperley--Lieb diagrams (as presented in 
\cite{BS}), we construct a variant of the C*-algebra $\mathcal{B}$ of Hartglass 
and Penneys (section \ref{sec:B}). Next, we explain a way to associate 
operators in $\mathcal{B}$ and its strong closure to certain infinite diagrams 
(section \ref{sec:diagram}). Using an idea of Erlijman and Wenzl (cf.\ 
\cite{EW}), we then harness the unitary braiding on $\TLJcat(\delta)$ to define 
a *-homomorphism $\Phi\colon\mathcal{B}\otimes\mathcal{B}\to\mathcal{B}$ by 
superposition of diagrams (section \ref{sec:Phi}) and observe that the product 
on $K_0(\mathcal{B})$ induced by $\Phi$ recaptures the product in the fusion 
ring of $\TLJcat(\delta)$ (Remark \ref{rem:Kring}).

In section \ref{sec:Monoidal}, we first use $\Phi$ as well as interior and exterior tensor products of Hilbert C*-modules to define a tensor product of Hilbert $\mathcal{B}$-modules (section \ref{sec:Tensor}).
We next use this tensor product to equip the category $\ModB$ of Hilbert $\mathcal{B}$-modules with the structure of a C*-tensor category (section \ref{sec:ModB}) and supply it with a unitary braiding (section \ref{sec:Braiding}).

In section \ref{sec:Section5}, we first define a *-functor $F$ from $\TLJcat(\delta)$ into $\ModB$ and show that it is monoidal and braided (section \ref{sec:Functor}). In section \ref{sec:Equiv}, we then use $F$ to prove Theorem \ref{thm}, which states that $\TLJcat(\delta)$ is equivalent, as a braided C*-tensor category, to the full subcategory $\ModB^f$ of $\ModB$ whose objects are those modules which admit a finite orthonormal basis (and which is introduced in section \ref{sec:ModBf}). Thereafter, we note that the tensor category $\ModB$ ``categorifies'' the ring $K_0(\mathcal{B})$ (Remark \ref{rem:categorification}) and indicate how one can prove a version of Theorem \ref{thm} for arbitrary finitely generated rigid braided C*-tensor categories (Remark \ref{rem:gen}).

Finally, in section \ref{sec:Conclusion}, we pose some questions concerning representability of C*-tensor categories on Hilbert C*-modules and realizability of the representation category of the Virasoro algebra.

\section{Preliminaries}\label{sec:Prelim}

\subsection{Operators on Hilbert space}

In this paper, we consider operators on a complex Hilbert space $H$. We denote by $\mathbb{B}(H)$ the space of all bounded linear operators on $H$, which comes equipped with a plethora of topologies. In this paper, we will restrict attention to the norm topology, which is induced by the operator norm, and the strong operator topology, which is the topology of pointwise convergence in the norm on $H$, that is, $a_n\to a$ strongly if and only if $\|a_n(\xi)-a(\xi)\|\to 0$ for all $\xi\in H$. We will need the following standard fact.

\begin{fact}\label{lem:SOT}
Let $(a_n)_{n=0}^\infty$ be a bounded sequence in $\mathbb{B}(H)$ such that $a_na_m^* = a_n^*a_m = 0$ whenever $n\neq m$. Then
    $
        \sum_{n\geq 0}a_n
    $
    and
    $
        \sum_{n\geq 0}a_n^*
    $
    converge strongly in $\mathbb{B}(H)$.
\end{fact}

The normed space $\mathbb{B}(H)$ is a C*-algebra. It contains the C*-subalgebra $\mathbb{K}(H)$ of compact operators, which is the smallest C*-subalgebra of $\mathbb{B}(H)$ that contains all operators of finite rank. The following standard fact will be useful to us.

\begin{fact}\label{fact:SOTcpt}
    Let $(a_n)_{n=0}^\infty$ be a sequence in $\mathbb{B}(H)$ that converges strongly to some operator $a$. For any compact operator $x\in\mathbb{K}(H)$, we have that $\|a_nx-ax\|\to 0$.
\end{fact}

\subsection{Hilbert C*-modules}\label{sec:HilbMod}

A (right) Hilbert C*-module over a C*-algebra $B$ is a (right) $B$-module $M$ equipped with a $B$-valued inner product $\langle\cdot,\cdot\rangle\colon M\times M\to B$ such that $\xi\mapsto \|\langle\xi,\xi\rangle\|^{1/2}$ is a complete norm. The general theory of such modules is laid out very carefully in \cite{Lance}, to which we refer for precise definitions and all the information that the reader will need.

Let us comment on the notation and terminology used in the present paper. We use the symbol $\boxtimes$ for the exterior tensor product of Hilbert C*-modules (so that if $M$ is a Hilbert $A$-module and $N$ is a Hilbert $B$-module then $M\boxtimes N$ is a Hilbert $(A\otimes B)$-module) and the symbol $\otimes_\phi$ for the interior tensor product with respect to a *-homomorphism $\phi$.
By an orthonormal basis for a Hilbert $B$-module $M$, we shall mean a (possibly infinite) family $(\xi_j)_{j\in J}$ of elements in $M$ such that
\begin{itemize}
    \item[(i)] $\langle \xi_i,\xi_j\rangle = 0$ whenever $i\neq j$;
    \item[(ii)] $\langle \xi_j,\xi_j\rangle$ is a projection in $B$ for all $j\in J$;
    \item[(iii)] the Fourier expansion $\eta = \sum_{j\in J} \xi_j\langle \xi_j,\eta\rangle$ is valid for all $\eta\in M$.
\end{itemize}

\subsection{C*-tensor categories}

Below, we recall the notions of C*-tensor categories, semisimplicity, unitary braidings and monoidal *-functors. We refer to \cite{GLR}, \cite{DR}, \cite{LR}, \cite{EW} and \cite{NT} for more information.

\subsubsection{Definition of a C*-tensor category}\label{sec:TensorCat}

A category $\mathcal{C}$ is called a C*-tensor category if the following conditions are satisfied (where $\pi$, $\rho$ and $\nu$ denote arbitrary objects in $\mathcal{C}$):
\begin{itemize}
        \item[(1)] Each morphism set $\mathrm{Hom}(\pi,\rho)$ is a complex Banach space. Moreover, composition is bilinear and $\|fg\|\leq\|f\|\|g\|$ for any pair $(f,g)$ of composable morphisms.
        \item[(2)] There is an antilinear contravariant functor *$\colon \mathcal{C}\to\mathcal{C}$ such that $\pi^* = \pi$ for all objects $\pi$, $f^{**} = f$ for all morphisms $f$, and the C*-identity
        $
            \|f^*f\| = \|f\|^2
        $
        holds for all morphisms $f$. In particular, each endomorphism space $\mathrm{End}(\pi)\df \mathrm{Hom}(\pi,\pi)$ is a unital C*-algebra.
        \item[(3)] For any $f\in \mathrm{Hom}(\pi,\rho)$, the morphism $f^*f$ is a positive element of $\mathrm{End}(\pi)$.
        \item[(4)] There is a bilinear bifunctor $\otimes\colon \mathcal{C}\times\mathcal{C}\to\mathcal{C}$ and natural unitary isomorphisms
        $
            \alpha_{\pi,\rho,\nu}\colon (\pi\otimes\rho)\otimes \nu\to \pi\otimes (\rho\otimes\nu)
        $
        (called associators or associativity constraints) satisfying the pentagon identity (see Definition 2.1.1(iii) of \cite{NT} or equation (\ref{eq:pentagon}) below). [By definition, a (unitary) isomorphism in $\mathcal{C}$ is a morphism $u$ such that $u^*u = \id$ and $uu^* = \id$.]
        \item[(5)] There is a distinguished object $\mathbbm{1}$ in $\mathcal{C}$ (called the tensor unit) and natural unitary isomorphisms
        $
            \Psi_\pi^\ell\colon \mathbbm{1}\otimes\pi\to\pi$ and $\Psi^r_\pi\colon \pi\otimes\mathbbm{1}\to\pi
        $
        (called left and right unit constraints) satisfying the triangle identity (see Definition 2.1.1(iv) of \cite{NT} or equation (\ref{eq:triangle}) below).
        \item[(6)] $(f\otimes g)^* = f^*\otimes g^*$ for all morphisms $f$ and $g$.
        \item[(7)] The category $\mathcal{C}$ has subobjects and finite direct sums (see Definition 2.1.1(vi), (vii) of \cite{NT}).
        \item[(8)] The tensor unit is simple. [An object $\pi$ in $\mathcal{C}$ is said to be simple if $\mathrm{End}(\pi)=\mathbb{C}\id_\pi$.]
\end{itemize}

A C*-tensor category is said to be strict if the associators and unit constraints are identity morphisms.

\subsubsection{Semisimplicity}\label{sec:semisimple}

Briefly speaking, a C*-tensor category $\mathcal{C}$ is said to be semisimple if every object in $\mathcal{C}$ is isomorphic to a finite direct sum of simple objects. We next explain what this means in detail. Pick a set $\mathcal{S}$ of mutually non-isomorphic simple objects such that every simple object in $\mathcal{C}$ is isomorphic to some $s\in\mathcal{S}$. Given an object $\rho$ in $\mathcal{C}$, there exist non-negative integers $N^s$ (with $N^s = 0$ for all but finitely many $s$) such that
$
    \rho\cong \bigoplus_{s\in\mathcal{S}} s^{\oplus N^s},
$
where $s^{\oplus N^s}$ denotes a direct sum of $N^s$ copies of $s$. This means that, for each $s$ with $N^s>0$, there exist $N^s$ morphisms $v_{s,1},\ldots,v_{s,N^s}\in\mathrm{Hom}(s,\rho)$ such that $v_{s,j}^*v_{s,j}=\id_s$ for all $j$ and
$
    \id_\rho = \sum_{s\in\mathcal{S}}\sum_{j=1}^{N^s} v_{s,j}v_{s,j}^*.
$
In fact, $v_{s,1},\ldots,v_{s,N^s}$ form an orthonormal basis for $\mathrm{Hom}(s,\rho)$ equipped with the inner product $\langle\cdot,\cdot\rangle$ given by $\langle \xi,\eta\rangle\, \id_s = \xi^*\eta$ for $\xi,\eta\in\mathrm{Hom}(s,\rho)$. The number $N^s$ is called the multiplicity of $s$ in $\rho$ and is sometimes denoted by $(s,\rho)$. We write $s\prec \rho$ if $(s,\rho)>0$.
Since we mention it in a few places, we also recall that the fusion ring $\mathbb{Z}[\mathcal{S}]$ of $\mathcal{C}$ is the free abelian group generated by $\mathcal{S}$ and equipped with the product
$
    s\cdot t = \sum_{r\in\mathcal{S}} (r, s\otimes t) r.
$

\subsubsection{Unitary braidings}

A unitary braiding $\sigma$ on a C*-tensor category $\mathcal{C}$ is an assignment of a unitary isomorphism
		$
			\sigma_{\pi,\rho}\colon \pi\otimes\rho\to \rho\otimes\pi
		$
		to every pair $(\pi,\rho)$ of objects in $\mathcal{C}$, natural in $\pi$ and $\rho$, satisfying the hexagon identities (see \cite{Kassel} or equations (\ref{eq:hexagon1}) and (\ref{eq:hexagon2}) below). As in \cite{EW}, we call a C*-tensor category with a choice of unitary braiding a braided C*-tensor category.

\subsubsection{Monoidal functors}

A functor $F\colon \mathcal{C}\to\mathcal{D}$ between C*-tensor categories $\mathcal{C}$ and $\mathcal{D}$ is called a *-functor if $F$ is linear and satisfies $F(f^*) = F(f)^*$ for all morphisms $f$. It is said to be monoidal (or to be a tensor functor) if there are natural unitary isomorphisms
$
    J_{\pi,\rho}\colon F(\pi)\otimes F(\rho)\to F(\pi\otimes\rho)
$
that are compatible with the associators and unit constraints (see Definition 2.1.3 of \cite{NT} or equations (\ref{eq:monoidal})--(\ref{eq:unital_left}) below). If $F$ is a monoidal *-functor and $\mathcal{C}$ and $\mathcal{D}$ are both braided then we say that $F$ is braided if the isomorphisms $J$ are compatible with the braiding (see equation (\ref{eq:braided}) below).

\subsection{The Temperley--Lieb--Jones categories}

In this section, we recall the notion of Temperley--Lieb diagrams and of 
certain vector spaces, algebras and categories that one can associate to them.

\subsubsection{Temperley--Lieb--Jones algebras}

We recall first the notion of an $(m,n)$-Temperley--Lieb diagram (for $m,n\geq 
0$ of equal parity), which first appeared in \cite{Kau}. Such a diagram 
consists of $(m+n)/2$ non-crossing smooth strands inside a rectangle with $m$ 
nodes (or marked points) on the left side and $n$ nodes on the right side, each 
node being connected to a unique strand. (Some examples are shown in the next 
figure.)
Given $\delta\in\mathbb{C}$, denote by $\mathrm{TL}^0_{m,n}(\delta)$ the formal 
complex linear span of all isotopy classes of $(m,n)$-Temperley--Lieb diagrams 
and define a product
$
    \mathrm{TL}^0_{m,n}(\delta)\times \mathrm{TL}^0_{n,k}(\delta)\to \mathrm{TL}^0_{m,k}(\delta)
$
as follows. In order to multiply an $(m,n)$-Temperley--Lieb diagram by an 
$(n,k)$-Temperley--Lieb diagram, start by juxtaposing them, matching up the 
nodes to form a new diagram. Next, remove each closed loop at the cost of 
multiplying by the scalar $\delta$. The following figure gives an example of 
the product of a $(2,4)$-Temperley--Lieb diagram and a $(4,0)$-Temperley--Lieb 
diagram.
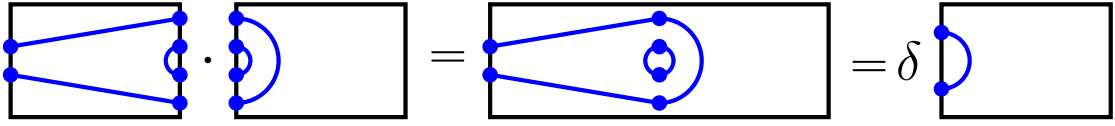
\begin{center}
	\begin{tikzpicture} [scale=.75]
		\begin{scope}[ultra thick]
			\node() at (3.5,1) {\tiny $\bullet$};
            \node() at (7.75,1) {\huge $=$};
			\node() at (15.5,1) {\huge $=\!\delta$};
			
			\draw (0,0) rectangle (3,2);
			
			\draw (4,0) rectangle (7,2);
			
            \draw (8.5,0) rectangle (14.5,2);
			
			\draw (16.5,0) rectangle (19.5,2);
			
			\begin{scope}[color=blue]
			\draw (0,1.25) -- (3,1.75);
            \draw (0,0.75) -- (3,0.25);
			\draw (3,1.25) arc (90:270:.25);
			
			\draw (4,1.75) arc (90:-90:.75);
			\draw (4,1.25) arc (90:-90:.25);

			\draw (8.5,1.25) -- (11.5,1.75);
            \draw (8.5,0.75) -- (11.5,0.25);
			\draw (11.5,1.25) arc (90:270:.25);

			\draw (11.5,1.75) arc (90:-90:.75);
			\draw (11.5,1.25) arc (90:-90:.25);
			
			\draw (16.5,1.5) arc (90:-90:.5);

			\draw[fill=blue] (0,1.25) circle (.1cm);
			\draw[fill=blue] (0,0.75) circle (.1cm);
			
			\draw[fill=blue] (3,1.75) circle (.1cm);
			\draw[fill=blue] (3,1.25) circle (.1cm);
			\draw[fill=blue] (3,0.75) circle (.1cm);
			\draw[fill=blue] (3,0.25) circle (.1cm);
			
			\draw[fill=blue] (4,1.75) circle (.1cm);
			\draw[fill=blue] (4,1.25) circle (.1cm);
			\draw[fill=blue] (4,0.75) circle (.1cm);
			\draw[fill=blue] (4,0.25) circle (.1cm);

			\draw[fill=blue] (8.5,1.25) circle (.1cm);
			\draw[fill=blue] (8.5,0.75) circle (.1cm);
			
			\draw[fill=blue] (11.5,1.75) circle (.1cm);
			\draw[fill=blue] (11.5,1.25) circle (.1cm);
			\draw[fill=blue] (11.5,0.75) circle (.1cm);
			\draw[fill=blue] (11.5,0.25) circle (.1cm);
			
			\draw[fill=blue] (16.5,1.5) circle (.1cm);
			\draw[fill=blue] (16.5,0.5) circle (.1cm);
			\end{scope}
		\end{scope}
	\end{tikzpicture}
\end{center}
In particular, $\mathrm{TL}^0_{n,n}(\delta)$ is an associative algebra, which 
is known as the $n$'th Temperley--Lieb algebra.
One can define a linear trace $\mathrm{Tr}_n^{\mathrm{TL}}$ on 
$\mathrm{TL}^0_{n,n}(\delta)$ as follows. If $D$ is an $(n,n)$-Temperley--Lieb 
diagram then $\mathrm{Tr}_n^{\mathrm{TL}}(D)$ is defined by a picture such as 
the one below (in which $n=3$), which is turned into a scalar by removing 
closed loops as explained above. (This trace is usually called a Markov trace.)
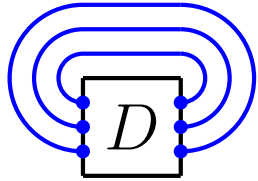
\begin{center}
	\begin{tikzpicture} [scale=.65]
		\begin{scope}[ultra thick]
            \node() at (1,1) {\Huge $D$};

			\draw (0,0) -- (0,2);
			\draw (0,2) -- (2,2);
			\draw (2,0) -- (2,2);
            \draw (0,0) -- (2,0);
			
			\begin{scope}[color=blue]
			\draw (0,2.5) -- (2,2.5);
            \draw (0,3.0) -- (2,3.0);
            \draw (0,3.5) -- (2,3.5);
			
            \draw (0,2.5) arc (90:270:0.5);
            \draw (0,3.0) arc (90:270:1.0);
            \draw (0,3.5) arc (90:270:1.5);

            \draw (2,2.5) arc (90:-90:0.5);
            \draw (2,3.0) arc (90:-90:1.0);
            \draw (2,3.5) arc (90:-90:1.5);

			\draw[fill=blue] (0,1.5) circle (.1cm);
			\draw[fill=blue] (0,1.0) circle (.1cm);
            \draw[fill=blue] (0,0.5) circle (.1cm);
			
			\draw[fill=blue] (2,1.5) circle (.1cm);
			\draw[fill=blue] (2,1.0) circle (.1cm);
            \draw[fill=blue] (2,0.5) circle (.1cm);

			\end{scope}
		\end{scope}
	\end{tikzpicture}
\end{center}
Moreover, one can define an antilinear *-operation $\mathrm{TL}^0_{m,n}(\delta)\to \mathrm{TL}^0_{n,m}(\delta)$ by reflecting diagrams about a vertical axis.

Jones famously proved (cf.\ \cite{J}) that the linear trace $\mathrm{Tr}_n^{\mathrm{TL}}$ is positive for all $n$ if and only if $\delta\in \{2\cos(\pi/(k+2))\,:\,k=1,2,\ldots\}\cup[2,\infty)$. Given $\delta$ in this range, put
$
    \mathrm{TLJ}_{m,n}(\delta) = \mathrm{TL}^0_{m,n}(\delta)/\{x\in\mathrm{TL}^0_{m,n}(\delta)\,:\,\mathrm{Tr}_n(x^*x) = 0\}.
$
Then the product above descends to a product
$
    \mathrm{TLJ}_{m,n}(\delta)\times \mathrm{TLJ}_{n,k}(\delta)\to \mathrm{TLJ}_{m,k}(\delta),
$
the above *-operation descends to a *-operation
$
    \mathrm{TLJ}_{m,n}(\delta)\to \mathrm{TLJ}_{n,m}(\delta),
$
and the trace $\mathrm{Tr}_n^{\mathrm{TL}}$ descends to a positive faithful 
trace on $\mathrm{TLJ}_{n,n}(\delta)$. Thus, $\mathrm{TLJ}_{n}(\delta)\df 
\mathrm{TLJ}_{n,n}(\delta)$ is a finite-dimensional C*-algebra, which is known 
as the $n$'th Temperley--Lieb--Jones C*-algebra.

\subsubsection{Temperley--Lieb--Jones C*-tensor categories}\label{sec:TLJ}

Let $\delta\in \{2\cos(\pi/(k+2))\,:\,k=1,2,\ldots\}\cup[2,\infty)$ be given. 
The Temperley--Lieb--Jones (or reduced Temperley--Lieb) C*-tensor category 
$\TLJcat(\delta)$ with parameter $\delta$ is defined as follows. Its objects 
are all formal finite sums $P_1\oplus\cdots\oplus P_k$, where $P_j$ is a 
projection in the C*-algebra $\mathrm{TLJ}_{n_j}(\delta)$ for each $j$. Given 
projections $P\in \mathrm{TLJ}_{n}(\delta)$ and $Q\in 
\mathrm{TLJ}_{m}(\delta)$, the morphism set $\mathrm{Hom}(P,Q)$ is 
$Q\mathrm{TLJ}_{m,n}(\delta)P$. More generally, given objects $\oplus_{j=1}^k 
P_j$ and $\oplus_{i=1}^r Q_i$, the morphism set $\mathrm{Hom}(\oplus_{j=1}^k 
P_j,\oplus_{i=1}^r Q_i)$ consists of all $r\times k$-matrices whose $(i,j)$'th 
entry is in $\mathrm{Hom}(P_j,Q_i)$. Composition of morphisms is given by 
multiplication of Temperley--Lieb diagrams combined with matrix multiplication.
The tensor product in $\TLJcat(\delta)$ is defined as follows. Given projections $P\in \mathrm{TLJ}_{n}(\delta)$ and $Q\in \mathrm{TLJ}_{m}(\delta)$, the tensor product $P\otimes Q$ is formed by stacking $P$ on top of $Q$ (or rather by the bilinear extension of this procedure applied to pairs of diagrams) to obtain a projection in $\mathrm{TLJ}_{n+m}(\delta)$. The tensor product of two objects $\oplus_{i=1}^k P_i$ and $\oplus_{j=1}^r Q_j$ is simply $\oplus_{(i,j)} (P_i\otimes Q_j)$. The tensor product of morphisms is given by vertical stacking combined with tensor multiplication of matrices, i.e.,
$
    (a_{ij})_{i,j}\otimes (b_{kl})_{k,l} = (a_{ij}\otimes b_{kl})_{(i,k),(j,l)}.
$
One can show that $\TLJcat(\delta)$ is a strict C*-tensor category, whose 
tensor unit is the empty Temperley--Lieb diagram.

\subsubsection{Jones--Wenzl projections}\label{sec:JW}

For any $\delta\in \{2\cos(\pi/(k+2))\,:\,k=1,2,\ldots\}\cup[2,\infty)$, the 
C*-tensor category $\TLJcat(\delta)$ is semisimple. Up to unitary isomorphism, 
the simple objects are the so-called Jones--Wenzl projections (cf.\ \cite{We}). 
If $\delta \geq 2$ then the Jones--Wenzl projections form an infinite sequence 
$(f^{(n)})_{n=0}^\infty$ with $f^{(n)}\in\mathrm{TLJ}_{n}(\delta)$ for all $n$, 
where
$f^{(0)}$ is the empty diagram and $f^{(1)}$ is a single strand. The remaining 
Jones--Wenzl projections are defined via Wenzl's recursive formula (see e.g.\ 
equation (2.1) in \cite{MPS}, in which $\delta$ is equal to $q+q^{-1}$ in their 
notation). It is a fact that
$
    f^{(1)}\otimes f^{(n)} \cong f^{(n-1)}\oplus f^{(n+1)}
$
in $\TLJcat(\delta)$ for all $n\geq 1$.
If $\delta = 2\cos(\pi/(k+2))$ with $k\geq 1$ then the Jones--Wenzl projections 
form a finite sequence $f^{(0)}, f^{(1)}, \ldots, f^{(k)}$, defined recursively 
as above. In this case,
$
    f^{(1)}\otimes f^{(n)} \cong f^{(n-1)}\oplus f^{(n+1)}
$
in $\TLJcat(\delta)$ for $1\leq n\leq k-1$ while $f^{(1)}\otimes f^{(k)}\cong f^{(k-1)}$.

In either case, the category $\TLJcat(\delta)$ is generated by the object $\pi 
= f^{(1)}$ in the sense that every simple object occurs as a direct summand of 
some tensor power $\pi^{\otimes n}$ of $\pi$. Note in this 
connection that $\mathrm{Hom}(\pi^{\otimes n},\pi^{\otimes m}) = 
\mathrm{TLJ}_{m,n}(\delta)$ for all $n,m\geq 0$.

\subsubsection{The unitary braiding}

If $\delta\in \{2\cos(\pi/(k+2))\,:\,k=1,2,\ldots\}\cup\{2\}$ then $\TLJcat(\delta)$ is a braided C*-tensor category. Specifically, one defines a unitary braiding $\sigma^{\mathrm{TL}}$ as follows. Consider the unitary Kauffman element
\begin{center}
    \begin{tikzpicture}
        \begin{scope}[ultra thick]
            \node() at (2.3,0.5) {$\sigma^{\mathrm{TL}}_{\pi,\pi}\,=\,z^{-1}$};
            \node() at (5.45,0.5) {$+\,\,z$};

            \draw (3.5,0) rectangle (4.5,1);

            \begin{scope}[color=blue]
            \aline{3.5}{0.75}{4.5}{0.75}{blue}
            \aline{3.5}{0.25}{4.5}{0.25}{blue}

            \draw[fill=blue] (3.5,0.75) circle (.1cm);
            \draw[fill=blue] (4.5,0.25) circle (.1cm);
            \draw[fill=blue] (3.5,0.25) circle (.1cm);
            \draw[fill=blue] (4.5,0.75) circle (.1cm);
            \end{scope}

            \draw (6,0) rectangle (7,1);

            \begin{scope}[color=blue]
            \draw (6,0.75) arc (90:-90:.25);
            \draw (7,0.75) arc (90:270:.25);

            \draw[fill=blue] (6,0.75) circle (.1cm);
            \draw[fill=blue] (7,0.25) circle (.1cm);
            \draw[fill=blue] (6,0.25) circle (.1cm);
            \draw[fill=blue] (7,0.75) circle (.1cm);
            \end{scope}
        \end{scope}
    \end{tikzpicture}
\end{center}
of $\mathrm{TLJ}_2(\delta)$, where $z = i e^{\pi i/[2(k+2)]}$ if $\delta = 2\cos(\pi/(k+2))$ while $z = i$ if $\delta = 2$.
We will use the following conventional graphical representation of the Kauffman element as a crossing.
\begin{center}
    \begin{tikzpicture}
        \begin{scope}[ultra thick]
            \node() at (-0.8,0.5) {$\sigma^{\mathrm{TL}}_{\pi,\pi}\,=\,$};

            \draw (0,0) rectangle (1, 1);

            \begin{scope}[color=blue]
            \aline{0}{.75}{1}{.25}{blue}

            \draw[fill=blue] (0,0.75) circle (.1cm);
            \draw[fill=blue] (1,0.25) circle (.1cm);
            \end{scope}

            \begin{scope}[color=red]
            \sline{0}{.25}{1}{.75}{red}

            \draw[fill=red] (0,0.25) circle (.1cm);
            \draw[fill=red] (1,0.75) circle (.1cm);
            \end{scope}

            \node() at (4.8,0.5) {$\big(\sigma^{\mathrm{TL}}_{\pi,\pi}\big)^{-1}=\,$};

            \draw (6,0) rectangle (7, 1);

            \begin{scope}[color=red]
            \aline{6}{.25}{7}{.75}{red}

            \draw[fill=red] (6,0.25) circle (.1cm);
            \draw[fill=red] (7,0.75) circle (.1cm);
            \end{scope}

            \begin{scope}[color=blue]
            \sline{6}{.75}{7}{.25}{blue}

            \draw[fill=blue] (6,0.75) circle (.1cm);
            \draw[fill=blue] (7,0.25) circle (.1cm);
            \end{scope}
        \end{scope}
    \end{tikzpicture}
\end{center}
Using it, one can define a unitary element $\sigma_{\pi^{\otimes n},\pi^{\otimes m}}^{\mathrm{TL}}$ of $\mathrm{End}(\pi^{\otimes (n+m)}) = \mathrm{TLJ}_{n+m}(\delta)$ by a braid diagram like the one below (which corresponds to the case $n = 2$ and $m = 3$).
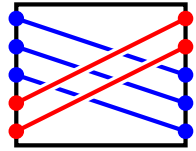
\begin{center}\label{crossings}
    \begin{tikzpicture}[scale=.75]
        \begin{scope}[ultra thick]
            \draw (0,0.5) rectangle (3, 3);

            \aline{0}{2.75}{3}{1.75}{blue}
            \aline{0}{2.25}{3}{1.25}{blue}
            \aline{0}{1.75}{3}{0.75}{blue}

            \sline{0}{1.25}{3}{2.75}{red}
            \sline{0}{0.75}{3}{2.25}{red}

            \begin{scope}[color=blue]
                \draw[fill=blue] (0,2.75) circle (.1cm);
                \draw[fill=blue] (0,2.25) circle (.1cm);
                \draw[fill=blue] (0,1.75) circle (.1cm);
                \draw[fill=blue] (3,1.75) circle (.1cm);
                \draw[fill=blue] (3,1.25) circle (.1cm);
                \draw[fill=blue] (3,0.75) circle (.1cm);
            \end{scope}

            \begin{scope}[color=red]
                \draw[fill=red] (3,2.75) circle (.1cm);
                \draw[fill=red] (3,2.25) circle (.1cm);
                \draw[fill=red] (0,1.25) circle (.1cm);
                \draw[fill=red] (0,0.75) circle (.1cm);
            \end{scope}
        \end{scope}
    \end{tikzpicture}
\end{center}
Given projections $P\in \mathrm{End}(\pi^{\otimes n}) = 
\mathrm{TLJ}_{n}(\delta)$ and $Q\in \mathrm{End}(\pi^{\otimes m}) = 
\mathrm{TLJ}_{m}(\delta)$, one defines a unitary isomorphism 
$\sigma^{\mathrm{TL}}_{P,Q}$ in $\mathrm{Hom}(P\otimes Q,Q\otimes P)$ by 
$\sigma_{P,Q}^{\mathrm{TL}} = \sigma_{\pi^{\otimes n},\pi^{\otimes 
m}}^{\mathrm{TL}}\circ (P\otimes Q)$. To see that 
$\sigma^{\mathrm{TL}}_{P,Q}$ is indeed 
an element of $\mathrm{Hom}(P\otimes Q,Q\otimes P) = (Q\otimes 
P)\mathrm{TLJ}_{n+m}(\delta)(P\otimes Q)$, one 
uses the isotopy 
invariance of the Temperley--Lieb diagrams along with the following two 
identities, which follow easily from the definition of the crossing.
\begin{center}\label{pull-over}
    \begin{tikzpicture}[scale=.75]
        \begin{scope}[ultra thick]
        \node() at (2.5,1.25) {=};
        \node() at (10.5,1.25) {=};
            \draw (0,0.5) rectangle (2, 2);

            \aline{0}{1.75}{1.25}{0.75}{blue}
            \aline{1.25}{0.75}{2}{0.75}{blue}

            \sline{0}{1.25}{1.25}{1.75}{red}
            \sline{0}{0.75}{1.25}{1.25}{red}
            \draw[color=red] (1.25,1.75) arc (90:-90:.25);

            \begin{scope}[color=blue]
                \draw[fill=blue] (0,1.75) circle (.1cm);
                \draw[fill=blue] (2,0.75) circle (.1cm);
            \end{scope}

            \begin{scope}[color=red]
                \draw[fill=red] (0,1.25) circle (.1cm);
                \draw[fill=red] (0,0.75) circle (.1cm);
            \end{scope}
            \draw (3,0.5) rectangle (5, 2);

            \aline{3}{1.75}{5}{0.75}{blue}

            \sline{3}{1.25}{3.25}{1.25}{red}
            \sline{3}{0.75}{3.25}{0.75}{red}
            \draw[color=red] (3.25,1.25) arc (90:-90:.25);

            \begin{scope}[color=blue]
                \draw[fill=blue] (3,1.75) circle (.1cm);
                \draw[fill=blue] (5,0.75) circle (.1cm);
            \end{scope}

            \begin{scope}[color=red]
                \draw[fill=red] (3,1.25) circle (.1cm);
                \draw[fill=red] (3,0.75) circle (.1cm);
            \end{scope}
            \draw (8,0.5) rectangle (10,2);
            
            \aline{8}{1.75}{9.25}{1.25}{blue}
            \aline{8}{1.25}{9.25}{0.75}{blue}
            \draw[color=blue] (9.25,1.25) arc (90:-90:.25);
            
            \sline{8}{0.75}{9.25}{1.75}{red}
            \sline{9.25}{1.75}{10}{1.75}{red}
            
            \begin{scope}[color=blue]
                \draw[fill=blue] (8,1.75) circle (.1cm);
                \draw[fill=blue] (8,1.25) circle (.1cm);
            \end{scope}

            \begin{scope}[color=red]
                \draw[fill=red] (8,0.75) circle (.1cm);
                \draw[fill=red] (10,1.75) circle (.1cm);
            \end{scope}            
            \draw (11,0.5) rectangle (13,2);
            
            \aline{11}{1.75}{11.25}{1.75}{blue}
            \aline{11}{1.25}{11.25}{1.25}{blue}
            \draw[color=blue] (11.25,1.75) arc (90:-90:.25);    
            
            \aline{11}{0.75}{13}{1.75}{red}
       
            \begin{scope}[color=blue]
                \draw[fill=blue] (11,1.75) circle (.1cm);
                \draw[fill=blue] (11,1.25) circle (.1cm);
            \end{scope}

            \begin{scope}[color=red]
                \draw[fill=red] (11,0.75) circle (.1cm);
                \draw[fill=red] (13,1.75) circle (.1cm);
            \end{scope}   
        \end{scope}
    \end{tikzpicture}
\end{center}
Finally, the unitary braiding $\sigma^{\mathrm{TL}}$ is given by the unitary 
isomorphisms $\sigma^{\mathrm{TL}}_{\oplus_iP_i,\oplus_jQ_j}$ in 
$\mathrm{Hom}\big((\oplus_i P_i)\otimes(\oplus_j Q_j), (\oplus_j Q_j)\otimes 
(\oplus_i P_i)\big)$ defined as direct sums of the 
$\sigma^{\mathrm{TL}}_{P_i,Q_j}$. 

\section{On a C*-algebra $\mathcal{B}$ and a *-homomorphism $\mathcal{B}\otimes\mathcal{B}\to\mathcal{B}$}\label{sec:BandPhi}

In this section, we define a Hilbert space $H$, a C*-algebra $\mathcal{B}\subset \mathbb{B}(H)$ and a *-homomorphism $\Phi\colon \mathcal{B}\otimes\mathcal{B}\to\mathcal{B}$, drawing inspiration from \cite{BS}, \cite{HP1}, \cite{Yuan} and \cite{EW}. Our starting point is the braided C*-tensor category $\mathcal{C} = \TLJcat(\delta)$ with $\delta\in \{2\cos(\pi/(k+2))\,:\,k=1,2,\ldots\}\cup\{2\}$, its tensor unit $\mathbbm{1}$, the generating object $\pi$, and a set $\mathcal{S}$ of simple objects in $\mathcal{C}$ chosen as in section \ref{sec:semisimple}. Put $\mathcal{G} = \{\mathbbm{1},\pi\}$ and denote by $\mathcal{G}^\infty$ the set of infinite sequences $\vec\mu = (\mu_1,\mu_2,\ldots)$ of elements in $\mathcal{G}$ for which there exists $n = n_{\vec\mu}\geq 0$ such that $\mu_k = \mathbbm{1}$ for $k>n$. Given such a sequence $\vec\mu$, we put $o(\vec\mu) = \mu_1\otimes\mu_2\otimes\cdots$. As $\mathcal{C}$ is a strict C*-tensor category, this infinite tensor product makes sense.

\subsection{Definition of $\mathcal{B}$}\label{sec:B}

For each $s\in\mathcal{S}$ and $\vec\mu\in\mathcal{G}^\infty$, the morphism space $\mathrm{Hom}(s,o(\vec\mu))$ is equipped with the inner product $\langle\cdot,\cdot\rangle$ given by $\langle \xi,\eta\rangle\,\id_s = \xi^*\eta$. We denote by $H^s$ the orthogonal direct sum of the Hilbert spaces $\mathrm{Hom}(s,o(\vec\mu))$ as $\vec\mu$ varies through $\mathcal{G}^\infty$. In symbols,
$$
    H^s = \bigoplus_{\vec\mu\in\mathcal{G}^\infty} \mathrm{Hom}(s,o(\vec\mu)).
$$
Next, we put
$$
    H = \bigoplus_{s\in\mathcal{S}} H^s.
$$
Given $\vec x, \vec y\in \mathcal{G}^\infty$ and $a\in\mathrm{Hom}(o(\vec y),o(\vec x))$, define a linear operator
$
    L_{\vec x, \vec y}(a)\colon H\to H
$
by the formula
$$
    L_{\vec x, \vec y}(a)\xi = \delta_{\vec y,\vec \mu}(a\circ \xi)\in \mathrm{Hom}(s,o(\vec x))
$$
for $\xi\in\mathrm{Hom}(s,o(\vec\mu))$. It is a bounded operator whose adjoint operator is
$$
    \big(L_{\vec x, \vec y}(a)\big)^* = L_{\vec y,\vec x}(a^*).
$$
Moreover,
$$
    L_{\vec x, \vec y}(a)\circ L_{\vec v, \vec w}(b) = \delta_{\vec y, \vec v}L_{\vec x, \vec w}(a\circ b).
$$
In particular, $p_{\vec x} \df L_{\vec x,\vec x}(\id_{o(\vec x)})$\label{def:px} is a projection in $\mathbb{B}(H)$ for each $\vec x\in\mathcal{G}^\infty$. Clearly, $\id_H = \sum_{\vec x\in\mathcal{G}^\infty} p_{\vec x}$.

\begin{lem}\label{lem:iso}
    We have that $\|L_{\vec x, \vec y}(a)\| = \|a\|$ for all $a\in\mathrm{Hom}(o(\vec y),o(\vec x))$.
\end{lem}

\begin{proof}
    Consider the *-homomorphism $\phi\colon \mathrm{End}(o(\vec y))\to 
    \mathbb{B}(H)$ given by $\phi(a) = L_{\vec y,\vec y}(a)$. The 
    semisimplicity of $\mathcal{C}$ implies that $\phi$ is injective. Since 
    every injective *-homomor\-phism between C*-algebras is isometric, it 
    follows that
    $$
        \|L_{\vec x, \vec y}(a)\|^2 = \|L_{\vec x, \vec y}(a)^*L_{\vec x, \vec 
        y}(a)\| = \|L_{\vec y,\vec y}(a^*a)\| = \|\phi(a^*a)\| = \|a^*a\| = 
        \|a\|^2
    $$
    for all $a\in\mathrm{Hom}(o(\vec y),o(\vec x))$.
\end{proof}

For each $n\geq 0$, denote by $\mathcal{B}_n$ the finite-dimensional C*-algebra spanned by the operators of the form $L_{\vec x, \vec y}(a)$, where $x_k = y_k = \mathbbm{1}$ for all $k>n$.
Each $\mathcal{B}_n$ admits a positive faithful trace $\mathrm{Tr}_n$ defined by
$
    \mathrm{Tr}_n(L_{\vec x, \vec y}(a)) = \delta_{\vec x,\vec y}\mathrm{Tr}_{k}^{\mathrm{TL}}(a),
$
where $k$ is the number of entries in $\vec x$ that equal $\pi$. Moreover, $\mathcal{B}_n\subseteq \mathcal{B}_{n+1}$ for all $n$. Denote by $\mathcal{B}$ the smallest C*-subalgebra of $\mathbb{B}(H)$ that contains every $\mathcal{B}_n$, i.e.,
$$
    \mathcal{B} = \overline{\bigcup_{n\geq 0}\mathcal{B}_n}.
$$
The following result describes the structure of $\mathcal{B}$.

\begin{lem}\label{lem:sum}
    We have that
    $$
        \mathcal{B} \cong \bigoplus_{s\in\mathcal{S}}\mathbb{K}(H^s).
    $$
\end{lem}

\begin{proof}
    Note first that $\mathcal{B}\subset\mathbb{K}(H)$. Indeed, each operator $L_{\vec x, \vec y}(a)$ is compact because it can be written as $L_{\vec x, \vec y}(a) P$, where $P$ is the orthogonal projection onto the finite-dimensional subspace $\bigoplus_{s\prec o(\vec y)}\mathrm{Hom}(s,o(\vec y))$. Conversely, if $\xi$ is a unit vector in $\mathrm{Hom}(s,o(\vec y))$ and $\eta$ is a unit vector in $\mathrm{Hom}(s,o(\vec x))$ then $\mathcal{B}$ contains the rank one operator $L_{\vec x,\vec y}(\eta\xi^*)\in\mathbb{K}(H^s)$, which maps $\xi$ onto $\eta$. Thus, for each $s\in\mathcal{S}$, $\mathcal{B}$ contains a complete set of matrix units for $\mathbb{K}(H^s)$. The result follows.
\end{proof}

The next lemma will be used to define certain morphisms between tensor products of $\mathcal{B}$-modules.

\begin{lem}\label{lem:idealize}
     Assume that $\sum_{n\geq 0}v_n$ and $\sum_{n\geq 0}v_n^*$ converge strongly in $\mathbb{B}(H)$, where $v_n\in\mathcal{B}$ for all $n$. Put $v = \sum_{n\geq 0}v_n$. Then $vb\in\mathcal{B}$ and $bv\in\mathcal{B}$ for all $b\in\mathcal{B}$.
\end{lem}

\begin{proof}
Note that $v^* = \sum_{n\geq 0}v_n^*$, where the sum converges in the strong operator topology.
Let $b\in\mathcal{B}$ be given.
    By Fact \ref{fact:SOTcpt}, $\sum_{n\geq 0}v_nb$ converges to $vb$ in norm because $b\in\mathbb{K}(H)$. Similarly, $\sum_{n\geq 0}v_n^*b^*$ converges to $v^*b^*$ in norm. Since $\mathcal{B}$ is a C*-subalgebra of $\mathbb{B}(H)$, the lemma follows.
\end{proof}

\subsection{Diagrammatic operators}\label{sec:diagram}

In effect, the above construction allows us to associate operators to certain 
kinds of diagrams. These diagrams all consist of strands inside a rectangle 
with an infinite sequence of nodes, some empty and some non-empty (or 
filled-in), attached to each of its (left and right) sides such that every 
strand connects two distinct non-empty nodes and every non-empty node is the 
end point of a unique strand. The simplest such diagram is a dilute 
Temperley--Lieb diagram (cf.\ e.g.\ \cite{BS}). It has only finitely many 
non-empty nodes, which are connected by non-crossing strands. The top of such a 
diagram is depicted below.
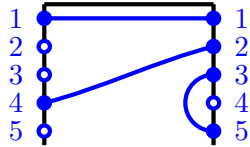
\begin{center}
	\begin{tikzpicture} [scale=.75]
		\begin{scope}[ultra thick]
			\draw (0,0) -- (0,2.5);
			\draw (0,2.5) -- (3,2.5);
			\draw (3,0) -- (3,2.5);
						
			\begin{scope}[color=blue]
			\draw (0,2.25) -- (3,2.25);
			\draw (0,.75) .. controls (1,1) and (2,1.5) .. (3,1.75);
			\draw (3,1.25) arc (90:270:.5);
			
			\draw[fill=blue] (0,2.25) circle (.1cm);
			\draw[fill=white] (0,1.75) circle (.1cm);
			\draw[fill=white] (0,1.25) circle (.1cm);
			\draw[fill=blue] (0,0.75) circle (.1cm);
			\draw[fill=white] (0,0.25) circle (.1cm);

            \node() at (-0.5,2.25) {$1$};
            \node() at (-0.5,1.75) {$2$};
			\node() at (-0.5,1.25) {$3$};
			\node() at (-0.5,0.75) {$4$};
			\node() at (-0.5,0.25) {$5$};

			\draw[fill=blue] (3,2.25) circle (.1cm);
			\draw[fill=blue] (3,1.75) circle (.1cm);
			\draw[fill=blue] (3,1.25) circle (.1cm);
			\draw[fill=white] (3,0.75) circle (.1cm);
			\draw[fill=blue] (3,0.25) circle (.1cm);
			
            \node() at (3.5,2.25) {$1$};
            \node() at (3.5,1.75) {$2$};
			\node() at (3.5,1.25) {$3$};
			\node() at (3.5,0.75) {$4$};
			\node() at (3.5,0.25) {$5$};
			\end{scope}
		\end{scope}
	\end{tikzpicture}
\end{center}
The diagram in the figure gives rise to the operator $L_{\vec x,\vec y}(a)$, 
where $\vec x = (\pi,\mathbbm{1},\mathbbm{1},\pi,\mathbbm{1},\ldots)$, $\vec y 
= (\pi,\pi,\pi,\mathbbm{1},\pi,\ldots)$, and $a$ is the morphism given by the 
pictured Temperley--Lieb diagram. By definition, the C*-algebra $\mathcal{B}$ 
is generated by operators arising from dilute Temperley--Lieb diagrams.

The following figure illustrates the product of two diagrammatic operators. Note that the patterns of empty and non-empty nodes have to match in the middle for the product to be non-zero.
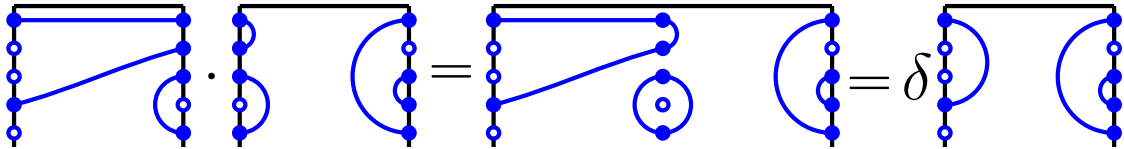
\begin{center}
	\begin{tikzpicture} [scale=.75]
		\begin{scope}[ultra thick]
			\node() at (3.5,1.25) {\tiny $\bullet$};
			\node() at (15.5,1.25) {\Huge $=\!\delta$};
            \node() at (7.75,1.25) {\Huge $=$};
			
			\draw (0,0) -- (0,2.5);
			\draw (0,2.5) -- (3,2.5);
			\draw (3,0) -- (3,2.5);
			
			\draw (4,0) -- (4,2.5);
			\draw (4,2.5) -- (7,2.5);
			\draw (7,0) -- (7,2.5);

            \draw (8.5,0) -- (8.5,2.5);
            \draw (8.5,2.5) -- (14.5,2.5);
            \draw (14.5,0) -- (14.5,2.5);
			
			\draw (16.5,0) -- (16.5,2.5);
			\draw (16.5,2.5) -- (19.5,2.5);
			\draw (19.5,0) -- (19.5,2.5);
			
			\begin{scope}[color=blue]
			\draw (0,2.25) -- (3,2.25);
			\draw (0,.75) .. controls (1,1) and (2,1.5) .. (3,1.75);
			\draw (3,1.25) arc (90:270:.5);
			
			\draw (4,2.25) arc (90:-90:.25);
			\draw (4,1.25) arc (90:-90:.5);
			\draw (7,2.25) arc (90:270:1);
			\draw (7,1.25) arc (90:270:.25);

			\draw (8.5,2.25) -- (11.5,2.25);
			\draw (8.5,.75) .. controls (9.5,1) and (10.5,1.5) .. (11.5,1.75);
			\draw (11.5,1.25) arc (90:270:.5);
			
			\draw (11.5,2.25) arc (90:-90:.25);
			\draw (11.5,1.25) arc (90:-90:.5);
			\draw (14.5,2.25) arc (90:270:1);
			\draw (14.5,1.25) arc (90:270:.25);
			
			\draw (16.5,2.25) arc (90:-90:.75);
			\draw (19.5,2.25) arc (90:270:1);
			\draw (19.5,1.25) arc (90:270:.25);
			
			\draw[fill=blue] (0,2.25) circle (.1cm);
			\draw[fill=white] (0,1.75) circle (.1cm);
			\draw[fill=white] (0,1.25) circle (.1cm);
			\draw[fill=blue] (0,0.75) circle (.1cm);
			\draw[fill=white] (0,0.25) circle (.1cm);
			
			\draw[fill=blue] (3,2.25) circle (.1cm);
			\draw[fill=blue] (3,1.75) circle (.1cm);
			\draw[fill=blue] (3,1.25) circle (.1cm);
			\draw[fill=white] (3,0.75) circle (.1cm);
			\draw[fill=blue] (3,0.25) circle (.1cm);
			
			\draw[fill=blue] (4,2.25) circle (.1cm);
			\draw[fill=blue] (4,1.75) circle (.1cm);
			\draw[fill=blue] (4,1.25) circle (.1cm);
			\draw[fill=white] (4,0.75) circle (.1cm);
			\draw[fill=blue] (4,0.25) circle (.1cm);
			
			\draw[fill=blue] (7,2.25) circle (.1cm);
			\draw[fill=white] (7,1.75) circle (.1cm);
			\draw[fill=blue] (7,1.25) circle (.1cm);
			\draw[fill=blue] (7,0.75) circle (.1cm);
			\draw[fill=blue] (7,0.25) circle (.1cm);

			\draw[fill=blue] (8.5,2.25) circle (.1cm);
			\draw[fill=white] (8.5,1.75) circle (.1cm);
			\draw[fill=white] (8.5,1.25) circle (.1cm);
			\draw[fill=blue] (8.5,0.75) circle (.1cm);
			\draw[fill=white] (8.5,0.25) circle (.1cm);

			\draw[fill=blue] (11.5,2.25) circle (.1cm);
			\draw[fill=blue] (11.5,1.75) circle (.1cm);
			\draw[fill=blue] (11.5,1.25) circle (.1cm);
			\draw[fill=white] (11.5,0.75) circle (.1cm);
			\draw[fill=blue] (11.5,0.25) circle (.1cm);

			\draw[fill=blue] (14.5,2.25) circle (.1cm);
			\draw[fill=white] (14.5,1.75) circle (.1cm);
			\draw[fill=blue] (14.5,1.25) circle (.1cm);
			\draw[fill=blue] (14.5,0.75) circle (.1cm);
			\draw[fill=blue] (14.5,0.25) circle (.1cm);
			
			\draw[fill=blue] (16.5,2.25) circle (.1cm);
			\draw[fill=white] (16.5,1.75) circle (.1cm);
			\draw[fill=white] (16.5,1.25) circle (.1cm);
			\draw[fill=blue] (16.5,0.75) circle (.1cm);
			\draw[fill=white] (16.5,0.25) circle (.1cm);
			
			\draw[fill=blue] (19.5,2.25) circle (.1cm);
			\draw[fill=white] (19.5,1.75) circle (.1cm);
			\draw[fill=blue] (19.5,1.25) circle (.1cm);
			\draw[fill=blue] (19.5,0.75) circle (.1cm);
			\draw[fill=blue] (19.5,0.25) circle (.1cm);
			\end{scope}
		\end{scope}
	\end{tikzpicture}
\end{center}

The unitary braiding on $\mathcal{C}$ allows us to also associate operators to certain diagrams that involve crossings. For instance, we can associate operators to what one might call ``finite dilute braid diagrams''. Such a diagram has only finitely many non-empty nodes (which is what the term ``finite'' in the name of the diagrams refers to). Moreover, every strand connects a node on the left side to one on the right side, and any two given  strands are only allowed to cross a finite number of times. The top of a sample diagram of this type is shown below.
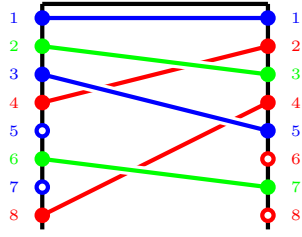
\begin{center}\label{fig:ass}
    \begin{tikzpicture}[scale=.75]
        \begin{scope}[ultra thick]
            \draw (0,0) -- (0,4);
            \draw (0,4) -- (4,4);
            \draw (4,0) -- (4,4);

            \aline{0}{0.25}{4}{2.25}{red}
            \aline{0}{2.25}{4}{3.25}{red}
			
            \sline{0}{1.25}{4}{0.75}{green}
            \sline{0}{3.25}{4}{2.75}{green}

            \sline{0}{3.75}{4}{3.75}{blue}
            \sline{0}{2.75}{4}{1.75}{blue}

			\begin{scope}[color=red]
                \draw[fill=red] (0,0.25) circle (.1cm);
                \node() at (-0.5,0.25) {\tiny $8$};
                \draw[fill=red] (0,2.25) circle (.1cm);
                \node() at (-0.5,2.25) {\tiny $4$};

                \draw[fill=red] (4,3.25) circle (.1cm);
                \node() at (4.5,3.25) {\tiny $2$};
                \draw[fill=red] (4,2.25) circle (.1cm);
                \node() at (4.5,2.25) {\tiny $4$};
                \draw[fill=white] (4,1.25) circle (.1cm);
                \node() at (4.5,1.25) {\tiny $6$};
                \draw[fill=white] (4,0.25) circle (.1cm);
                \node() at (4.5,0.25) {\tiny $8$};
            \end{scope}

            \begin{scope}[color=blue]
            	\draw[fill=blue] (0,3.75) circle (.1cm);
                \node() at (-0.5,3.75) {\tiny $1$};
            	\draw[fill=blue] (0,2.75) circle (.1cm);
                \node() at (-0.5,2.75) {\tiny $3$};
            	\draw[fill=white] (0,1.75) circle (.1cm);
                \node() at (-0.5,1.75) {\tiny $5$};
            	\draw[fill=white] (0,0.75) circle (.1cm);
                \node() at (-0.5,0.75) {\tiny $7$};
            	
            	\draw[fill=blue] (4,3.75) circle (.1cm);
                \node() at (4.5,3.75) {\tiny $1$};
            	\draw[fill=blue] (4,1.75) circle (.1cm);
                \node() at (4.5,1.75) {\tiny $5$};
            \end{scope}

            \begin{scope}[color=green]
            	\draw[fill=green] (0,3.25) circle (.1cm);
                \node() at (-0.5,3.25) {\tiny $2$};
            	\draw[fill=green] (0,1.25) circle (.1cm);
                \node() at (-0.5,1.25) {\tiny $6$};
            	
            	\draw[fill=green] (4,2.75) circle (.1cm);
                \node() at (4.5,2.75) {\tiny $3$};
            	\draw[fill=green] (4,0.75) circle (.1cm);
                \node() at (4.5,0.75) {\tiny $7$};
            \end{scope}
        \end{scope}
    \end{tikzpicture}
\end{center}
If one such diagram can be obtained from another by a finite sequence of Reidemeister moves of types 2 and 3 then these two diagrams give rise to the same operator. Indeed, the unitary braiding engenders, in a natural way, a group homomorphism from Artin's braid group on $n$ strands into the unitary group of $\mathrm{End}(\pi^{\otimes n})$ for every $n$ (see e.g.\ page 374 in \cite{EW}).
In particular, every finite dilute braid diagram gives rise to a partial isometry in $\mathcal{B}$.

We will also in a slightly different way associate operators to what might be termed ``(possibly) infinite dilute braid diagrams''. These diagrams are defined in the same way as their finite cousins, except that they are allowed to have infinitely many non-empty nodes and hence infinitely many strands. Let $D$ be such a diagram and denote by $\ell(D)$ the pattern of empty and non-empty nodes on its left side. Denote by $\supp(D)$ the set of patterns that can be obtained from $\ell(D)$ by replacing all but finitely many non-empty nodes by empty ones. Given $\vec x\in \supp(D)$, we get a finite dilute braid diagram $D_{\vec x}$\label{def:Dx} by removing from $D$ every strand whose left end point corresponds to an empty node in $\vec x$ and replacing both end points of each removed strand by empty nodes. As mentioned above, this new diagram gives rise to a partial isometry in $\mathcal{B}$, which we denote by $v(D,\vec x)$\label{def:vDx}. Since
$$
    v(D,\vec x)^* v(D,\vec y) = v(D,\vec x)v(D,\vec y)^* = 0
$$
whenever $\vec x\neq \vec y$, Fact \ref{lem:SOT} implies that $\sum_{\vec x\in\supp(D)} v(D,\vec x)$ is strongly convergent in $\mathbb{B}(H)$. We put
$$
    v(D) = \sum_{\vec x\in\supp(D)} v(D,\vec x).
$$
Although $v(D)$ need not belong to $\mathcal{B}$, Fact \ref{lem:SOT} and Lemma \ref{lem:idealize} imply that
$$
    v(D)\cdot b\in\mathcal{B},\qquad b\cdot v(D)\in\mathcal{B}
$$
for all $b\in\mathcal{B}$. If $D$ has no empty nodes then $v(D)$ is a unitary operator in $\mathbb{B}(H)$. This follows from the fact that multiplication in $\mathbb{B}(H)$ is jointly strongly continuous on bounded sets. In general, $v(D)$ is a partial isometry in $\mathbb{B}(H)$ whose range projection is $\sum_{\vec x\in\supp(D)}p_{\vec x}$. (Recall that $p_{\vec x}$ was defined on page \pageref{def:px}.)

\subsection{Definition of $\Phi\colon \mathcal{B}\otimes\mathcal{B}\to\mathcal{B}$}\label{sec:Phi}

Define, for each $n\geq 0$, a unitary element $U_n\in\mathcal{B}_{2n}$ in terms of the unitary braiding $\sigma^{\mathrm{TL}}$ on $\mathcal{C}$ in the same way as on page 374 in \cite{EW} (when $s$ there is $2$), except that we sum over all patterns $\vec x\in\mathcal{G}^{2n}$. As an example, the following figure shows two of the terms in the definition of $U_3$.
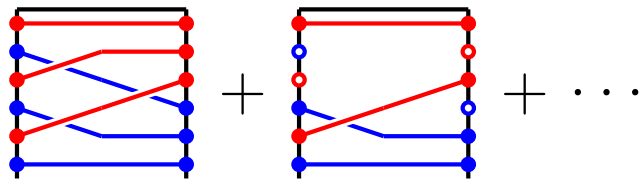
\begin{center}
    \begin{tikzpicture}[scale = .75]
        \begin{scope}[ultra thick]
            \node() at (4,1.5) {\Huge +};
            \node() at (9,1.5) {\Huge +};
            \node() at (10.55,1.5) {\Huge $\cdots$};

            \draw (0,0) -- (0,3);
            \draw (0,3) -- (3,3);
            \draw (3,0) -- (3,3);

            \draw (5,0) -- (5,3);
            \draw (5,3) -- (8,3);
            \draw (8,0) -- (8,3);

            \aline{0}{0.25}{3}{0.25}{blue}
            \aline{1.5}{0.75}{3}{0.75}{blue}
            \aline{0}{1.25}{1.5}{0.75}{blue}
            \aline{0}{2.25}{1.5}{1.75}{blue}
            \aline{1.5}{1.75}{3}{1.25}{blue}
            \aline{5}{0.25}{8}{0.25}{blue}
            \aline{6.5}{0.75}{8}{0.75}{blue}
            \aline{5}{1.25}{6.5}{0.75}{blue}

            \sline{0}{2.75}{3}{2.75}{red}
            \sline{1.5}{2.25}{3}{2.25}{red}
            \sline{0}{1.75}{1.5}{2.25}{red}
            \sline{0}{0.75}{1.5}{1.25}{red}
            \sline{1.5}{1.25}{3}{1.75}{red}
            \sline{5}{2.75}{8}{2.75}{red}
            \sline{5}{0.75}{6.5}{1.25}{red}
            \sline{6.5}{1.25}{8}{1.75}{red}

            \begin{scope}[color=blue]
                \draw[fill=blue] (0,0.25) circle (.1cm);
                \draw[fill=blue] (0,1.25) circle (.1cm);
                \draw[fill=blue] (0,2.25) circle (.1cm);
                \draw[fill=blue] (3,0.25) circle (.1cm);
                \draw[fill=blue] (3,0.75) circle (.1cm);
                \draw[fill=blue] (3,1.25) circle (.1cm);

                \draw[fill=blue] (5,0.25) circle (.1cm);
                \draw[fill=blue] (5,1.25) circle (.1cm);
                \draw[fill=white] (5,2.25) circle (.1cm);
                \draw[fill=blue] (8,0.25) circle (.1cm);
                \draw[fill=blue] (8,0.75) circle (.1cm);
                \draw[fill=white] (8,1.25) circle (.1cm);
            \end{scope}

            \begin{scope}[color=red]
                \draw[fill=red] (0,0.75) circle (.1cm);
                \draw[fill=red] (0,1.75) circle (.1cm);
                \draw[fill=red] (0,2.75) circle (.1cm);
                \draw[fill=red] (3,1.75) circle (.1cm);
                \draw[fill=red] (3,2.25) circle (.1cm);
                \draw[fill=red] (3,2.75) circle (.1cm);

                \draw[fill=red] (5,0.75) circle (.1cm);
                \draw[fill=white] (5,1.75) circle (.1cm);
                \draw[fill=red] (5,2.75) circle (.1cm);
                \draw[fill=red] (8,1.75) circle (.1cm);
                \draw[fill=white] (8,2.25) circle (.1cm);
                \draw[fill=red] (8,2.75) circle (.1cm);
            \end{scope}
        \end{scope}
	\end{tikzpicture}	
\end{center}
We can think of $U_3$ as $v(D)$, where $D$ is the diagram on the left, all nodes below the displayed part of the diagram being empty. However, in this case it is just a finite sum.

We can now define a *-homomorphism $\Phi_n\colon\mathcal{B}_n\otimes\mathcal{B}_n\to \mathcal{B}_{2n}$ by
$$
    L_{\vec x, \vec y}(a)\otimes L_{\vec v, \vec w}(b)\longmapsto U_n\circ L_{\vec x\vec v, \vec y\vec w}(a\otimes b)\circ U_n^*,
$$
where $\vec x\vec v = (x_1,\ldots,x_n,v_1,\ldots,v_n,\ldots)$ (and similarly for $\vec y\vec w$). The faithfulness of the traces $\mathrm{Tr}_n$ and the fact that $\mathrm{Tr}_{2n}\circ \Phi_n = \mathrm{Tr}_n\otimes\mathrm{Tr}_n$ on elements of the form $L_{\vec x, \vec y}(a)\otimes L_{\vec v, \vec w}(b)$ imply that $\Phi_n$ is a well-defined isometric *-homomorphism.
The purpose of the unitaries $U_n$ is to ensure that
$$
    \Phi_{n+1}\circ (\iota_n\otimes\iota_n) = \iota_{2n+1}\circ\iota_{2n}\circ \Phi_n
$$
for all $n\geq 0$, where $\iota_n$ is the inclusion map $\mathcal{B}_n\to\mathcal{B}_{n+1}$.
This allows us to extend the *-homo\-morphisms $\Phi_n$ to an isometric *-homomorphism
$$
    \Phi\colon\mathcal{B}\otimes\mathcal{B}\longrightarrow\mathcal{B}.
$$
Diagrammatically, the effect of applying $\Phi$ to a tensor product $L_{\vec x, 
\vec y}(a)\otimes L_{\vec v, \vec w}(b)$ of operators arising from dilute 
Temperley--Lieb or braid diagrams
is to superimpose the one on the left on top of the one on the right in such a way that the nodes are interleaved.

\begin{rem}\label{rem:Kring}
By Lemma \ref{lem:sum}, $K_0(\mathcal{B})$ is isomorphic to the fusion ring $\mathbb{Z}[\mathcal{S}]$ as an abelian group. It is also easy to check that the induced product map
$$
    K_0(\Phi)\colon K_0(\mathcal{B})\otimes_{\mathbb{Z}}K_0(\mathcal{B})\longrightarrow K_0(\mathcal{B})
$$
on $K_0(\mathcal{B})$ agrees with the product on the fusion ring. (This boils down to the fact that $L_{\vec\mu,\vec\mu}(vv^*)$ is a rank one projection in $\mathbb{K}(H^s)$ for any $\vec\mu\in\mathcal{G}^\infty$ and any unit vector $v\in\mathrm{Hom}(s,o(\vec\mu))$.) Below, we will ``categorify'' this statement, by using $\Phi$ to define a tensor product of right Hilbert $\mathcal{B}$-modules that recaptures the tensor product in $\mathcal{C}$ (see also Remark \ref{rem:categorification}).
\end{rem}

\section{On the braided C*-tensor categories $\mathrm{Mod}_{\mathcal{B}}$ and $\mathrm{Mod}_{\mathcal{B}}^f$}\label{sec:Monoidal}

In this section, we use the *-homomorphism $\Phi$ from the previous section to endow the category $\ModB$ of (right) Hilbert $\mathcal{B}$-modules with the structure of a braided C*-tensor category. We also introduce the full subcategory $\ModB^f$ of modules admitting a finite orthonormal basis.

\subsection{A tensor product of right Hilbert $\mathcal{B}$-modules}\label{sec:Tensor}

Given two right Hilbert $\mathcal{B}$-modules $M_1$ and $M_2$, we define their tensor product by
$$
    M_1\otimes M_2 = (M_1\boxtimes M_2)\otimes_{\Phi}\mathcal{B},
$$
where $\Phi\colon\mathcal{B}\otimes\mathcal{B}\to\mathcal{B}$ is the *-homomorphism from the previous section. (See section \ref{sec:HilbMod} for an explanation of the notation.) Given adjointable maps $f_1\colon M_1\to N_1$ and $f_2\colon M_2\to N_2$ between right Hilbert $\mathcal{B}$-modules, we denote by $f_1\otimes f_2$ the adjointable map $M_1\otimes M_2\to N_1\otimes N_2$ given by
$$
    (f_1\otimes f_2)(\xi_1\otimes\xi_2\otimes b) = f_1(\xi_1)\otimes f_2(\xi_2)\otimes b.
$$
for $\xi_1\in M_1$, $\xi_2\in M_2$ and $b\in\mathcal{B}$.

As a simple example, let $p$ and $q$ be projections in $\mathcal{B}$. Then $p\mathcal{B}$ and $q\mathcal{B}$ are right Hilbert $\mathcal{B}$-modules (with inner product given by $(a,b)\mapsto a^*b$) and there exists a surjective $\mathcal{B}$-linear isometry $p\mathcal{B}\otimes q\mathcal{B}\to \Phi(p\otimes q)\mathcal{B}$ defined by $pa\otimes qb\otimes c\mapsto \Phi(pa\otimes qb)c$ for $a,b,c\in\mathcal{B}$.

We next relate the above tensor product to the standard direct sum of Hilbert $\mathcal{B}$-modules. Given finite families $(M_i)_{i\in I}$ and $(N_j)_{j\in J}$ of right Hilbert $\mathcal{B}$-modules, we have a surjective $\mathcal{B}$-linear isometry
$$
    \phi\colon (\oplus_i M_i)\otimes (\oplus_j N_j)\longrightarrow \oplus_{(i,j)}(M_i\otimes N_j)
$$
defined by $(\xi_i)_i\otimes (\eta_j)_j\otimes b\mapsto (\xi_i\otimes\eta_j\otimes b)_{(i,j)}$ for $\xi_i\in M_i$, $\eta_j\in N_j$ and $b\in\mathcal{B}$.

\subsection{The C*-tensor category $\mathrm{Mod}_{\mathcal{B}}$}\label{sec:ModB}

We denote by $\ModB$ the category whose objects are all right Hilbert $\mathcal{B}$-modules and whose morphism sets $\mathrm{Hom}(M,N)$ consist of all adjointable (or, equivalently, all bounded $\mathcal{B}$-linear, cf.\ \cite{F}) maps $M\to N$. Below, we will endow this category with the structure of a C*-tensor category. Note first that conditions (1), (2), (3), (6) and (7) in section \ref{sec:TensorCat} follow from the general theory of Hilbert C*-modules.
Thus, our goal in the present section is to
define associators, a tensor unit, and unit constraints satisfying conditions (4), (5) and (8).

\subsubsection{Associators}

We begin by defining associators in \ModB. To do so, we first define a unitary operator $V\in \mathbb{B}(H)$ as the operator associated to the following infinite braid diagram $D^\alpha$. (Note that, in notation introduced on page \pageref{def:Dx}, the multi-colored figure on page \pageref{fig:ass}
depicts $D^\alpha_{\vec x}$ ($=(D^\alpha)_{\vec x}$), where $\vec x = (\pi,\pi,\pi,\pi,\mathbbm{1},\pi,\mathbbm{1},\pi,\ldots)$.) First connect the nodes numbered $4$, $8$, $12$, $\ldots$ on the left side to those numbered $2$, $4$, $6$, $\ldots$ on the right side by strands in order. (These nodes and strands are colored red in the aforementioned figure.) Next connect, by (green) strands that cross over the ones already drawn, the nodes on the left side numbered $2$, $6$, $10$, $\ldots$ to those numbered $3$, $7$, $11$, $\ldots$ on the right side. Finally, connect, by (blue) strands that cross over the ones already drawn, the nodes on the left side numbered $1$, $3$, $5$, $\ldots$ to those numbered $1$, $5$, $9$, $\ldots$ on the right side.

We next observe that
\begin{equation}\label{eq:V}
    V\Phi(\Phi(b_1\otimes b_2)\otimes b_3)V^* = \Phi(b_1\otimes\Phi(b_2\otimes b_3))
\end{equation}
for all $b_1,b_2,b_3\in\mathcal{B}$. The following figures illustrate the case when $b_1,b_2,b_3\in\mathcal{B}_2$. In that case, the left hand side of equation (\ref{eq:V}) arises from the diagram
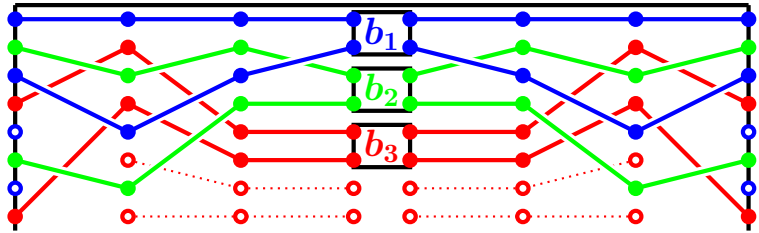
\begin{center}
    \begin{tikzpicture}[scale=0.75]
        \begin{scope}[ultra thick]
            \draw (0,0) -- (0,4);
            \draw (0,4) -- (13,4);
            \draw (13,4) -- (13,0);

            \draw (6,1.125) rectangle (7,1.875);
            \node() at (6.5,1.5) {\color{red}\Large $\boldsymbol{b_3}$};

            \aline{0}{2.25}{2}{3.25}{red}
            \aline{2}{3.25}{4}{1.75}{red}
            \aline{4}{1.75}{6}{1.75}{red}
            \aline{0}{0.25}{2}{2.25}{red}
            \aline{2}{2.25}{4}{1.25}{red}
            \aline{4}{1.25}{6}{1.25}{red}
            \draw[color=red,dotted,thick] (2,1.25) -- (4,0.75);
            \draw[color=red,dotted,thick] (4,0.75) -- (6,0.75);
            \draw[color=red,dotted,thick] (2,0.25) -- (6,0.25);

            \aline{13}{2.25}{11}{3.25}{red}
            \aline{11}{3.25}{9}{1.75}{red}
            \aline{9}{1.75}{7}{1.75}{red}
            \aline{13}{0.25}{11}{2.25}{red}
            \aline{11}{2.25}{9}{1.25}{red}
            \aline{9}{1.25}{7}{1.25}{red}
            \draw[color=red,dotted,thick] (11,1.25) -- (9,0.75);
            \draw[color=red,dotted,thick] (9,0.75) -- (7,0.75);
            \draw[color=red,dotted,thick] (11,0.25) -- (7,0.25);

            \draw (6,2.125) rectangle (7,2.875);
            \node() at (6.5,2.5) {\color{green}\Large $\boldsymbol{b_2}$};

            \sline{0}{3.25}{2}{2.75}{green}
            \sline{2}{2.75}{4}{3.25}{green}
            \sline{4}{3.25}{6}{2.75}{green}
            \sline{0}{1.25}{2}{0.75}{green}
            \sline{2}{0.75}{4}{2.25}{green}
            \sline{4}{2.25}{6}{2.25}{green}

            \sline{13}{3.25}{11}{2.75}{green}
            \sline{11}{2.75}{9}{3.25}{green}
            \sline{9}{3.25}{7}{2.75}{green}
            \sline{13}{1.25}{11}{0.75}{green}
            \sline{11}{0.75}{9}{2.25}{green}
            \sline{9}{2.25}{7}{2.25}{green}

            \draw (6,3.125) rectangle (7,3.875);
            \node() at (6.5,3.5) {\color{blue}\Large $\boldsymbol{b_1}$};

            \sline{0}{3.75}{6}{3.75}{blue}
            \sline{0}{2.75}{2}{1.75}{blue}
            \sline{2}{1.75}{4}{2.75}{blue}
            \sline{4}{2.75}{6}{3.25}{blue}

            \sline{7}{3.75}{13}{3.75}{blue}
            \sline{7}{3.25}{9}{2.75}{blue}
            \sline{9}{2.75}{11}{1.75}{blue}
            \sline{11}{1.75}{13}{2.75}{blue}

            \begin{scope}[color=red]
                \draw[fill=red] (0,2.25) circle (.1cm);
                \draw[fill=red] (0,0.25) circle (.1cm);

                \draw[fill=red] (2,3.25) circle (.1cm);
                \draw[fill=red] (2,2.25) circle (.1cm);

                \draw[fill=white] (2,1.25) circle (.1cm);
                \draw[fill=white] (2,0.25) circle (.1cm);

                \draw[fill=red] (4,1.75) circle (.1cm);
                \draw[fill=red] (4,1.25) circle (.1cm);

                \draw[fill=white] (4,0.75) circle (.1cm);
                \draw[fill=white] (4,0.25) circle (.1cm);

                \draw[fill=red] (6,1.75) circle (.1cm);
                \draw[fill=red] (6,1.25) circle (.1cm);

                \draw[fill=white] (6,0.75) circle (.1cm);
                \draw[fill=white] (6,0.25) circle (.1cm);

                \draw[fill=red] (13,2.25) circle (.1cm);
                \draw[fill=red] (13,0.25) circle (.1cm);

                \draw[fill=red] (11,3.25) circle (.1cm);
                \draw[fill=red] (11,2.25) circle (.1cm);

                \draw[fill=white] (11,1.25) circle (.1cm);
                \draw[fill=white] (11,0.25) circle (.1cm);

                \draw[fill=red] (9,1.75) circle (.1cm);
                \draw[fill=red] (9,1.25) circle (.1cm);

                \draw[fill=white] (9,0.75) circle (.1cm);
                \draw[fill=white] (9,0.25) circle (.1cm);

                \draw[fill=red] (7,1.75) circle (.1cm);
                \draw[fill=red] (7,1.25) circle (.1cm);

                \draw[fill=white] (7,0.75) circle (.1cm);
                \draw[fill=white] (7,0.25) circle (.1cm);
            \end{scope}

            \begin{scope}[color=green]
                \draw[fill=green] (0,3.25) circle (.1cm);
                \draw[fill=green] (0,1.25) circle (.1cm);

                \draw[fill=green] (2,2.75) circle (.1cm);
                \draw[fill=green] (2,0.75) circle (.1cm);

                \draw[fill=green] (4,3.25) circle (.1cm);
                \draw[fill=green] (4,2.25) circle (.1cm);

                \draw[fill=green] (6,2.75) circle (.1cm);
                \draw[fill=green] (6,2.25) circle (.1cm);

                \draw[fill=green] (7,2.75) circle (.1cm);
                \draw[fill=green] (7,2.25) circle (.1cm);

                \draw[fill=green] (11,2.75) circle (.1cm);
                \draw[fill=green] (11,0.75) circle (.1cm);

                \draw[fill=green] (9,3.25) circle (.1cm);
                \draw[fill=green] (9,2.25) circle (.1cm);

                \draw[fill=green] (13,3.25) circle (.1cm);
                \draw[fill=green] (13,1.25) circle (.1cm);
            \end{scope}

            \begin{scope}[color=blue]
                \draw[fill=blue] (0,3.75) circle (.1cm);
                \draw[fill=blue] (0,2.75) circle (.1cm);
                \draw[fill=white] (0,1.75) circle (.1cm);
                \draw[fill=white] (0,0.75) circle (.1cm);

                \draw[fill=blue] (2,3.75) circle (.1cm);
                \draw[fill=blue] (2,1.75) circle (.1cm);

                \draw[fill=blue] (4,3.75) circle (.1cm);
                \draw[fill=blue] (4,2.75) circle (.1cm);

                \draw[fill=blue] (6,3.75) circle (.1cm);
                \draw[fill=blue] (6,3.25) circle (.1cm);
                \draw[fill=blue] (7,3.75) circle (.1cm);
                \draw[fill=blue] (7,3.25) circle (.1cm);

                \draw[fill=blue] (11,3.75) circle (.1cm);
                \draw[fill=blue] (11,1.75) circle (.1cm);

                \draw[fill=blue] (9,3.75) circle (.1cm);
                \draw[fill=blue] (9,2.75) circle (.1cm);

                \draw[fill=blue] (13,3.75) circle (.1cm);
                \draw[fill=blue] (13,2.75) circle (.1cm);
                \draw[fill=white] (13,1.75) circle (.1cm);
                \draw[fill=white] (13,0.75) circle (.1cm);
            \end{scope}
        \end{scope}
    \end{tikzpicture}
\end{center}
Note that, in the above figure and the one below, the dotted lines do not represent strands, but only serve to keep track of the positions of empty nodes. Also, depending on which of the four nodes attached to each $b_i$ are empty and non-empty, the solid lines may or may not represent strands. For comparison, the right hand side of equation (\ref{eq:V}) arises from the diagram
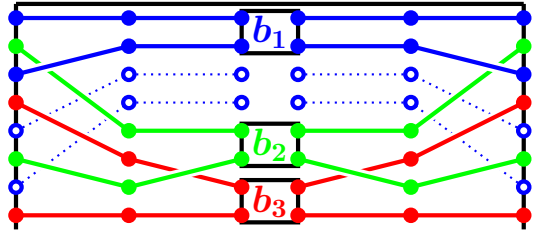
\begin{center}
    \begin{tikzpicture}[scale=0.75]
        \begin{scope}[ultra thick]
            \draw (2,0) -- (2,4);
            \draw (2,4) -- (11,4);
            \draw (11,4) -- (11,0);

            \draw[dotted,color=blue,thick] (2,1.75) -- (4,2.75);
            \draw[dotted,color=blue,thick] (4,2.75) -- (6,2.75);
            \draw[dotted,color=blue,thick] (2,0.75) -- (4,2.25);
            \draw[dotted,color=blue,thick] (4,2.25) -- (6,2.25);

            \draw[dotted,color=blue,thick] (11,1.75) -- (9,2.75);
            \draw[dotted,color=blue,thick] (9,2.75) -- (7,2.75);
            \draw[dotted,color=blue,thick] (11,0.75) -- (9,2.25);
            \draw[dotted,color=blue,thick] (9,2.25) -- (7,2.25);

            \draw (6,0.125) rectangle (7,0.875);
            \node() at (6.5,0.5) {\color{red}\Large $\boldsymbol{b_3}$};

            \sline{2}{2.25}{4}{1.25}{red}
            \sline{4}{1.25}{6}{0.75}{red}
            \sline{2}{0.25}{6}{0.25}{red}

            \sline{11}{2.25}{9}{1.25}{red}
            \sline{9}{1.25}{7}{0.75}{red}
            \sline{11}{0.25}{7}{0.25}{red}

            \draw (6,1.125) rectangle (7,1.875);
            \node() at (6.5,1.5) {\color{green}\Large $\boldsymbol{b_2}$};

            \sline{2}{3.25}{4}{1.75}{green}
            \sline{4}{1.75}{6}{1.75}{green}
            \sline{2}{1.25}{4}{0.75}{green}
            \sline{4}{0.75}{6}{1.25}{green}

            \sline{11}{3.25}{9}{1.75}{green}
            \sline{9}{1.75}{7}{1.75}{green}
            \sline{11}{1.25}{9}{0.75}{green}
            \sline{9}{0.75}{7}{1.25}{green}

            \draw (6,3.125) rectangle (7,3.875);
            \node() at (6.5,3.5) {\color{blue}\Large $\boldsymbol{b_1}$};

            \sline{2}{3.75}{6}{3.75}{blue}
            \sline{2}{2.75}{4}{3.25}{blue}
            \sline{4}{3.25}{6}{3.25}{blue}

            \sline{11}{3.75}{7}{3.75}{blue}
            \sline{11}{2.75}{9}{3.25}{blue}
            \sline{9}{3.25}{7}{3.25}{blue}

            \begin{scope}[color=red]
                \draw[fill=red] (2,2.25) circle (.1cm);
                \draw[fill=red] (2,0.25) circle (.1cm);

                \draw[fill=red] (4,1.25) circle (.1cm);
                \draw[fill=red] (4,0.25) circle (.1cm);

                \draw[fill=red] (6,0.75) circle (.1cm);
                \draw[fill=red] (6,0.25) circle (.1cm);

                \draw[fill=red] (9,1.25) circle (.1cm);
                \draw[fill=red] (9,0.25) circle (.1cm);

                \draw[fill=red] (7,0.75) circle (.1cm);
                \draw[fill=red] (7,0.25) circle (.1cm);

                \draw[fill=red] (11,2.25) circle (.1cm);
                \draw[fill=red] (11,0.25) circle (.1cm);
            \end{scope}

            \begin{scope}[color=green]
                \draw[fill=green] (2,3.25) circle (.1cm);
                \draw[fill=green] (2,1.25) circle (.1cm);

                \draw[fill=green] (4,1.75) circle (.1cm);
                \draw[fill=green] (4,0.75) circle (.1cm);

                \draw[fill=green] (6,1.75) circle (.1cm);
                \draw[fill=green] (6,1.25) circle (.1cm);

                \draw[fill=green] (9,1.75) circle (.1cm);
                \draw[fill=green] (9,0.75) circle (.1cm);

                \draw[fill=green] (7,1.75) circle (.1cm);
                \draw[fill=green] (7,1.25) circle (.1cm);

                \draw[fill=green] (11,3.25) circle (.1cm);
                \draw[fill=green] (11,1.25) circle (.1cm);
            \end{scope}

            \begin{scope}[color=blue]
                \draw[fill=blue] (2,3.75) circle (.1cm);
                \draw[fill=blue] (2,2.75) circle (.1cm);
                \draw[fill=white] (2,1.75) circle (.1cm);
                \draw[fill=white] (2,0.75) circle (.1cm);

                \draw[fill=blue] (4,3.75) circle (.1cm);
                \draw[fill=blue] (4,3.25) circle (.1cm);
                \draw[fill=white] (4,2.75) circle (.1cm);
                \draw[fill=white] (4,2.25) circle (.1cm);

                \draw[fill=blue] (6,3.75) circle (.1cm);
                \draw[fill=blue] (6,3.25) circle (.1cm);
                \draw[fill=white] (6,2.75) circle (.1cm);
                \draw[fill=white] (6,2.25) circle (.1cm);

                \draw[fill=blue] (9,3.75) circle (.1cm);
                \draw[fill=blue] (9,3.25) circle (.1cm);
                \draw[fill=white] (9,2.75) circle (.1cm);
                \draw[fill=white] (9,2.25) circle (.1cm);

                \draw[fill=blue] (7,3.75) circle (.1cm);
                \draw[fill=blue] (7,3.25) circle (.1cm);
                \draw[fill=white] (7,2.75) circle (.1cm);
                \draw[fill=white] (7,2.25) circle (.1cm);

                \draw[fill=blue] (11,3.75) circle (.1cm);
                \draw[fill=blue] (11,2.75) circle (.1cm);
                \draw[fill=white] (11,1.75) circle (.1cm);
                \draw[fill=white] (11,0.75) circle (.1cm);
            \end{scope}
        \end{scope}
    \end{tikzpicture}
\end{center}
In general, one of these diagrams can be obtained from the other by a finite sequence of Reidemeister moves of types 2 and 3. Thus, the associated operators are equal.

We can now define associators as follows. Given right Hilbert $\mathcal{B}$-modules $M_1$, $M_2$ and $M_3$, consider the formula
$$
    \alpha_{M_1,M_2,M_3}\big(\xi_1\otimes \xi_2a\otimes b\otimes\xi_3cd\otimes e\big) = \xi_1\otimes\xi_2\otimes\xi_3\otimes \Phi(a\otimes c)\otimes V\Phi(b\otimes d)e,
$$
where $\xi_1\in M_1$, $\xi_2\in M_2$, $\xi_3\in M_3$ and $a,b,c,d,e\in\mathcal{B}$. Here, $\xi_1\otimes \xi_2a\otimes b\otimes\xi_3cd\otimes e$ on the left hand side is viewed as an element of
$$
    (M_1\otimes M_2)\otimes M_3 = \bigg(\big((M_1\boxtimes M_2)\otimes_{\Phi}\mathcal{B}\big)\boxtimes M_3\bigg)\otimes_{\Phi}\mathcal{B}
$$
while $\xi_1\otimes\xi_2\otimes\xi_3\otimes \Phi(a\otimes c)\otimes V\Phi(b\otimes d)e$ on the right hand side is viewed as an element of
$$
    M_1\otimes (M_2\otimes M_3) = \bigg(M_1\boxtimes \big((M_2\boxtimes M_3)\otimes_{\Phi}\mathcal{B}\big)\bigg)\otimes_{\Phi}\mathcal{B}.
$$
On the one hand, we get that
\begin{align*}
    \langle [\xi_1\otimes &(\xi_2\otimes\xi_3\otimes\Phi(a\otimes c))]\otimes V\Phi(b\otimes d)e, [\eta_1\otimes(\eta_2\otimes\eta_3\otimes\Phi(a_1\otimes c_1))]\otimes V\Phi(b_1\otimes d_1)e_1\rangle\\
    &= [V\Phi(b\otimes d)e]^*\Phi\left(\langle\xi_1,\eta_1\rangle\otimes\langle (\xi_2\otimes\xi_3)\otimes \Phi(a\otimes c),(\eta_2\otimes\eta_3)\otimes \Phi(a_1\otimes c_1)\rangle \right)[V\Phi(b_1\otimes d_1)e_1]\\
    &= [V\Phi(b\otimes d)e]^*\Phi\left(\langle\xi_1,\eta_1\rangle\otimes \left[\Phi(a\otimes c)^*\Phi(\langle \xi_2,\eta_2\rangle \otimes \langle \xi_3,\eta_3\rangle)\Phi(a_1\otimes c_1)\right]\right)[V\Phi(b_1\otimes d_1)e_1]\\
    &= [V\Phi(b\otimes d)e]^*\Phi\left(\langle\xi_1,\eta_1\rangle\otimes \Phi\big(a^*\langle \xi_2,\eta_2\rangle a_1 \otimes c^*\langle \xi_3,\eta_3\rangle c_1\big)\right)[V\Phi(b_1\otimes d_1)e_1]\\
    &= [e^*\Phi(b\otimes d)^*]V^*\Phi\left(\langle\xi_1,\eta_1\rangle\otimes \Phi\big(\langle \xi_2 a,\eta_2 a_1\rangle \otimes \langle \xi_3 c,\eta_3 c_1\rangle\big)\right)V[\Phi(b_1\otimes d_1)e_1].
\end{align*}
On the other hand, we have that
\begin{align*}
    \langle [((\xi_1\,\otimes\, &\xi_2a)\otimes b)\otimes\xi_3cd]\otimes e, [((\eta_1\otimes\eta_2a_1)\otimes b_1)\otimes\eta_3c_1d_1]\otimes e_1\rangle\\
    &= e^*\Phi(\langle (\xi_1\otimes \xi_2a)\otimes b, (\eta_1\otimes\eta_2a_1)\otimes b_1\rangle\otimes \langle\xi_3cd,\eta_3c_1d_1\rangle)e_1\\
    &= e^*\Phi([b^*\Phi(\langle\xi_1,\eta_1\rangle\otimes \langle\xi_2a,\eta_2a_1\rangle) b_1]\otimes d^*\langle\xi_3c,\eta_3c_1\rangle d_1) e_1\\
    &= e^*\Phi(b^*\otimes d^*)\Phi\left(\Phi(\langle\xi_1,\eta_1\rangle\otimes \langle\xi_2a,\eta_2a_1\rangle)\otimes \langle\xi_3c,\eta_3c_1\rangle\right)\Phi(b_1\otimes d_1)e_1.
\end{align*}
Since these two expressions coincide by equation (\ref{eq:V}), the above formula defines a $\mathcal{B}$-linear isometry
$$
    \alpha_{M_1,M_2,M_3}\colon (M_1\otimes M_2)\otimes M_3\longrightarrow M_1\otimes (M_2\otimes M_3).
$$
Similarly, we can define a $\mathcal{B}$-linear isometry
$$
    \beta_{M_1,M_2,M_3}\colon M_1\otimes (M_2\otimes M_3) \longrightarrow (M_1\otimes M_2)\otimes M_3
$$
by the formula
$$
    \beta_{M_1,M_2,M_3}(\xi_1ab\otimes\xi_2c\otimes\xi_3\otimes d\otimes e) = \xi_1\otimes\xi_2\otimes\Phi(a\otimes c)\otimes \xi_3\otimes V^*\Phi(b\otimes d)e.
$$
As this is the inverse of $\alpha_{M_1,M_2,M_3}$, we get that $\alpha_{M_1,M_2,M_3}$ is a unitary isomorphism in $\ModB$. The assignment $(M_1,M_2,M_2)\mapsto \alpha_{M_1,M_2,M_3}$ is clearly natural in $M_1$, $M_2$ and $M_3$.

\subsubsection{Pentagon identity}

In order to show that $\ModB$ along with the associators $\alpha_{M_1,M_2,M_3}$ and the unit constraints that we define below is a C*-tensor category, we must verify the pentagon identity, which in the present context is the identity
\begin{equation}\label{eq:pentagon}
  \begin{aligned}
    (\id_{M_1}\otimes\alpha_{M_2,M_3,M_4})\circ \alpha_{M_1,M_2\otimes M_3,M_4}\circ(\alpha_{M_1,M_2,M_3}&\otimes\id_{M_4})\\ &= \alpha_{M_1,M_2,M_3\otimes M_4}\circ\alpha_{M_1\otimes M_2,M_3,M_4}
  \end{aligned}
\end{equation}
for any objects $M_1$, $M_2$, $M_3$ and $M_4$ in $\ModB$. We verify it by applying both sides to an element of the form
\begin{equation}\label{eq:pentagon1}
    \xi_1 a\otimes\xi_2bb'b''\otimes c\otimes\xi_3 d d'd''d'''\otimes e\otimes \xi_4 f_0ff'f''f'''\otimes g
\end{equation}
in the quadruple tensor product
$
    \big((M_1\otimes M_2)\otimes M_3\big)\otimes M_4.
$
Let us first consider the left hand side. First, $\alpha_{M_1,M_2,M_3}\otimes\id_{M_4}$ maps the given element to
$$
    \xi_1 a\otimes\big[\xi_2b\otimes\xi_3d\otimes\Phi(b'b''\otimes d'd'')\big]\otimes V\Phi(c\otimes d''')e\otimes\xi_4 f_0ff'f''f'''\otimes g.
$$
Next, $\alpha_{M_1,M_2\otimes M_3,M_4}$ maps the above element to
$$
    \xi_1 a\otimes \big[\xi_2 b\otimes\xi_3 d\otimes\Phi(b'\otimes d')\big]\otimes\xi_4f_0f\otimes\Phi(\Phi(b''\otimes d'')\otimes f')\otimes V\Phi(V\Phi(c\otimes d''')e\otimes f''f''')g.
$$
Finally, $\id_{M_1}\otimes\alpha_{M_2,M_3,M_4}$ maps this element to
$$
    \xi_1 a\otimes \left(\xi_2 b\otimes \big[\xi_3\otimes \xi_4\otimes\Phi(d\otimes f_0)\big]\otimes V\Phi(\Phi(b'b''\otimes d'd'')\otimes ff')\right)\otimes V\Phi(V\Phi(c\otimes d''')e\otimes f''f''')g.
$$
We now consider the right hand side. First, $\alpha_{M_1\otimes M_2,M_3,M_4}$ maps the element in equation (\ref{eq:pentagon1}) to
$$
    \xi_1 a\otimes \xi_2 bb'b''\otimes c\otimes \left[\xi_3\otimes \xi_4\otimes \Phi(dd'd''d'''\otimes f_0ff'f'')\right]\otimes V\Phi(e\otimes f''')g.
$$
Next, $\alpha_{M_1,M_2,M_3\otimes M_4}$ maps the above element to {\footnotesize
$$
    \xi_1 a\otimes \left(\xi_2bb'\otimes\left[\xi_3\otimes\xi_4\otimes\Phi(dd'\otimes f_0f)\right]\otimes\Phi(b''\otimes\Phi(d''\otimes f'))\right)\otimes V\Phi(c\otimes\Phi(d'''\otimes f''))V\Phi(e\otimes f''')g,
$$}which is equal to {\footnotesize
$$
    \xi_1 a\otimes \left(\xi_2b\otimes\left[\xi_3\otimes\xi_4\otimes\Phi(d\otimes f_0)\right]\otimes\Phi(b'b''\otimes\Phi(d'd''\otimes ff'))\right)\otimes V\Phi(c\otimes\Phi(d'''\otimes f''))V\Phi(e\otimes f''')g,
$$}and, in turn, to {\footnotesize
$$
    \xi_1 a\otimes \left(\xi_2b\otimes\left[\xi_3\otimes\xi_4\otimes\Phi(d\otimes f_0)\right]\otimes V\Phi(\Phi(b'b''\otimes d'd'')\otimes ff')V^*\right)\otimes V\Phi(c\otimes\Phi(d'''\otimes f''))V\Phi(e\otimes f''')g.
$$}We now see that the pentagon identity reduces to the identity {\footnotesize
$$
    \Phi(a\otimes \Phi(\Phi(b''\otimes d'')\otimes f'))V\Phi(V\Phi(c\otimes d''')\otimes f'')
    = \Phi(a\otimes \Phi(\Phi(b''\otimes d'')\otimes f')V^*)V^2\Phi(\Phi(c\otimes d''')\otimes f'').
$$
}Since $\mathcal{B}$ is generated by operators arising from dilute 
Temperley--Lieb diagrams, and because $V = \sum_{\vec x}v(D^\alpha,\vec x)$ for 
a certain infinite braid diagram $D^\alpha$ (see page \pageref{def:vDx}), it 
suffices to prove that
$$
    v(D^\alpha,\vec x)\Phi(v(D^\alpha,\vec y)\otimes p_{\vec z}) = \Phi(p_{\vec\mu}\otimes v(D^\alpha,\vec \nu)^*)v(D^\alpha,\vec \beta)v(D^\alpha,\vec \gamma)
$$
whenever $\vec x, \vec y, \vec z, \vec\mu, \vec\nu,\vec\beta,\vec\gamma\in\mathcal{G}^\infty$ are such that the patterns agree. (Recall that $p_{\vec z}$ was defined on page \pageref{def:px}.) In this identity, each side is the operator associated to some finite dilute braid diagram. One can easily check that both of these diagrams consist of strands that live on four separate layers, as we next explain. The bottom layer $L_1$ consists of those strands whose left end point is at one of the non-empty nodes numbered $4$, $8$, $12$, $\ldots$, the next layer $L_2$ at those numbered $6$, $14$, $22$, $\ldots$, the next layer $L_3$ at those numbered $2$, $10$, $18$, $\ldots$, and the top layer $L_4$ at those numbered $1$, $3$, $5$, $\ldots$. This means that, in both diagrams, every crossing is of the following sort: A strand from $L_j$ crosses over a strand from $L_i$ with $j>i$. It is easily deduced from this that one of the diagrams can be obtained from the other by a finite sequence of Reidemeister moves of types $2$ and $3$, from which the identity follows.

\subsubsection{Tensor unit and unit constraints}

Denote by $p_\ast$ the operator in $\mathcal{B}$ that is associated to the empty diagram. We will exhibit $p_\ast\mathcal{B}$ as a tensor unit in $\ModB$ by defining explicit unit constraints
$$
    \Psi_M^\ell\colon p_\ast\mathcal{B}\otimes M\longrightarrow M,\qquad \Psi^r_M\colon M\otimes p_\ast\mathcal{B}\longrightarrow M
$$
for each object $M$ in $\ModB$. First, we define two partial isometries $W^\ell$ and $W^r$ in $\mathbb{B}(H)$. Namely, $W^r$ is the operator associated to the infinite dilute braid diagram
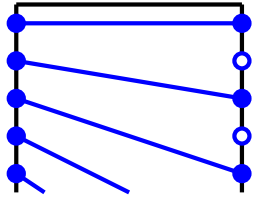
\begin{center}
    \begin{tikzpicture}
        \begin{scope}[ultra thick]
            \draw (0,0) -- (0,2.5);
            \draw (0,2.5) -- (3,2.5);
            \draw (3,0) -- (3,2.5);

            \begin{scope}[color=blue]
                \draw (0,2.25) -- (3,2.25);
                \draw (0,1.75) -- (3,1.25);
                \draw (0,1.25) -- (3,0.25);
                \draw (0,0.75) -- (1.5,0);
                \draw (0,0.25) -- (0.375,0);

			    \draw[fill=blue] (0,2.25) circle (.1cm);
			    \draw[fill=blue] (0,1.75) circle (.1cm);
			    \draw[fill=blue] (0,1.25) circle (.1cm);
			    \draw[fill=blue] (0,0.75) circle (.1cm);
			    \draw[fill=blue] (0,0.25) circle (.1cm);
			
			    \draw[fill=blue] (3,2.25) circle (.1cm);
			    \draw[fill=white] (3,1.75) circle (.1cm);
			    \draw[fill=blue] (3,1.25) circle (.1cm);
			    \draw[fill=white] (3,0.75) circle (.1cm);
			    \draw[fill=blue] (3,0.25) circle (.1cm);
            \end{scope}
        \end{scope}
    \end{tikzpicture}
\end{center}
which we will call $D^r$, while $W^\ell$ is the operator associated to the diagram
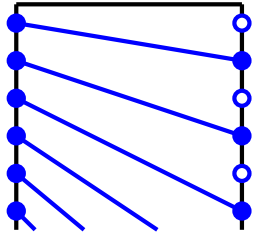
\begin{center}
    \begin{tikzpicture}
        \begin{scope}[ultra thick]
            \draw (0,-.5) -- (0,2.5);
            \draw (0,2.5) -- (3,2.5);
            \draw (3,-.5) -- (3,2.5);

            \begin{scope}[color=blue]
                \draw (0,2.25) -- (3,1.75);
                \draw (0,1.75) -- (3,0.75);
                \draw (0,1.25) -- (3,-0.25);
                \draw (0,0.75) -- (1.875,-0.5);
                \draw (0,0.25) -- (0.9,-0.5);
                \draw (0,-0.25) -- (0.25,-0.5);

			    \draw[fill=blue] (0,2.25) circle (.1cm);
			    \draw[fill=blue] (0,1.75) circle (.1cm);
			    \draw[fill=blue] (0,1.25) circle (.1cm);
			    \draw[fill=blue] (0,0.75) circle (.1cm);
			    \draw[fill=blue] (0,0.25) circle (.1cm);
                \draw[fill=blue] (0,-0.25) circle (.1cm);
			
			    \draw[fill=white] (3,2.25) circle (.1cm);
			    \draw[fill=blue] (3,1.75) circle (.1cm);
			    \draw[fill=white] (3,1.25) circle (.1cm);
			    \draw[fill=blue] (3,0.75) circle (.1cm);
			    \draw[fill=white] (3,0.25) circle (.1cm);
                \draw[fill=blue] (3,-.25) circle (.1cm);
            \end{scope}
        \end{scope}
    \end{tikzpicture}
\end{center}
which we call $D^\ell$. We have that
$$
    W^r\Phi(b\otimes p_\ast)(W^r)^* = b,\qquad (W^r)^*W^r\Phi(b\otimes p_\ast) = \Phi(b\otimes p_\ast) = \Phi(b\otimes p_\ast)(W^r)^*W^r
$$
for all $b\in\mathcal{B}$. It follows from this that we may define a unitary isomorphism
$$
    \Psi^r_M\colon M\otimes p_\ast\mathcal{B}\longrightarrow M
$$
by the formula
$$
    \Psi^r_M\big((\xi b_0\otimes p_\ast b_1)\otimes a\big) = \xi W^r\Phi(b_0\otimes p_\ast b_1)a
$$
for $\xi\in M$ and $b_0,b_1,a\in\mathcal{B}$. Note that the adjoint (and inverse) of $\Psi^r_M$ is given by the formula
$$
    (\Psi^r_M)^*(\eta c) = (\eta\otimes p_\ast)\otimes (W^r)^*c
$$
for $\eta\in M$ and $c\in\mathcal{B}$. Clearly, the assignment $M\mapsto \Psi^r_M$ is natural in $M$.

Similarly, we can define a unitary isomorphism
$$
    \Psi_M^\ell\colon p_\ast\mathcal{B}\otimes M\longrightarrow M
$$
in $\ModB$ by the formula
$$
    \Psi_M^\ell\big((p_\ast b_0\otimes \xi b_1)\otimes a\big) = \xi W^\ell\Phi(p_\ast b_0\otimes b_1)a
$$
for $\xi\in M$ and $b_0,b_1,a\in\mathcal{B}$. Again, the assignment $M\mapsto \Psi^\ell_M$ is natural in $M$.

\subsubsection{Triangle identity}

In the present context, the triangle identity states that
\begin{equation}\label{eq:triangle}
    (\id_{M_1}\otimes\Psi^\ell_{M_2})\circ\alpha_{M_1,p_\ast\mathcal{B},M_2} = \Psi^r_{M_1}\otimes\id_{M_2}
\end{equation}
for any objects $M_1$ and $M_2$ in $\ModB$. By applying both sides to an element of the form $\xi b_0\otimes b_1b_1'\otimes b_2\otimes\eta b_3b_3'b_3''\otimes b_4$, we see that the verification reduces to proving the identity
$$
    \Phi(b\otimes W^\ell\Phi(p_\ast c\otimes d))V = \Phi(W^r\Phi(b\otimes p_\ast c)\otimes d).
$$
for $b,c,d\in\mathcal{B}$. Similarly to the case of the pentagon identity, it suffices to prove that
$$
    \Phi(p_{\vec x}\otimes v(D^\ell,\vec y))v(D^\alpha,\vec z) = \Phi(v(D^r,\vec \beta)\otimes p_{\vec\gamma})
$$
whenever $\vec x, \vec y, \vec z, \vec\beta,\vec\gamma\in\mathcal{G}^\infty$ are such that the patterns agree. Note that the operator on the left hand side arises from a finite dilute braid diagram such as
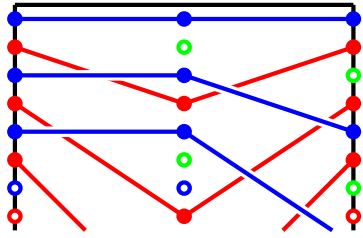
\begin{center}
    \begin{tikzpicture}[scale=0.75]
        \begin{scope}[ultra thick]
            \draw (0,0) -- (0,4);
            \draw (0,4) -- (6,4);
            \draw (6,0) -- (6,4);

            \aline{0}{3.25}{3}{2.25}{red}
            \aline{0}{2.25}{3}{0.25}{red}
            \aline{0}{1.25}{1.25}{0}{red}

            \aline{3}{2.25}{6}{3.25}{red}
            \aline{3}{0.25}{6}{2.25}{red}
            \aline{4.75}{0}{6}{1.25}{red}

            \sline{0}{3.75}{3}{3.75}{blue}
            \sline{0}{2.75}{3}{2.75}{blue}
            \sline{0}{1.75}{3}{1.75}{blue}

            \sline{3}{3.75}{6}{3.75}{blue}
            \sline{3}{2.75}{6}{1.75}{blue}
            \sline{3}{1.75}{5.625}{0}{blue}

            \begin{scope}[color=blue]
                \draw[fill=blue] (0,3.75) circle (.1cm);
                \draw[fill=blue] (0,2.75) circle (.1cm);
                \draw[fill=blue] (0,1.75) circle (.1cm);
                \draw[fill=white] (0,0.75) circle (.1cm);

                \draw[fill=blue] (3,3.75) circle (.1cm);
                \draw[fill=blue] (3,2.75) circle (.1cm);
                \draw[fill=blue] (3,1.75) circle (.1cm);
                \draw[fill=white] (3,0.75) circle (.1cm);

                \draw[fill=blue] (6,3.75) circle (.1cm);
                \draw[fill=blue] (6,1.75) circle (.1cm);
            \end{scope}

            \begin{scope}[color=green]
                \draw[fill=white] (3,3.25) circle (.1cm);
                \draw[fill=white] (3,1.25) circle (.1cm);

                \draw[fill=white] (6,2.75) circle (.1cm);
                \draw[fill=white] (6,0.75) circle (.1cm);
            \end{scope}

            \begin{scope}[color=red]
                \draw[fill=red] (0,3.25) circle (.1cm);
                \draw[fill=red] (0,2.25) circle (.1cm);
                \draw[fill=red] (0,1.25) circle (.1cm);
                \draw[fill=white] (0,0.25) circle (.1cm);

                \draw[fill=red] (3,2.25) circle (.1cm);
                \draw[fill=red] (3,0.25) circle (.1cm);

                \draw[fill=red] (6,3.25) circle (.1cm);
                \draw[fill=red] (6,2.25) circle (.1cm);
                \draw[fill=red] (6,1.25) circle (.1cm);
                \draw[fill=white] (6,0.25) circle (.1cm);
            \end{scope}
        \end{scope}
    \end{tikzpicture}
\end{center}
while the operator on the right hand side arises from
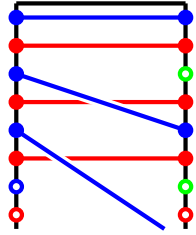
\begin{center}
    \begin{tikzpicture}[scale=0.75]
        \begin{scope}[ultra thick]
            \draw (0,0) -- (0,4);
            \draw (0,4) -- (3,4);
            \draw (3,0) -- (3,4);

            \aline{0}{3.25}{3}{3.25}{red}
            \aline{0}{2.25}{3}{2.25}{red}
            \aline{0}{1.25}{3}{1.25}{red}

            \sline{0}{3.75}{3}{3.75}{blue}
            \sline{0}{2.75}{3}{1.75}{blue}
            \sline{0}{1.75}{2.625}{0}{blue}

            \begin{scope}[color=blue]
                \draw[fill=blue] (0,3.75) circle (.1cm);
                \draw[fill=blue] (0,2.75) circle (.1cm);
                \draw[fill=blue] (0,1.75) circle (.1cm);
                \draw[fill=white] (0,0.75) circle (.1cm);

                \draw[fill=blue] (3,3.75) circle (.1cm);
                \draw[fill=blue] (3,1.75) circle (.1cm);
            \end{scope}

            \begin{scope}[color=green]
                \draw[fill=white] (3,2.75) circle (.1cm);
                \draw[fill=white] (3,0.75) circle (.1cm);
            \end{scope}

            \begin{scope}[color=red]
                \draw[fill=red] (0,3.25) circle (.1cm);
                \draw[fill=red] (0,2.25) circle (.1cm);
                \draw[fill=red] (0,1.25) circle (.1cm);
                \draw[fill=white] (0,0.25) circle (.1cm);

                \draw[fill=red] (3,3.25) circle (.1cm);
                \draw[fill=red] (3,2.25) circle (.1cm);
                \draw[fill=red] (3,1.25) circle (.1cm);
                \draw[fill=white] (3,0.25) circle (.1cm);
            \end{scope}
        \end{scope}
    \end{tikzpicture}
\end{center}
which can be obtained from the top diagram by a finite sequence of Reidemeister moves of type 2.

\subsubsection{Simplicity of the tensor unit}

To finish the proof that $\ModB$ is a C*-tensor category, we note that $p_\ast\mathcal{B}$ is a simple object in $\ModB$. Indeed, one easily checks that
$$
    \mathrm{End}(p_\ast\mathcal{B}) \cong p_\ast\mathcal{B}p_\ast = p_\ast\mathcal{B}_0p_\ast = \mathbb{C}p_\ast
$$
(see also the proof of Lemma \ref{lem:full} below).

\subsection{A unitary braiding on $\mathrm{Mod}_{\mathcal{B}}$}\label{sec:Braiding}

We next define a unitary braiding on $\ModB$ and verify the hexagon identities.

\subsubsection{Definition of the braiding}

Denote by $U$ the unitary operator in $\mathbb{B}(H)$ that is associated to the infinite braid diagram $D^\sigma$ that is formed as follows. First connect the nodes on the left side numbered $2$, $4$, $\ldots$ to those on the right side numbered $1$, $3$, $\ldots$ by red strands (as in the following figure). Next, for each of the remaining nodes on the left numbered $2k-1$, say, draw a blue strand from it to the top of the diagram, crossing over the red strands whose left end point is above it, and then continue this strand to the node numbered $2k$ on the right side, now crossing under the red strands whose right end point is above that node. The following figure shows one of the associated finite dilute braid diagrams $D^\sigma_{\vec x}$ ($=(D^\sigma)_{\vec x}$).
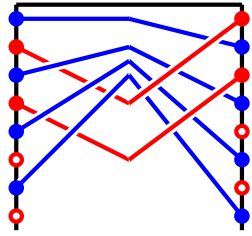
\begin{center}
    \begin{tikzpicture}[scale=0.75]
        \begin{scope}[ultra thick]
            \draw (0,0) -- (0,4);
            \draw (0,4) -- (4,4);
            \draw (4,0) -- (4,4);

            \aline{0}{3.25}{2}{2.25}{red}
            \aline{0}{2.25}{2}{1.25}{red}
			
            \sline{0}{0.75}{2}{2.75}{blue}
            \sline{0}{1.75}{2}{3}{blue}
            \sline{0}{2.75}{2}{3.25}{blue}
            \sline{0}{3.75}{2}{3.75}{blue}
			
            \aline{2}{3.75}{4}{3.25}{blue}
            \aline{2}{3.25}{4}{2.25}{blue}
            \aline{2}{3}{4}{1.25}{blue}
            \aline{2}{2.75}{4}{0.25}{blue}

            \sline{2}{2.25}{4}{3.75}{red}
            \sline{2}{1.25}{4}{2.75}{red}			

            \draw[color=blue] (2,3.75) circle (.0015cm);
            \draw[color=blue] (2,3.245) circle (.0015cm);
            \draw[color=blue] (2,2.995) circle (.0015cm);
            \draw[color=blue] (2,2.745) circle (.0015cm);
            \draw[color=red] (2,2.255) circle (.0015cm);
            \draw[color=red] (2,1.255) circle (.0015cm);			

			\begin{scope}[color=red]
                \draw[fill=red] (0,3.25) circle (.1cm);
                \draw[fill=red] (0,2.25) circle (.1cm);
                \draw[fill=white] (0,1.25) circle (.1cm);
                \draw[fill=white] (0,0.25) circle (.1cm);

                \draw[fill=red] (4,3.75) circle (.1cm);
                \draw[fill=red] (4,2.75) circle (.1cm);
                \draw[fill=white] (4,1.75) circle (.1cm);
                \draw[fill=white] (4,0.75) circle (.1cm);
            \end{scope}

            \begin{scope}[color=blue]
            	\draw[fill=blue] (0,3.75) circle (.1cm);
            	\draw[fill=blue] (0,1.75) circle (.1cm);
            	
            	\draw[fill=blue] (4,3.25) circle (.1cm);
            	\draw[fill=blue] (4,1.25) circle (.1cm);
            \end{scope}

            \begin{scope}[color=blue]
            	\draw[fill=blue] (0,2.75) circle (.1cm);
            	\draw[fill=blue] (0,.75) circle (.1cm);
            	
            	\draw[fill=blue] (4,2.25) circle (.1cm);
            	\draw[fill=blue] (4,.25) circle (.1cm);
            \end{scope}
        \end{scope}
    \end{tikzpicture}
\end{center}

Note that
\begin{equation}\label{eq:braid}
    U\Phi(b_1\otimes b_2)U^* = \Phi(b_2\otimes b_1)
\end{equation}
for all $b_1,b_2\in\mathcal{B}$.

Equation (\ref{eq:braid}) allows us, given two objects $M_1$ and $M_2$ in $\ModB$, to define a unitary isomorphism
$$
    \sigma_{M_1,M_2}\colon M_1\otimes M_2\longrightarrow M_2\otimes M_1
$$
by the formula
$$
    \sigma_{M_1,M_2}((\xi_1\otimes\xi_2)\otimes a) = (\xi_2 \otimes\xi_1)\otimes Ua,
$$
for $\xi_1\in M_1$, $\xi_2\in M_2$ and $a\in\mathcal{B}$. The assignment $(M_1,M_2)\mapsto \sigma_{M_1,M_2}$ is clearly natural in $M_1$ and $M_2$ and will turn out to be a unitary braiding on $\ModB$.

\subsubsection{Hexagon identities}

In the present context, the two hexagon identities are
\begin{align}
    \label{eq:hexagon1}\alpha_{M_2,M_3,M_1}\circ \sigma_{M_1,M_2\otimes M_3}\circ\alpha_{M_1,M_2,M_3} &= (\id_{M_2}\otimes \sigma_{M_1,M_3})\circ \alpha_{M_2,M_1,M_3}\circ (\sigma_{M_1,M_2}\otimes\id_{M_3}),\\
    \label{eq:hexagon2}\alpha_{M_2,M_3,M_1}\circ \sigma_{M_2\otimes M_3,M_1}^*\circ\alpha_{M_1,M_2,M_3} &= (\id_{M_2}\otimes \sigma_{M_3,M_1}^*)\circ \alpha_{M_2,M_1,M_3}\circ (\sigma_{M_2,M_1}^*\otimes\id_{M_3})
\end{align}
for any objects $M_1$, $M_2$ and $M_2$ in $\ModB$. Let us prove the first identity and leave the second one to the reader. The left hand side maps an element of the form $\xi_1 aa'\otimes \xi_2b\otimes c\otimes\xi_3 dd'd''\otimes e$ to
$$
    \xi_2\otimes\xi_3\otimes \xi_1\otimes\Phi(d\otimes a)\otimes V\Phi(\Phi(b\otimes d')\otimes a')UV\Phi(c\otimes d'')e
$$
while the right hand side maps it to
$$
    \xi_2\otimes\xi_3\otimes\xi_1\otimes\Phi(d\otimes a)\otimes\Phi(b\otimes U\Phi(a'\otimes d'))V\Phi(Uc\otimes d'')e.
$$
Thus, the first hexagon identity would follow from the identities
$$
    \Phi(a\otimes\Phi(b\otimes c))VUV\Phi(d\otimes e) = \Phi(a\otimes\Phi(b\otimes c)U)V\Phi(Ud\otimes e)
$$
for $a,b,c,d,e\in\mathcal{B}$. As in the case of the pentagon identity, this reduces to showing that
$$
    v(D^\alpha,\vec x)v(D^\sigma,\vec y)v(D^\alpha,\vec z) = \Phi(p_{\vec\mu}\otimes v(D^\sigma,\vec\nu))v(D^\alpha,\vec\beta)\Phi(v(D^\sigma,\vec\gamma)\otimes p_{\vec \epsilon})
$$
whenever $\vec x, \vec y, \vec z, \vec \mu, \vec\nu, \vec\beta,\vec\gamma,\vec\epsilon\in\mathcal{G}^\infty$ are such that the patterns agree. In this identity, the operator on each side arises from a certain finite dilute braid diagram.
The next figure shows a sample pair of diagrams that can appear. On the left hand side, we could have
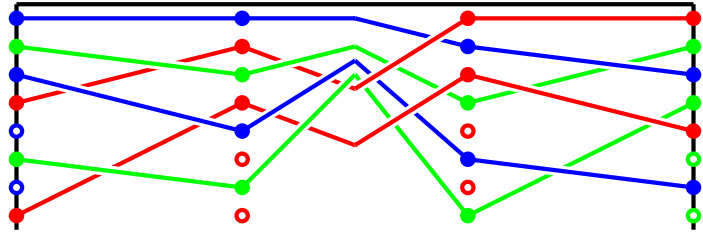
\begin{center}
	\begin{tikzpicture} [scale=.75]
		\begin{scope}[ultra thick]

\draw (0,0) -- (0,4);
\draw (0,4) -- (12,4);
\draw (12,0) -- (12,4);

\sline{0}{0.25}{4}{2.25}{red}			
\sline{0}{2.25}{4}{3.25}{red}
\sline{4}{3.25}{6}{2.5}{red}
\sline{4}{2.25}{6}{1.5}{red}

\sline{0}{1.25}{4}{0.75}{green}
\sline{0}{3.25}{4}{2.75}{green}
\sline{4}{0.75}{6}{2.75}{green}			
\sline{4}{2.75}{6}{3.25}{green}

\sline{0}{3.75}{6}{3.75}{blue}
\sline{0}{2.75}{4}{1.75}{blue}
\sline{4}{1.75}{6}{3}{blue}

\aline{6}{2.75}{8}{0.25}{green}
\aline{6}{3.25}{8}{2.25}{green}
\aline{8}{0.25}{12}{2.25}{green}
\aline{8}{2.25}{12}{3.25}{green}

\aline{6}{3.75}{8}{3.25}{blue}
\aline{6}{3}{8}{1.25}{blue}
\sline{8}{3.25}{12}{2.75}{blue}
\sline{8}{1.25}{12}{0.75}{blue}

\sline{6}{2.5}{8}{3.75}{red}
\sline{6}{1.5}{8}{2.75}{red}
\sline{8}{3.75}{12}{3.75}{red}
\sline{8}{2.75}{12}{1.75}{red}

\draw[color=red] (6,2.506) circle (.004cm);
\draw[color=red] (6,1.506) circle (.004cm);
\draw[color=blue] (6,3.75) circle (.0015cm);
\draw[color=blue] (6,2.995) circle (.0019cm);
\draw[color=green] (6,3.248) circle (.0015cm);
\draw[color=green] (6,2.748) circle (.0010cm);
			
			\begin{scope}[color=red]
                \draw[fill=red] (0,0.25) circle (.1cm);
                \draw[fill=red] (0,2.25) circle (.1cm);

                \draw[fill=red] (4,3.25) circle (.1cm);
                \draw[fill=red] (4,2.25) circle (.1cm);
                \draw[fill=white] (4,1.25) circle (.1cm);
                \draw[fill=white] (4,0.25) circle (.1cm);

                \draw[fill=red] (8,3.75) circle (.1cm);
                \draw[fill=red] (8,2.75) circle (.1cm);
                \draw[fill=white] (8,1.75) circle (.1cm);
                \draw[fill=white] (8,0.75) circle (.1cm);

                \draw[fill=red] (12,3.75) circle (.1cm);
                \draw[fill=red] (12,1.75) circle (.1cm);
            \end{scope}

            \begin{scope}[color=blue]
            	\draw[fill=blue] (0,3.75) circle (.1cm);
            	\draw[fill=blue] (0,2.75) circle (.1cm);
            	\draw[fill=white] (0,1.75) circle (.1cm);
            	\draw[fill=white] (0,0.75) circle (.1cm);
            	
            	\draw[fill=blue] (4,3.75) circle (.1cm);
            	\draw[fill=blue] (4,1.75) circle (.1cm);
            	
            	\draw[fill=blue] (8,3.25) circle (.1cm);
            	\draw[fill=blue] (8,1.25) circle (.1cm);
            	
            	\draw[fill=blue] (12,2.75) circle (.1cm);
            	\draw[fill=blue] (12,.75) circle (.1cm);
            \end{scope}

            \begin{scope}[color=green]
            	\draw[fill=green] (0,3.25) circle (.1cm);
            	\draw[fill=green] (0,1.25) circle (.1cm);
            	
            	\draw[fill=green] (4,2.75) circle (.1cm);
            	\draw[fill=green] (4,.75) circle (.1cm);
            	
            	\draw[fill=green] (8,2.25) circle (.1cm);
            	\draw[fill=green] (8,.25) circle (.1cm);
            	
            	\draw[fill=green] (12,3.25) circle (.1cm);
            	\draw[fill=green] (12,2.25) circle (.1cm);
            	\draw[fill=white] (12,1.25) circle (.1cm);
            	\draw[fill=white] (12,0.25) circle (.1cm);
            \end{scope}
		\end{scope}
	\end{tikzpicture}
\end{center}
which would be paired with the following diagram on the right hand side.
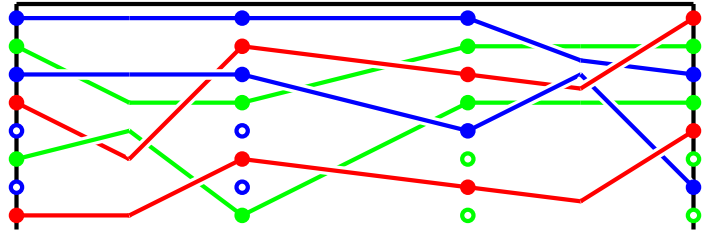
\begin{center}
	\begin{tikzpicture} [scale=.75]
		\begin{scope}[ultra thick]

\draw (0,0) -- (0,4);
\draw (0,4) -- (12,4);
\draw (12,0) -- (12,4);


\aline{0}{2.25}{2}{1.25}{red}
\aline{0}{0.25}{2}{0.25}{red}
\sline{0}{3.25}{2}{2.25}{green}
\sline{0}{1.25}{2}{1.75}{green}
\sline{0}{3.75}{2}{3.75}{blue}
\sline{0}{2.75}{2}{2.75}{blue}


\aline{2}{2.25}{4}{2.25}{green}
\aline{2}{1.75}{4}{0.25}{green}
\aline{4}{2.25}{8}{3.25}{green}
\aline{4}{0.25}{8}{2.25}{green}
\aline{8}{3.25}{10}{3.25}{green}
\aline{8}{2.25}{10}{2.25}{green}

\sline{2}{1.25}{4}{3.25}{red}
\sline{2}{0.25}{4}{1.25}{red}
\sline{4}{3.25}{8}{2.75}{red}
\sline{4}{1.25}{8}{0.75}{red}
\sline{8}{2.75}{10}{2.5}{red}
\sline{8}{0.75}{10}{0.5}{red}

\sline{2}{3.75}{4}{3.75}{blue}
\sline{2}{2.75}{4}{2.75}{blue}
\sline{4}{3.75}{8}{3.75}{blue}
\sline{4}{2.75}{8}{1.75}{blue}
\sline{8}{3.75}{10}{3}{blue}
\sline{8}{1.75}{10}{2.75}{blue}


\aline{10}{3.25}{12}{3.25}{green}
\aline{10}{2.25}{12}{2.25}{green}
\sline{10}{3}{12}{2.75}{blue}
\sline{10}{2.75}{12}{0.75}{blue}
\sline{10}{2.5}{12}{3.75}{red}
\sline{10}{0.5}{12}{1.75}{red}

\draw[color=red] (1.99,1.255) circle (.0025cm);
\draw[color=red] (2,0.252) circle (.003cm);
\draw[color=green] (2,2.25) circle (.0010cm);
\draw[color=green] (2,1.75) circle (.0010cm);

\draw[color=red] (10,2.505) circle (.0025cm);
\draw[color=red] (10,0.505) circle (.003cm);
\draw[color=blue] (10,3) circle (.0010cm);
\draw[color=blue] (10,2.745) circle (.0010cm);
			
			\begin{scope}[color=red]
                \draw[fill=red] (0,0.25) circle (.1cm);
                \draw[fill=red] (0,2.25) circle (.1cm);

                \draw[fill=red] (4,3.25) circle (.1cm);
                \draw[fill=red] (4,1.25) circle (.1cm);

                \draw[fill=red] (8,2.75) circle (.1cm);
                \draw[fill=red] (8,0.75) circle (.1cm);

                \draw[fill=red] (12,3.75) circle (.1cm);
                \draw[fill=red] (12,1.75) circle (.1cm);
            \end{scope}

            \begin{scope}[color=blue]
            	\draw[fill=blue] (0,3.75) circle (.1cm);
            	\draw[fill=blue] (0,2.75) circle (.1cm);
            	\draw[fill=white] (0,1.75) circle (.1cm);
            	\draw[fill=white] (0,0.75) circle (.1cm);
            	
            	\draw[fill=blue] (4,3.75) circle (.1cm);
            	\draw[fill=blue] (4,2.75) circle (.1cm);
            	\draw[fill=white] (4,1.75) circle (.1cm);
            	\draw[fill=white] (4,0.75) circle (.1cm);
            	
            	\draw[fill=blue] (8,3.75) circle (.1cm);
            	\draw[fill=blue] (8,1.75) circle (.1cm);
            	
            	\draw[fill=blue] (12,2.75) circle (.1cm);
            	\draw[fill=blue] (12,.75) circle (.1cm);
            \end{scope}

            \begin{scope}[color=green]
            	\draw[fill=green] (0,3.25) circle (.1cm);
            	\draw[fill=green] (0,1.25) circle (.1cm);
            	
            	\draw[fill=green] (4,2.25) circle (.1cm);
            	\draw[fill=green] (4,.25) circle (.1cm);
            	
            	\draw[fill=green] (8,3.25) circle (.1cm);
            	\draw[fill=green] (8,2.25) circle (.1cm);
            	\draw[fill=white] (8,1.25) circle (.1cm);
            	\draw[fill=white] (8,0.25) circle (.1cm);
            	
            	\draw[fill=green] (12,3.25) circle (.1cm);
            	\draw[fill=green] (12,2.25) circle (.1cm);
            	\draw[fill=white] (12,1.25) circle (.1cm);
            	\draw[fill=white] (12,0.25) circle (.1cm);
            \end{scope}
		\end{scope}
	\end{tikzpicture}
\end{center}
Note that, in both diagrams, the blue strands always cross over the green strands. Thus, one can transform both diagrams into the same diagram by pulling the green and blue strands up and pulling the red strands down. In the case of our sample pair of diagrams, the common diagram is
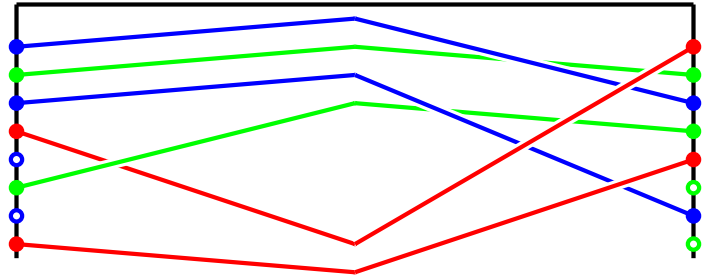
\begin{center}
	\begin{tikzpicture} [scale=.75]
		\begin{scope}[ultra thick]

\draw (0,0) -- (0,4.5);
\draw (0,4.5) -- (12,4.5);
\draw (12,0) -- (12,4.5);

\aline{0}{0.25}{6}{-0.25}{red}			
\aline{0}{2.25}{6}{0.25}{red}

\sline{0}{1.25}{6}{2.75}{green}
\sline{0}{3.25}{6}{3.75}{green}

\sline{0}{3.75}{6}{4.25}{blue}
\sline{0}{2.75}{6}{3.25}{blue}

\aline{6}{2.75}{12}{2.25}{green}
\aline{6}{3.75}{12}{3.25}{green}

\sline{6}{4.25}{12}{2.75}{blue}
\sline{6}{3.25}{12}{0.75}{blue}

\sline{6}{0.25}{12}{3.75}{red}
\sline{6}{-0.25}{12}{1.75}{red}

\draw[color=red] (6,0.25) circle (.0015cm);
\draw[color=red] (6,-0.25) circle (.0015cm);
\draw[color=blue] (6,4.25) circle (.0015cm);
\draw[color=blue] (6,3.25) circle (.0015cm);
\draw[color=green] (6,3.75) circle (.0015cm);
\draw[color=green] (6,2.75) circle (.0015cm);
			
			\begin{scope}[color=red]
                \draw[fill=red] (0,0.25) circle (.1cm);
                \draw[fill=red] (0,2.25) circle (.1cm);

                \draw[fill=red] (12,3.75) circle (.1cm);
                \draw[fill=red] (12,1.75) circle (.1cm);
            \end{scope}

            \begin{scope}[color=blue]
            	\draw[fill=blue] (0,3.75) circle (.1cm);
            	\draw[fill=blue] (0,2.75) circle (.1cm);
            	\draw[fill=white] (0,1.75) circle (.1cm);
            	\draw[fill=white] (0,0.75) circle (.1cm);
            	
            	\draw[fill=blue] (12,2.75) circle (.1cm);
            	\draw[fill=blue] (12,.75) circle (.1cm);
            \end{scope}

            \begin{scope}[color=green]
            	\draw[fill=green] (0,3.25) circle (.1cm);
            	\draw[fill=green] (0,1.25) circle (.1cm);
            	
            	\draw[fill=green] (12,3.25) circle (.1cm);
            	\draw[fill=green] (12,2.25) circle (.1cm);
            	\draw[fill=white] (12,1.25) circle (.1cm);
            	\draw[fill=white] (12,0.25) circle (.1cm);
            \end{scope}
		\end{scope}
	\end{tikzpicture}
\end{center}
Since this only involves Reidemeister moves of types 2 and 3, the associated operators are equal.

\subsection{The full C*-tensor subcategory $\mathrm{Mod}_{\mathcal{B}}^f$}\label{sec:ModBf}

Denote by $\ModB^f$ the full subcategory of $\ModB$ whose objects are those right Hilbert $\mathcal{B}$-modules which admit a finite orthonormal basis in the sense of section \ref{sec:HilbMod}. (Note that, by \cite{BG}, every module in $\ModB$ admits a possibly infinite orthonormal basis.) Clearly, $\ModB^f$ contains the tensor unit in $\ModB$. In order to check that $\ModB^f$ is a C*-tensor subcategory of $\ModB$, we must show that $\ModB^f$ is closed under tensor products. To do this, let $M$ and $N$ be objects in $\ModB^f$. Choose finite orthonormal bases $(\xi_i)_i$ and $(\eta_j)_j$ for $M$ and $N$, respectively. Using the identities $\xi_i\langle \xi_i,\xi_i\rangle = \xi_i$ and $\eta_j\langle \eta_j,\eta_j\rangle = \eta_j$, one verifies that the elements
$
    \xi_i\otimes\eta_j\otimes \Phi(\langle\xi_i,\xi_i\rangle\otimes\langle \eta_j,\eta_j\rangle)
$
form a finite orthonormal basis for $M\otimes N$, showing that $M\otimes N$ is an object in $\ModB^f$. Note that this could also be deduced from the following easily proved fact.

\begin{fact}\label{lem:ONB}
    Let $\{\xi_1,\ldots,\xi_k\}$ be a finite orthonormal basis for a right Hilbert $\mathcal{B}$-module $M$. Then $M$ is isomorphic to $\bigoplus_{j=1}^k p_j\mathcal{B}$, where $p_j = \langle\xi_j,\xi_j\rangle$ for all $j$.
\end{fact}

\section{Realizing $\TLJcat(\delta)$ as right Hilbert $\mathcal{B}$-modules}\label{sec:Section5}

In this section, we show that $\TLJcat(\delta)$ is equivalent to $\ModB^f$ as a braided C*-tensor category.

\subsection{A braided monoidal *-functor $F\colon \TLJcat(\delta)\to \mathrm{Mod}_{\mathcal{B}}$}\label{sec:Functor}

We will now define a functor $F\colon \mathcal{C} \to \ModB$ (where $\mathcal{C} = \TLJcat(\delta)$, as above). The following notation will be convenient. Setting $\vec x_n = (\pi,\ldots,\pi,\mathbbm{1},\mathbbm{1},\ldots)$, with $n$ leading copies of $\pi$, we denote $L_{\vec x_n,\vec x_m}(a)$ by $L_{n,m}(a)$ for any $a\in\mathrm{Hom}(\pi^{\otimes m},\pi^{\otimes n}) = \mathrm{TLJ}_{n,m}(\delta)$. We also put $L_n(a) = L_{n,n}(a)$ for any $a\in \mathrm{End}(\pi^{\otimes n}) = \mathrm{TLJ}_n(\delta)$. Finally, we denote by $p_n$ the projection $p_{\vec x_n}$ (as defined on page \pageref{def:px}).

We define $F$ on objects as follows. Given a projection $P\in\mathrm{TLJ}_n(\delta)$, we define $F(P)$ by the formula
$$
    F(P) = L_n(P)\mathcal{B},
$$
on the right hand side of which we view $P$ as a morphism. Given an object $\oplus_j P_j$ in $\mathcal{C}$, we put
$$
    F\big(\!\oplus_j P_j\big) = \oplus_j F(P_j).
$$
On the right hand side, the symbol $\oplus$ denotes the standard direct sum of right Hilbert $\mathcal{B}$-modules.

We next define $F$ on morphisms. Given a morphism $a\in\mathrm{Hom}(P,Q)$, where $P\in\mathrm{TLJ}_n(\delta)$ and $Q\in\mathrm{TLJ}_m(\delta)$ are projections, we define $F(a)$ to be the adjointable map $L_n(P)\mathcal{B}\to L_m(Q)\mathcal{B}$ given by left-multiplication by $L_{m,n}(a)$. For any morphism $(a_{ij})_{i,j}\in\mathrm{Hom}(\oplus_j P_j,\oplus_i Q_i)$, we define $F((a_{ij})_{i,j})$ to be the adjointable map $\oplus_j F(P_j)\to \oplus_i F(Q_i)$ associated to the matrix $(F(a_{ij}))_{i,j}$.

It is clear that $F$ is a *-functor. We will prove that it is in fact a braided monoidal *-functor. Thus, given any two objects $\rho$ and $\nu$ in $\mathcal{C}$, we will define a unitary isomorphism
$
    J_{\rho,\nu}\colon F(\rho)\otimes F(\nu)\to F(\rho\otimes\nu)
$
in such a way that the assignment $(\rho,\nu)\mapsto J_{\rho,\nu}$ is natural in $\rho$ and $\nu$ and all four of the following identities hold for all objects $\rho$, $\nu$ and $\mu$ in $\mathcal{C}$:
\begin{align}
    \label{eq:monoidal}J_{\rho,\nu\otimes \mu}\circ (\id_{F(\rho)}\otimes J_{\nu,\mu})\circ \alpha_{F(\rho),F(\nu),F(\mu)} &= J_{\rho\otimes\nu,\mu}\circ (J_{\rho,\nu}\otimes \id_{F(\mu)}),\\
    \label{eq:unital_right}J_{\rho,\mathbbm{1}} &= 
    \Psi^r_{F(\rho)},\\
    \label{eq:unital_left}J_{\mathbbm{1},\rho} &= 
    \Psi^\ell_{F(\rho)},\\
    \label{eq:braided}F(\sigma_{\rho,\nu}^{\mathrm{TL}})\circ J_{\rho,\nu} &= J_{\nu,\rho}\circ \sigma_{F(\rho),F(\nu)}.
\end{align}
(We are using the fact that $\mathcal{C}$ is strict.)
First, given $n,m\geq 0$, we define a unitary isomorphism
$$
    J_{n,m}\colon p_n\mathcal{B}\otimes p_m\mathcal{B}\longrightarrow p_{n+m}\mathcal{B}
$$
by the formula
$
    J_{n,m}(p_na\otimes p_mb\otimes c) = U_{n,m}\Phi(p_n a\otimes p_m b)c
$
for $a,b,c\in\mathcal{B}$, where $U_{n,m}$ is the partial isometry in $\mathcal{B}$ arising from a diagram like the one depicted below (in the case $n=4$ and $m=2$).
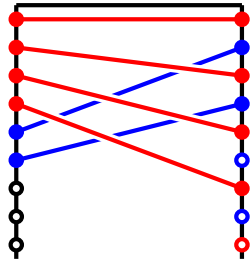
\begin{center}
    \begin{tikzpicture}[scale=0.75]
        \begin{scope}[ultra thick]
            \draw (0,0) -- (0, 4.5);
            \draw (0, 4.5) -- (4, 4.5);
            \draw (4,0) -- (4, 4.5);

            \aline{0}{2.25}{4}{3.75}{blue}
            \aline{0}{1.75}{4}{2.75}{blue}

            \sline{0}{4.25}{4}{4.25}{red}
            \sline{0}{3.75}{4}{3.25}{red}
            \sline{0}{3.25}{4}{2.25}{red}
            \sline{0}{2.75}{4}{1.25}{red}

            \begin{scope}
                \draw[fill=white] (0,1.25) circle (.1cm);
                \draw[fill=white] (0,0.75) circle (.1cm);
                \draw[fill=white] (0,0.25) circle (.1cm);
            \end{scope}

            \begin{scope}[color=blue]
                \draw[fill=blue] (0,2.25) circle (.1cm);
                \draw[fill=blue] (0,1.75) circle (.1cm);

                \draw[fill=blue] (4,3.75) circle (.1cm);
                \draw[fill=blue] (4,2.75) circle (.1cm);

                \draw[fill=white] (4,1.75) circle (.1cm);
                \draw[fill=white] (4,0.75) circle (.1cm);
            \end{scope}

            \begin{scope}[color=red]
            	\draw[fill=red] (0,4.25) circle (.1cm);
            	\draw[fill=red] (0,3.75) circle (.1cm);
            	\draw[fill=red] (0,3.25) circle (.1cm);
            	\draw[fill=red] (0,2.75) circle (.1cm);
            	
            	\draw[fill=red] (4,4.25) circle (.1cm);
            	\draw[fill=red] (4,3.25) circle (.1cm);
            	\draw[fill=red] (4,2.25) circle (.1cm);
            	\draw[fill=red] (4,1.25) circle (.1cm);
            	\draw[fill=white] (4,0.25) circle (.1cm);
            \end{scope}
        \end{scope}
    \end{tikzpicture}
\end{center}
The fact that $J_{n,m}$ is a well-defined unitary isomorphism comes down to the easily verified identities $U_{n,m}^*U_{n,m} = \Phi(p_n\otimes p_m)$ and $U_{n,m}U_{n,m}^* = p_{n+m}$.

Given projections $P\in\mathrm{TLJ}_n(\delta)$ and $Q\in\mathrm{TLJ}_m(\delta)$, we define $J_{P,Q}$ as the restriction of $J_{n,m}$. More generally, given two objects $\oplus_i P_i$ and $\oplus_j Q_j$ in $\mathcal{C}$, we define
$J_{\oplus_i P_i,\oplus_j Q_j}$ as the composition $\big(\oplus_{(i,j)} J_{P_i,Q_j}\big)\circ\phi$, where $\phi$ is the unitary isomorphism
$
    \big(\oplus_i F(P_i)\big)\otimes \big(\oplus_j F(Q_j)\big)\to \oplus_{(i,j)}(F(P_i)\otimes F(Q_j))
$
from section \ref{sec:Tensor}. Note that the domain of $J_{\oplus_i P_i,\oplus_j Q_j}$ is
$
    F(\oplus_i P_i)\otimes F(\oplus_j Q_j) = \big(\oplus_i F(P_i)\big)\otimes \big(\oplus_i F(P_i)\big)
$
while its codomain is
$
    F\big((\oplus_i P_i)\otimes (\oplus_j Q_j)) = \oplus_{(i,j)} (F(P_i)\otimes F(Q_j)).
$

It is easy to reduce the naturality of $J$ as well as the identities in 
equations (\ref{eq:monoidal})--(\ref{eq:braided}) to the case where $\rho$, 
$\nu$ and $\mu$ are projections in Temperley--Lieb--Jones C*-algebras. Since 
$J_{P,Q}$ is defined as the restriction of $J_{n,m}$, it is in fact enough to 
verify these identities in the case where $\rho$, $\nu$ and $\mu$ are identity 
elements in such algebras. This case can be taken care of by straightforward 
diagrammatic arguments. For the convenience of the reader, we indicate the 
proofs of equations (\ref{eq:monoidal}) and (\ref{eq:braided}), starting with 
the latter. In the case under consideration, equation (\ref{eq:braided}) is just
$$
    F(\sigma_{\pi^{\otimes n},\pi^{\otimes m}}^{\mathrm{TL}})\circ J_{n,m} = J_{m,n}\circ \sigma_{F(\pi^{\otimes n}),F(\pi^{\otimes m})}.
$$
The verification of this identity amounts to proving that
\begin{equation}\label{eq:braiding2}
    L_{n+m}(\sigma_{\pi^{\otimes n},\pi^{\otimes m}}^{\mathrm{TL}})\circ U_{n,m}\circ \Phi(p_n\otimes p_m) = U_{m,n}\circ U\circ \Phi(p_n\otimes p_m).
\end{equation}
In the case when $n=m=3$, the left hand side arises from the finite braid diagram
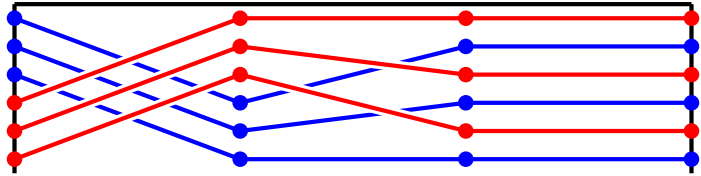
\begin{center}
    \begin{tikzpicture}[scale=0.75]
        \begin{scope}[ultra thick]
        \draw (0,0) -- (0,3);
        \draw (0,3) -- (12,3);
        \draw (12,0) -- (12,3);

        \aline{0}{2.75}{4}{1.25}{blue}
        \aline{0}{2.25}{4}{0.75}{blue}
        \aline{0}{1.75}{4}{0.25}{blue}

        \aline{4}{1.25}{8}{2.25}{blue}
        \aline{4}{0.75}{8}{1.25}{blue}
        \aline{4}{0.25}{8}{0.25}{blue}

        \aline{8}{2.25}{12}{2.25}{blue}
        \aline{8}{1.25}{12}{1.25}{blue}
        \aline{8}{0.25}{12}{0.25}{blue}

        \sline{0}{1.25}{4}{2.75}{red}
        \sline{0}{0.75}{4}{2.25}{red}
        \sline{0}{0.25}{4}{1.75}{red}

        \sline{4}{2.75}{8}{2.75}{red}
        \sline{4}{2.25}{8}{1.75}{red}
        \sline{4}{1.75}{8}{0.75}{red}

        \sline{8}{2.75}{12}{2.75}{red}
        \sline{8}{1.75}{12}{1.75}{red}
        \sline{8}{0.75}{12}{0.75}{red}

            \begin{scope}[color=blue]
                \draw[fill=blue] (0,2.75) circle (.1cm);
                \draw[fill=blue] (0,2.25) circle (.1cm);
                \draw[fill=blue] (0,1.75) circle (.1cm);

                \draw[fill=blue] (4,1.25) circle (.1cm);
                \draw[fill=blue] (4,0.75) circle (.1cm);
                \draw[fill=blue] (4,0.25) circle (.1cm);

                \draw[fill=blue] (8,2.25) circle (.1cm);
                \draw[fill=blue] (8,1.25) circle (.1cm);
                \draw[fill=blue] (8,0.25) circle (.1cm);

                \draw[fill=blue] (12,2.25) circle (.1cm);
                \draw[fill=blue] (12,1.25) circle (.1cm);
                \draw[fill=blue] (12,0.25) circle (.1cm);
            \end{scope}

            \begin{scope}[color=red]
            	\draw[fill=red] (0,1.25) circle (.1cm);
            	\draw[fill=red] (0,0.75) circle (.1cm);
            	\draw[fill=red] (0,0.25) circle (.1cm);

            	\draw[fill=red] (4,2.75) circle (.1cm);            	
            	\draw[fill=red] (4,2.25) circle (.1cm);
            	\draw[fill=red] (4,1.75) circle (.1cm);

            	\draw[fill=red] (8,2.75) circle (.1cm);
            	\draw[fill=red] (8,1.75) circle (.1cm);
                \draw[fill=red] (8,0.75) circle (.1cm);

            	\draw[fill=red] (12,2.75) circle (.1cm);
            	\draw[fill=red] (12,1.75) circle (.1cm);
                \draw[fill=red] (12,0.75) circle (.1cm);          	
            \end{scope}
            \end{scope}
    \end{tikzpicture}
\end{center}
while the right hand side arises from the diagram
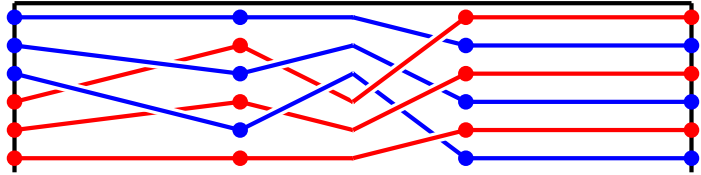
\begin{center}
    \begin{tikzpicture}[scale=0.75]
        \begin{scope}[ultra thick]
        \draw (0,0) -- (0,3);
        \draw (0,3) -- (12,3);
        \draw (12,0) -- (12,3);

        \aline{0}{1.25}{4}{2.25}{red}
        \aline{0}{0.75}{4}{1.25}{red}
        \aline{0}{0.25}{4}{0.25}{red}

        \aline{4}{2.25}{6}{1.25}{red}
        \aline{4}{1.25}{6}{0.75}{red}
        \aline{4}{0.25}{6}{0.25}{red}

        \sline{0}{2.75}{4}{2.75}{blue}
        \sline{0}{2.25}{4}{1.75}{blue}
        \sline{0}{1.75}{4}{0.75}{blue}

        \sline{4}{2.75}{6}{2.75}{blue}
        \sline{4}{1.75}{6}{2.25}{blue}
        \sline{4}{0.75}{6}{1.75}{blue}

        \aline{6}{2.75}{8}{2.25}{blue}
        \aline{6}{2.25}{8}{1.25}{blue}
        \aline{6}{1.75}{8}{0.25}{blue}

        \aline{8}{2.25}{12}{2.25}{blue}
        \aline{8}{1.25}{12}{1.25}{blue}
        \aline{8}{0.25}{12}{0.25}{blue}

        \sline{6}{1.25}{8}{2.75}{red}
        \sline{6}{0.75}{8}{1.75}{red}
        \sline{6}{0.25}{8}{0.75}{red}

        \sline{8}{2.75}{12}{2.75}{red}
        \sline{8}{1.75}{12}{1.75}{red}
        \sline{8}{0.75}{12}{0.75}{red}

            \begin{scope}[color=blue]
                \draw[fill=blue] (0,2.75) circle (.1cm);
                \draw[fill=blue] (0,2.25) circle (.1cm);
                \draw[fill=blue] (0,1.75) circle (.1cm);

                \draw[fill=blue] (4,2.75) circle (.1cm);
                \draw[fill=blue] (4,1.75) circle (.1cm);
                \draw[fill=blue] (4,0.75) circle (.1cm);

                \draw[fill=blue] (8,2.25) circle (.1cm);
                \draw[fill=blue] (8,1.25) circle (.1cm);
                \draw[fill=blue] (8,0.25) circle (.1cm);

                \draw[fill=blue] (12,2.25) circle (.1cm);
                \draw[fill=blue] (12,1.25) circle (.1cm);
                \draw[fill=blue] (12,0.25) circle (.1cm);
            \end{scope}

            \begin{scope}[color=red]
            	\draw[fill=red] (0,1.25) circle (.1cm);
            	\draw[fill=red] (0,0.75) circle (.1cm);
            	\draw[fill=red] (0,0.25) circle (.1cm);

            	\draw[fill=red] (4,2.25) circle (.1cm);            	
            	\draw[fill=red] (4,1.25) circle (.1cm);
            	\draw[fill=red] (4,0.25) circle (.1cm);

            	\draw[fill=red] (8,2.75) circle (.1cm);
            	\draw[fill=red] (8,1.75) circle (.1cm);
                \draw[fill=red] (8,0.75) circle (.1cm);

            	\draw[fill=red] (12,2.75) circle (.1cm);
            	\draw[fill=red] (12,1.75) circle (.1cm);
                \draw[fill=red] (12,0.75) circle (.1cm);          	
            \end{scope}

\draw[color=blue] (6,2.75) circle (.0015cm);
\draw[color=blue] (6,2.25) circle (.0015cm);
\draw[color=blue] (6,1.75) circle (.0015cm);
\draw[color=red] (6,1.25) circle (.0015cm);
\draw[color=red] (6,0.75) circle (.0015cm);
\draw[color=red] (6,0.25) circle (.0015cm);
            \end{scope}
    \end{tikzpicture}
\end{center}
As one of these diagrams can in general be obtained from the other by a finite sequence of Reidemeister moves of type $2$, we get equation (\ref{eq:braiding2}).
Similarly, equation (\ref{eq:monoidal}) reduces to the identity
\begin{equation}\label{eq:aux}
    U_{n,m+k}\circ\Phi(p_n\otimes U_{m,k})\circ V = U_{n+m,k}\circ\Phi(U_{n,m}\otimes p_k).
\end{equation}
In the case when $n = 3$ and $m = k = 2$, the left hand side arises from the diagram
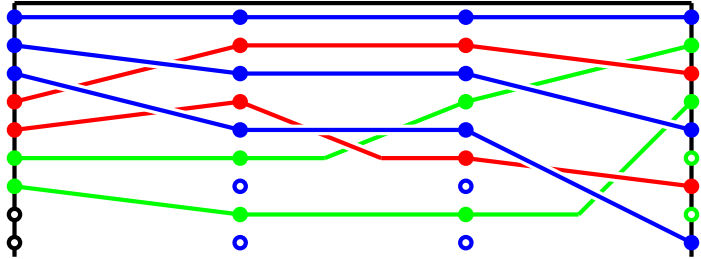
\begin{center}
    \begin{tikzpicture}[scale=0.75]
        \begin{scope}[ultra thick]
            \draw (0,0) -- (0,4.5);
            \draw (0,4.5) -- (12,4.5);
            \draw (12,0) -- (12,4.5);

            \aline{0}{1.75}{5.5}{1.75}{green}
            \aline{0}{1.25}{4}{0.75}{green}
            \aline{5.5}{1.75}{8}{2.75}{green}
            \aline{8}{2.75}{12}{3.75}{green}
            \aline{4}{0.75}{10}{0.75}{green}
            \aline{10}{0.75}{12}{2.75}{green}

            \aline{0}{2.75}{4}{3.75}{red}
            \sline{4}{3.75}{8}{3.75}{red}
            \sline{8}{3.75}{12}{3.25}{red}
            \aline{0}{2.25}{4}{2.75}{red}
            \sline{4}{2.75}{6.5}{1.75}{red}
            \sline{6.5}{1.75}{8}{1.75}{red}
            \sline{8}{1.75}{12}{1.25}{red}

            \sline{0}{4.25}{12}{4.25}{blue}
            \sline{0}{3.75}{4}{3.25}{blue}
            \sline{4}{3.25}{8}{3.25}{blue}
            \sline{8}{3.25}{12}{2.25}{blue}
            \sline{0}{3.25}{4}{2.25}{blue}
            \sline{4}{2.25}{8}{2.25}{blue}
            \sline{8}{2.25}{12}{0.25}{blue}

            \draw[fill=white] (0,0.75) circle (.1cm);
            \draw[fill=white] (0,0.25) circle (.1cm);

            \begin{scope}[color=red]
                \draw[fill=red] (0,2.75) circle (.1cm);
                \draw[fill=red] (0,2.25) circle (.1cm);

                \draw[fill=red] (4,3.75) circle (.1cm);
                \draw[fill=red] (4,2.75) circle (.1cm);

                \draw[fill=red] (8,3.75) circle (.1cm);
                \draw[fill=red] (8,1.75) circle (.1cm);

                \draw[fill=red] (12,3.25) circle (.1cm);
                \draw[fill=red] (12,1.25) circle (.1cm);
            \end{scope}

            \begin{scope}[color=green]
                \draw[fill=green] (0,1.75) circle (.1cm);
                \draw[fill=green] (0,1.25) circle (.1cm);

                \draw[fill=green] (4,1.75) circle (.1cm);
                \draw[fill=green] (4,0.75) circle (.1cm);

                \draw[fill=green] (8,2.75) circle (.1cm);
                \draw[fill=green] (8,0.75) circle (.1cm);

                \draw[fill=green] (12,3.75) circle (.1cm);
                \draw[fill=green] (12,2.75) circle (.1cm);
                \draw[fill=white] (12,1.75) circle (.1cm);
                \draw[fill=white] (12,0.75) circle (.1cm);
            \end{scope}

            \begin{scope}[color=blue]
                \draw[fill=blue] (0,4.25) circle (.1cm);
                \draw[fill=blue] (0,3.75) circle (.1cm);
                \draw[fill=blue] (0,3.25) circle (.1cm);

                \draw[fill=blue] (4,4.25) circle (.1cm);
                \draw[fill=blue] (4,3.25) circle (.1cm);
                \draw[fill=blue] (4,2.25) circle (.1cm);
                \draw[fill=white] (4,1.25) circle (.1cm);
                \draw[fill=white] (4,0.25) circle (.1cm);

                \draw[fill=blue] (8,4.25) circle (.1cm);
                \draw[fill=blue] (8,3.25) circle (.1cm);
                \draw[fill=blue] (8,2.25) circle (.1cm);
                \draw[fill=white] (8,1.25) circle (.1cm);
                \draw[fill=white] (8,0.25) circle (.1cm);

                \draw[fill=blue] (12,4.25) circle (.1cm);
                \draw[fill=blue] (12,2.25) circle (.1cm);
                \draw[fill=blue] (12,0.25) circle (.1cm);
            \end{scope}
        \end{scope}
    \end{tikzpicture}
\end{center}
while the right hand side arises from
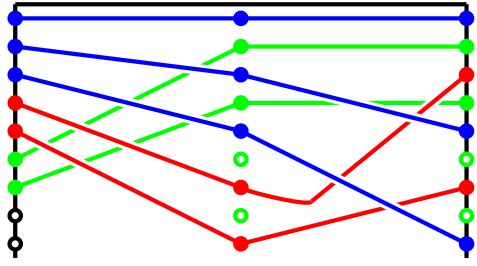
\begin{center}
    \begin{tikzpicture}[scale=0.75]
        \begin{scope}[ultra thick]
            \draw (0,0) -- (0,4.5);
            \draw (0,4.5) -- (8,4.5);
            \draw (8,0) -- (8,4.5);

            \aline{0}{1.75}{4}{3.75}{green}
            \aline{4}{3.75}{8}{3.75}{green}
            \aline{0}{1.25}{4}{2.75}{green}
            \aline{4}{2.75}{8}{2.75}{green}

            \sline{0}{2.75}{4}{1.25}{red}
            \draw[color=red] (4,1.25) .. controls (4,1.2) and (5.2,.9) .. (5.25,1);
            \sline{5.25}{1}{8}{3.25}{red}
            \sline{0}{2.25}{4}{0.25}{red}
            \sline{4}{0.25}{8}{1.25}{red}

            \sline{0}{4.25}{8}{4.25}{blue}
            \sline{0}{3.75}{4}{3.25}{blue}
            \sline{4}{3.25}{8}{2.25}{blue}
            \sline{0}{3.25}{4}{2.25}{blue}
            \sline{4}{2.25}{8}{0.25}{blue}

            \draw[fill=white] (0,0.75) circle (.1cm);
            \draw[fill=white] (0,0.25) circle (.1cm);

            \begin{scope}[color=red]
                \draw[fill=red] (0,2.75) circle (.1cm);
                \draw[fill=red] (0,2.25) circle (.1cm);

                \draw[fill=red] (4,1.25) circle (.1cm);
                \draw[fill=red] (4,0.25) circle (.1cm);

                \draw[fill=red] (8,3.25) circle (.1cm);
                \draw[fill=red] (8,1.25) circle (.1cm);
            \end{scope}

            \begin{scope}[color=green]
                \draw[fill=green] (0,1.75) circle (.1cm);
                \draw[fill=green] (0,1.25) circle (.1cm);

                \draw[fill=green] (4,3.75) circle (.1cm);
                \draw[fill=green] (4,2.75) circle (.1cm);
                \draw[fill=white] (4,1.75) circle (.1cm);
                \draw[fill=white] (4,0.75) circle (.1cm);

                \draw[fill=green] (8,3.75) circle (.1cm);
                \draw[fill=green] (8,2.75) circle (.1cm);
                \draw[fill=white] (8,1.75) circle (.1cm);
                \draw[fill=white] (8,0.75) circle (.1cm);
            \end{scope}

            \begin{scope}[color=blue]
                \draw[fill=blue] (0,4.25) circle (.1cm);
                \draw[fill=blue] (0,3.75) circle (.1cm);
                \draw[fill=blue] (0,3.25) circle (.1cm);

                \draw[fill=blue] (4,4.25) circle (.1cm);
                \draw[fill=blue] (4,3.25) circle (.1cm);
                \draw[fill=blue] (4,2.25) circle (.1cm);

                \draw[fill=blue] (8,4.25) circle (.1cm);
                \draw[fill=blue] (8,2.25) circle (.1cm);
                \draw[fill=blue] (8,0.25) circle (.1cm);
            \end{scope}
        \end{scope}
    \end{tikzpicture}
\end{center}
As in the proof of the pentagon identity, equation (\ref{eq:aux}) is verified in general by noting that the strands live on three separate layers (corresponding to the three colors used in the figures).

\subsection{$\TLJcat(\delta)$ and $\mathrm{Mod}_{\mathcal{B}}^f$ are equivalent}\label{sec:Equiv}

Finally, we will prove that $F$ is fully faithful (i.e., restricts to a bijection on each morphism space) and that if we restrict its codomain to the subcategory $\ModB^f$ then it is essentially surjective (i.e., hits every isomorphism class of objects).

\begin{lem}\label{lem:full}
    The functor $F$ is fully faithful.
\end{lem}

\begin{proof}
    It suffices to prove that $F$ restricts to a bijective map $\mathrm{Hom}(P,Q)\to \mathrm{Hom}(F(P),F(Q))$ for any given pair of projections $P\in\mathrm{TLJ}_n(\delta)$ and $Q\in\mathrm{TLJ}_m(\delta)$.

    We first prove injectivity. Let $a\in \mathrm{Hom}(P,Q)$ be such that $F(a) = 0$. Since $F$ is linear on morphisms, we need only show that $a=0$. Since $F(a)$ is left-multiplication by $L_{m,n}(a)$, we get that
    $0 = L_{m,n}(a)\circ L_n(P) = L_{m,n}(aP) = L_{m,n}(a)$, whereby $\|a\| = \|L_{m,n}(a)\| = 0$ by Lemma \ref{lem:iso}.

    We next prove surjectivity. Let $f\in \mathrm{Hom}(F(P),F(Q))$ be given. Then $f$ performs left-multipli\-cation by $b = f(L_n(P))\in L_m(Q)\mathcal{B}L_n(P)$. Set $K = \max\{m,n\}$ so that $L_n(P),L_m(Q)\in\mathcal{B}_K$. Since $\mathcal{B}_K$ is finite-dimensional and $L_m(Q)\mathcal{B}_kL_n(P)\subseteq L_m(Q)\mathcal{B}_KL_n(P)$ for all $k\geq 0$, it follows that $b\in L_m(Q)\mathcal{B}_KL_n(P)$. Thus, $b = L_{m,n}(a)$ for some morphism $a\in\mathrm{Hom}(P,Q)$.
\end{proof}

\begin{lem}
    The (codomain-restricted) functor $F\colon \TLJcat(\delta)\to \ModB^f$ is essentially surjective.
\end{lem}

\begin{proof}
    Let $M$ be any object in $\ModB^f$. By Fact \ref{lem:ONB}, $M$ is 
    isomorphic to a direct sum of modules of the form $p\mathcal{B}$, where $p$ 
    is a projection in $\mathcal{B}$. Writing such a projection $p$ as a finite 
    sum of minimal projections, we get that $p\mathcal{B}$ is isomorphic to a 
    direct sum of modules of the form $q\mathcal{B}$, where $q$ is a minimal 
    projection in $\mathcal{B}$. Thus, we may assume that $M = q\mathcal{B}$, 
    where $q$ is a rank one projection in, say, the summand $\mathbb{K}(H^s)$. 
    Pick $n\geq 0$ such that $s\prec\pi^{\otimes n}$, and let $v$ be a unit 
    vector in $\mathrm{Hom}(s,\pi^{\otimes n})$. Then $L_{n}(vv^*)$ is a rank 
    one projection in the summand $\mathbb{K}(H^s)$, and therefore Murray--von 
    Neumann equivalent to $q$ in $\mathcal{B}$. It follows that
    $
        q\mathcal{B} \cong L_{n}(vv^*)\mathcal{B} = F(vv^*),
    $
    which yields the stated result.
\end{proof}

In conclusion, we have the following theorem.

\begin{thm}\label{thm}
    If $\delta\in \{2\cos(\pi/(k+2))\,:\,k=1,2,\ldots\}\cup\{2\}$ then the 
    Temperley--Lieb--Jones C*-tensor category $\TLJcat(\delta)$ and the 
    category $\ModB^f$ are equivalent as braided C*-tensor categories.
\end{thm}

We end this section with a couple of remarks.

\begin{rem}
    Although we have not defined the conjugate of an arbitrary object in $\ModB^f$, we have shown that every such object is isomorphic to $F(P)$ for some object $P$ in $\TLJcat(\delta)$. Since $F$ is a monoidal *-functor, $F(P)$ has a conjugate (namely $F(\bar{P}) = F(P)$). Thus, every object in $\ModB^f$ does have a conjugate and $\ModB^f$ is in fact a rigid braided C*-tensor category. (See \cite{LR} and e.g.\ section 2.2 of \cite{NT} for the concepts of conjugates and rigidity in C*-tensor categories.)
\end{rem}

\begin{rem}\label{rem:categorification}
    Denote by $R$ the fusion ring of $\ModB^f$ consisting of formal differences $[X]-[Y]$ of isomorphism classes of modules in $\ModB^f$.
    Define a group homomorphism $\phi\colon R\to K_0(\mathcal{B})$ by $\phi([X]) = [\mathbf{p}]$, where $\mathbf{p} = \diag(p_1,\ldots,p_n)$ is any diagonal projection in $M_n(\mathcal{B})$ for which $X\cong\bigoplus_{j=1}^n p_j\mathcal{B}$. Since $\mathcal{B}$-linear maps $\bigoplus_{j=1}^n p_j\mathcal{B}\to \bigoplus_{i=1}^m q_i\mathcal{B}$ may be identified with $m\times n$-matrices whose $(i,j)$'th entry belongs to $q_i\mathcal{B}p_j$ in such a way that composition corresponds to matrix multiplication and adjoints correspond to matrix adjoints, we get that $\phi$ is well-defined. (In the module picture of $K_0(\mathcal{B})$, $[\mathbf{p}]$ corresponds to $[X\otimes\tilde{\mathcal{B}}]$, where $\tilde{\mathcal{B}}$ is the unitalization of $\mathcal{B}$.) The map $\phi$ is injective because $\mathcal{B}$ is an AF-algebra and hence admits cancellation. Since $\phi([q_s\mathcal{B}]) = [q_s]$ for all $s\in\mathcal{S}$, where $q_s$ is any minimal projection in the summand $\mathbb{K}(H^s)$ of $\mathcal{B}$, and the classes $[q_s]$ generate $K_0(\mathcal{B})$, it follows that $\phi$ is surjective.
    If $X = \bigoplus_{i=1}^n p_i\mathcal{B}$ and $Y = \bigoplus_{j=1}^m q_j\mathcal{B}$ then
    $
        X\otimes Y\cong \bigoplus_{(i,j)} \Phi(p_i\otimes q_j)\mathcal{B}
    $,
    and $K_0(\Phi)([\diag(p_1,\ldots,p_n)]\otimes [\diag(q_1,\ldots,q_m)])$ is the class of the diagonal $mn\times mn$-matrix whose $(i,j)$'th diagonal entry is $\Phi(p_i\otimes q_j)$.
    Thus, we may now conclude that $\phi$ is an isomorphism of rings. In this sense, the above equivalence of categories $\TLJcat(\delta)\cong \ModB^f$ ``categorifies'' the isomorphism $\mathbb{Z}[\mathcal{S}]\cong K_0(\mathcal{B})$ of rings that was exhibited in Remark \ref{rem:Kring}.
\end{rem}

\begin{rem}\label{rem:gen}
    Theorem \ref{thm} can in fact be proved in greater generality, as we next indicate. Let $\mathcal{C}$ be a finitely generated rigid (see Definition 2.2.1 of \cite{NT}) braided C*-tensor category. The assumption that $\mathcal{C}$ is rigid implies that $\mathcal{C}$ is semisimple (see section \ref{sec:semisimple}) and that each $\mathrm{End}_{\mathcal{C}}(\rho)$ is a finite-dimensional C*-algebra equipped with a canonical positive faithful trace (cf.\ \cite{LR}; see also \cite{NT}). The assumption that $\mathcal{C}$ is finitely generated means that there exists a finite set $\mathcal{L}$ of objects such that every simple object in $\mathcal{C}$ occurs as a direct summand of a tensor product of objects in $\mathcal{L}$.

    By a version of the Mac Lane Coherence Theorem (that can e.g.\ be deduced from the proof of Theorem XI.5.3 in \cite{Kassel}), we may assume that $\mathcal{C}$ is strict.
    Denote by $\pi$ a direct sum of the objects in $\mathcal{L}$.
    By Theorem 2.17 in \cite{BHP}, for example, the category $\mathcal{C}$ is equivalent, as a C*-tensor category, to the category $\mathcal{D}$ whose objects are formal finite sums $P_1\oplus\cdots\oplus P_k$ of projections $P_j\in \mathrm{End}_{\mathcal{C}}(\pi^{\otimes n_j})$ and whose morphisms $\oplus_j P_j\to \oplus_iQ_i$ are matrices whose $(i,j)$'th entry belongs to $Q_i\mathrm{Hom}_{\mathcal{C}}(\pi^{\otimes n_j},\pi^{\otimes m_i})P_j$ (when $Q_i\in\mathrm{End}_{\mathcal{C}}(\pi^{\otimes m_i})$). One can use the unitary braiding $\sigma$ on $\mathcal{C}$ to define a unitary braiding $\tilde{\sigma}$ on $\mathcal{D}$ by $\tilde{\sigma}_{P,Q} = \sigma_{\pi^{\otimes n},\pi^{\otimes m}}\circ (P\otimes Q)$ (when $P\in\mathrm{End}(\pi^{\otimes n})$ and $Q\in\mathrm{End}(\pi^{\otimes m})$). Then $\mathcal{C}$ and $\mathcal{D}$ are equivalent as braided C*-tensor categories.

    Next, put $\mathcal{G} = \{\mathbbm{1},\pi\}$ and choose $\mathcal{S}$ as 
    in section \ref{sec:semisimple}. Then, as in section \ref{sec:B}, we 
    construct a Hilbert space $H = \oplus_{s\in\mathcal{S}}H^s$ [where $H^s = 
    \oplus_{\vec x\in\mathcal{G}^\infty}\mathrm{Hom}_{\mathcal{C}}(s,o(\vec 
    x))$], operators $L_{\vec x,\vec y}(a)$, and a C*-algebra $\mathcal{B}$ 
    that is *-isomorphic to $\oplus_{s\in\mathcal{S}}\mathbb{K}(H^s)$. We can 
    also define a *-homomorphism 
    $\Phi\colon\mathcal{B}\otimes\mathcal{B}\to\mathcal{B}$ and equip $\ModB$ 
    with associators, unit constraints and a unitary braiding as in sections 
    \ref{sec:Phi} and \ref{sec:Monoidal} by using the well-known graphical 
    calculus for braided tensor categories (cf.\ e.g.\ \cite{Tur}). Hence, 
    $\ModB$ obtains the structure of a braided C*-tensor category. Finally, we 
    can define a braided monoidal *-functor $F\colon\mathcal{D}\to\ModB^f$ as 
    in section \ref{sec:Functor} and show, as in section \ref{sec:Equiv}, that 
    $F$ is an equivalence of categories. Thus, the initial category 
    $\mathcal{C}$ is equivalent to $\ModB^f$ as a braided C*-tensor category. 

    Let us finally mention some examples of categories to which this 
    generalization of Theorem \ref{thm} applies. Firstly, $\mathcal{C}$ 
    could be the representation category of a 
    compact group. Secondly, and more interestingly for us, $\mathcal{C}$ could 
    be a further example of the
    Verlinde fusion category in conformal field theory e.g.\ arising from the 
    finite-level, positive-energy representation theory of the loop group of a 
    compact, simple, connected, simply-connected Lie group (cf.\ \cite{PS}, 
    \cite{Wa2}). 
    (The 
    Temperley--Lieb--Jones category is the Verlinde fusion category arising 
    from 
    $\mathrm{SU}(2)$.) These latter categories can also be constructed from 
    certain 
    quantum 
    groups at roots of
    unity (cf. \cite{We2}; see also section 6A of \cite{EW}). Thirdly, there 
    are 
    examples 
    arising from the quantum double construction applied to not necessarily 
    braided categories, which yields braided C*-tensor categories. The 
    most prominent of these is the quantum double of the Haagerup subfactor, 
    which has 
    attracted much attention recently due to evidence that this system should 
    arise from a conformal field theory (cf.\ \cite{EG}).
\end{rem}

\section{Concluding remarks and outlook}\label{sec:Conclusion}

In the present paper, we have shown how to realize certain braided C*-tensor categories as categories of (right) Hilbert C*-modules with a natural tensor product structure (see Theorem \ref{thm} and Remark \ref{rem:gen}) or, phrased differently, how certain braided C*-tensor categories act faithfully on certain C*-algebras via Hilbert C*-modules. In light of this, it is natural to ask on which C*-algebras a given C*-tensor category (possibly without a unitary braiding) can act (faithfully) in this sense.
In this context, it may be noted that, starting from $\TLJcat(\delta)$, for example, one can define a variant of the Hilbert C*-bimodule $\mathcal{X}$ of Hartglass and Penneys (cf.\ \cite{HP1}) and use Pimsner's construction from \cite{Pi} to construct from it a Toeplitz type C*-algebra $\mathcal{T}$ that is $K\!K$-equivalent (by Theorem 4.4 of \cite{Pi}) to the C*-algebra $\mathcal{B}$ that appeared in the present paper. Perhaps this allows one to realize $\TLJcat(\delta)$ as a C*-tensor category of Hilbert $\mathcal{T}$-modules.

It is a long standing open problem to rigorously construct a conformal field 
theory (CFT) from a continuum scaling limit of a statistical mechanical model 
at criticality --- or to construct a CFT from a modular tensor category (cf.\ 
e.g.\ \cite{Pa}, \cite{CE}, \cite{PSa}, \cite{EG}, \cite{J14}, \cite{GS}, 
\cite{ILZ}, \cite{BGJST}). One aspect of this is to derive the category of 
representations of the Virasoro algebra from representations of Temperley--Lieb 
algebras $\mathrm{TL}_{N,N}^0(\delta)$ in the $N\to\infty$ limit in a 
mathematically rigorous way. The representation theory of the Virasoro algebra 
at central charge $c = 1-6/(k+2)(k+3)$, where $k=0,1,2,\ldots$, can be realized 
from the diagonal embedding $\mathsf{su}(2)_{k+1}\subset 
\mathsf{su}(2)_k\oplus\mathsf{su}(2)_1$ via a coset construction (cf.\ 
\cite{GKO}). Here, through the Sugawara construction, the affine Lie algebra 
$\mathsf{su}(2)_k$ has central charge $c_k = 3k/(k+2)$. It is then intriguing 
to ask whether there is a parallel coset construction starting from an 
embedding $\mathcal{B}^{(k)}\otimes\mathcal{B}^{(1)}\subset 
\mathcal{B}^{(k+1)}$, where $\mathcal{B}^{(k)}$ is the algebra constructed as 
above from the Temperley--Lieb category with parameter $\delta = 
2\cos(\pi/(k+2))$, that yields the representation category of the Virasoro 
algebra at central charge $c = 1-6/(k+2)(k+3)$.

\subsection*{Acknowledgements}

We would like to thank the following entities for their generous hospitality while research that would eventually lead to the present paper was carried out: The Isaac Newton Institute for Mathematical Sciences in Cambridge, England, during the research program {\it Operator Algebras: Subfactors and their Applications} in the spring of 2017; the Hausdorff Research Institute for Mathematics in Bonn, Germany, during the trimester program {\it von Neumann Algebras} in the summer of 2016; and the Dublin Institute for Advanced Studies in Dublin, Ireland, during a research visit in December 2017.
We would also like to thank the organizers of the conference {\it Young Mathematicians in C*-Algebras 2018} (YMC*A 2018) in Leuven, Belgium; the organizers of the Operator Algebra Seminar at the University of Copenhagen; and the organizers of the Analysis Seminar at Glasgow University (especially Jamie Gabe) for giving the first named author opportunities to present our work.
This research was supported by Engineering and Physical Sciences Research Council (EPSRC) grants EP/K032208/1 and EP/N022432/1.

\vspace{2mm}
{\footnotesize
\noindent {\sc Andreas N\ae s Aaserud} (ORCID: 0000-0001-9685-1640)\\
Email: andreas.naes.aaserud@gmail.com\\[1mm]
{\sc David E.\ Evans}\\
Email: evansde@cardiff.ac.uk\\[1.5mm]
\noindent {\it Both authors are affiliated with}\\[1mm]
{\sc School of Mathematics, Cardiff
University, Senghennydd Road, Cardiff, CF24 4AG, Wales, UK}
}

\begin{thebibliography}{8}
    \bibitem{AE} Aaserud, A.\ N.\ and Evans, D.\ E.: {\it K-theory of 
    AF-algebras from braided C*-tensor categories}, Rev.\ Math.\ Phys.\ (2020), 
    2030005. Online Ready. \url{https://doi.org/10.1142/S0129055X20300058}
    
    \bibitem{BG} Baki\'{c}, D.\ and Gulja\v{s}, B.: {\it Hilbert C*-modules 
    over C*-algebras of compact operators}, Acta Sci.\ Math.\ (Szeged), vol.\ 
    68 (2002), pp.\ 249--269.
    
    \bibitem{B} Bellet\^{e}te, J.: {\it The fusion rules for the 
    Temperley--Lieb algebra and its dilute generalization}, J.\ Phys.\ A: 
    Math.\ Theor., vol.\ 48 (2015), 395205 (67 pp). 
    \url{https://doi.org/10.1088/1751-8113/48/39/395205}

    \bibitem{BGJST} Bellet\^{e}te, J., Gainutdinov, A.\ M., Jacobsen, J.\ L., 
    Saleur, H.\ and Tavares, T.\ S.: {\it Topological defects in lattice models 
    and affine Temperley--Lieb algebra}, preprint (2018), 
    \href{https://arxiv.org/abs/1811.02551}{\texttt{arXiv:1811.02551}} [hep-th].

    \bibitem{BS} Bellet\^{e}te, J., and Saint-Aubin, Y.: {\it The principal 
    indecomposable modules of the dilute Temperley--Lieb algebra}, J.\ Math.\ 
    Phys., vol.\ 55 (2014), 111706 (41 pp). 
    \url{https://doi.org/10.1063/1.4901546}
    
    \bibitem{Bl} Blackadar, B.: ``$K$-Theory for Operator Algebras'', Second 
    Edition, Mathematical Sciences Research Institute Publications, 1998. 
    (First Edition published by Springer, 1986.)

    \bibitem{BHP} Brothier, A., Hartglass, M.\ and Penneys, D.: {\it Rigid 
    C*-tensor categories of bimodules over interpolated free group factors}, 
    J.\ Math.\ Phys., vol.\ 53 (2012), 123525 (43 pp). 
    \url{https://doi.org/10.1063/1.4769178}

    \bibitem{CE} Connes, A.\ and Evans, D.\ E.: {\it Embeddings of U(1)-current 
    algebras in non-commutative algebras of classical statistical mechanics}, 
    Comm.\ Math.\ Phys., vol.\ 121 (1989), pp.\ 507--525. 

    \bibitem{Coo} Cooper, B.: ``Almost Koszul Duality and Rational Conformal 
    Field Theory'', PhD thesis, University of Bath, 2007.
    
    \bibitem{DR} Doplicher, S.\ and Roberts, J.\ E.: {\it A new duality theory 
    for compact groups}, Invent.\ Math., vol.\ 98 (1989), pp.\ 157--218.
    
    \bibitem{EW} Erlijman, J.\ and Wenzl, H.: {\it Subfactors from braided C* 
    tensor categories}, Pacific J.\ Math., vol.\ 231 (2007), pp.\ 361--399.
    \url{http://dx.doi.org/10.2140/pjm.2007.231.361}

    \bibitem{EG} Evans, D.\ E.\ and Gannon, T.: {\it The exoticness and 
    realisability of twisted Haagerup--Izumi modular data}, Comm.\ Math.\ 
    Phys., vol.\ 307 (2011), pp.\ 463--512. 
    \url{https://doi.org/10.1007/s00220-011-1329-3}

    \bibitem{EP12} Evans, D.\ E.\ and Pugh, M.: {\it The Nakayama automorphism 
    of the almost Calabi--Yau algebras associated to SU(3) modular invariants}, 
    Comm.\ Math.\ Phys., vol.\ 312 (2012), pp.\ 179--222.
    \url{https://doi.org/10.1007/s00220-011-1389-4}
    
    \bibitem{F} Frank, M.: {\it Characterizing C*-algebras of compact operators 
    by generic categorical properties of Hilbert C*-modules}, J.\ $K$-Theory, 
    vol.\ 2 (2008), pp.\ 453--462.
    \url{https://doi.org/10.1017/is008001031jkt035}
    
    \bibitem{FHT1} Freed, D., Hopkins, M.\ and Teleman, C.: {\it Loop groups 
    and 
    twisted K-theory I}, J.\ Topology, vol.\ 4 (2011), pp.\ 737--798.
    \url{https://doi.org/10.1112/jtopol/jtr019}
    
	\bibitem{FHT2} Freed, D., Hopkins, M.\ and Teleman, C.: {\it Loop groups 
	and 
	twisted K-theory II}, J.\ Amer.\ Math.\ Soc., vol.\ 26 (2013), pp.\ 
	595--644.
	\url{https://doi.org/10.1090/S0894-0347-2013-00761-4}

	\bibitem{FHT3} Freed, D., Hopkins, M.\ and Teleman, C.: {\it Loop groups 
	and 
	twisted K-theory III}, Ann.\ Math.\ (2), vol.\ 174 (2011), pp.\ 947--1007.
	\url{http://dx.doi.org/10.4007/annals.2011.174.2.5}

    \bibitem{GS} Gainutdinov, A.\ M.\ and Saleur, H.: {\it Fusion and braiding 
    in finite and affine Temperley--Lieb categories}, preprint (2016), 
    \href{https://arxiv.org/abs/1606.04530}{\texttt{arXiv:1606.04530}} 
    [math.QA].
    
	\bibitem{GV} Gainutdinov, A.\ M.\ and Vasseur, R.: {\it Lattice fusion 
	rules 
	and logarithmic operator product expansions}, Nuclear Phys.\ B, vol.\ 868 
	(2013), pp.\ 223--270.
	\url{https://doi.org/10.1016/j.nuclphysb.2012.11.004}

    \bibitem{GLR} Ghez, P., Lima, R.\ and Roberts, J.\ E.: {\it W*-categories}, 
    Pacific J.\ Math., vol.\ 120 (1985), pp.\ 79--109.
    

    \bibitem{GKO} Goddard, P., Kent, A.\ and Olive, D.: {\it Unitary 
    representations of the Virasoro and super-$\!$Virasoro algebras}, Comm.\ 
    Math.\ 
    Phys., vol.\ 103 (1986), pp.\ 105--119.

    \bibitem{GHJ} Goodman, F.\ M., Harpe, P.\ de la and Jones, V.\ F.\ R.: 
    ``Coxeter Graphs and Towers of Algebras'', Mathematical Sciences Research 
    Institute Publications, Vol.\ 14, Springer, New York, 1989.

    \bibitem{Grimm} Grimm, U.: {\it Dilute algebras and solvable lattice 
    models}, in Proceedings of the Satellite Meeting of STATPHYS-19 on 
    Statistical Models, Yang--Baxter Equation and Related Topics, edited by M.\ 
    L.\ Ge and F.\ Y.\ Wu (World Scientific, Singapore, 1996), pp.\ 110--117.

    \bibitem{GJS} Guionnet, A., Jones, V.\ F.\ R.\ and Shlyakhtenko, D.: {\it 
    Random matrices, free probability, planar
    algebras and subfactors}, in Quanta of Maths: Conference in 
    Honor of Alain Connes, edited by E.\ Blanchard, D.\ Ellwood, M.\ 
    Khalkhali, M.\ Marcolli, H.\ Moscovici and S.\ Popa
    (Clay Math.\ Proc., Vol.\ 11, Amer.\ Math.\ Soc., Providence, RI, 2010), 
    pp.\ 201--239.
    
    \bibitem{HP1} Hartglass, M.\ and Penneys, D.: {\it C*-algebras from planar 
    algebras I: Canonical C*-algebras associated to a planar algebra}, Trans.\ 
    Amer.\ Math.\ Soc., vol.\ 369 (2017), pp.\ 3977--4019.
    \url{https://doi.org/10.1090/tran/6781}
    
    \bibitem{He} Henriques, A.\ G.: {\it What Chern--Simons theory assigns to a 
    point}, Proc.\ Natl.\ Acad.\ Sci.\ U.S.A., vol.\ 114 (2017), pp.\ 
    13418--13423.
    \url{https://doi.org/10.1073/pnas.1711591114}
    
    \bibitem{HiR} Higson, N.\ and Roe, J.: ``Analytic $K$-Homology'', Oxford 
    Mathematical Monographs, Oxford University Press, 2000.
    
    \bibitem{HL} Huang, Y.-Z.\ and Lepowsky, J.: {\it Tensor categories and the 
    mathematics of rational and logarithmic conformal field theories}, J.\ 
    Phys.\ A, vol.\ 46 (2013), 494009 (21 pp).
    \url{https://doi.org/10.1088/1751-8113/46/49/494009}

    \bibitem{ILZ} Iohara, K., Lehrer, G.\ I.\ and Zhang, R.\ B.: {\it 
    Temperley--Lieb at roots of unity, a fusion category and the Jones 
    quotient}, preprint (2017), 
    \href{https://arxiv.org/abs/1707.01196}{\texttt{arXiv:1707.01196}} 
    [math.RT].

    \bibitem{J} Jones, V.\ F.\ R.: {\it Index for subfactors}, Invent.\ Math., 
    vol.\ 72 (1983), pp.\ 1--25.
    
    \bibitem{J85} Jones, V.\ F.\ R.: {\it A polynomial invariant for knots via 
    Von Neumann algebras}, Bull.\ Amer.\ Math.\ Soc., vol.\ 12 (1985), pp.\ 
    103--111.
    
	\bibitem{J2} Jones, V.\ F.\ R.: {\it Planar algebras I}, preprint (1999),
	\href{https://arxiv.org/abs/math/9909027}{\texttt{arXiv:math/9909027}}
	[math.QA].

	\bibitem{J14} Jones, V.\ F.\ R.: {\it Some unitary representations of 
	Thompson's groups F and T}, J.\ Comb.\ Alg., vol.\ 1 (2017), pp.\ 1--44.
	\url{https://doi.org/10.4171/JCA/1-1-1} 

	\bibitem{JR} Jones, V.\ F.\ R.\ and Reznikoff, S.: {\it Hilbert space 
	representations of the annular Temperley--Lieb algebra}, Pacific J.\ Math., 
	vol.\ 228 (2006), pp.\ 219--249.
	\url{http://dx.doi.org/10.2140/pjm.2006.228.219}

    \bibitem{Kassel} Kassel, C.: ``Quantum Groups'', Graduate Texts in 
    Mathematics, Vol.\ 155, Springer, New York, 1995.

    \bibitem{Kau} Kauffman, L.\ H.: {\it State models and the Jones 
    polynomial}, Topology, vol.\ 26 (1987), pp.\ 395--407.

    \bibitem{Lance} Lance, E.\ C.: ``Hilbert C*-Modules'', London Mathematical 
    Society Lecture Note Series, Vol.\ 210, Cambridge University Press, 1995.
    
    \bibitem{LR} Longo, R.\ and Roberts, J.\ E.: {\it A theory of dimension}, 
    $K$-Theory, vol.\ 11 (1997), pp.\ 103--159.
    
    \bibitem{MPS} Morrison, S., Peters, E.\ and Snyder, N.: {\it Skein theory 
    for the $D_{2n}$ planar algebras}, J.\ Pure Appl.\ Algebra, vol.\ 214 
    (2010), pp.\ 117--139.
    \url{https://doi.org/10.1016/j.jpaa.2009.04.010}
    
    \bibitem{NT} Neshveyev, S.\ and Tuset, L.: ``Compact Quantum Groups and 
    Their Representation Categories'', Cours Sp\'{e}cialis\'{e}s, Collection 
    SMF, Soci\'{e}t\'{e} Math\'{e}matique de France, 2013.

    \bibitem{Pa} Pasquier, V.: {\it Two-dimensional critical systems labelled 
    by Dynkin diagrams}, Nuclear Phys.\ B, vol.\ 285 (1987), pp.\ 162--172.

    \bibitem{PSa} Pasquier, V.\ and Saleur, H.: {\it Common structures between 
    finite systems and conformal field theories through quantum groups}, 
    Nuclear Phys.\ B, vol.\ 330 (1990), pp.\ 523--556.

    \bibitem{Pi} Pimsner, M.\ V.: {\it A class of C*-algebras generalizing both 
    Cuntz--Krieger algebras and crossed products by $\mathbb{Z}$}, Fields 
    Inst.\ Commun., vol.\ 12 (1997), pp.\ 189-–212.
    
    \bibitem{Po} Popa, S.: {\it An axiomatization of the lattice of higher 
    relative commutants of a subfactor}, Invent.\ Math., vol.\ 120 (1995), pp.\ 
    427--445.
    
    \bibitem{PV} Popa, S.\ and Vaes, S.: {\it Representation theory for 
    subfactors, $\lambda$-lattices and C*-tensor categories}, Comm.\ Math.\ 
    Phys., vol.\ 340 (2015), pp.\ 1239--1280.
    \url{https://doi.org/10.1007/s00220-015-2442-5}

    \bibitem{PS} Pressley, A.\ and Segal, G.: ``Loop Groups'', Revised Edition, 
    Oxford Mathematical Monographs, Oxford Science Publications, Clarendon 
    Press, 1988. (Originally published in 1986.)
    
    \bibitem{RS} Read, N.\ and Saleur, H.: {\it Enlarged symmetry algebras of 
    spin chains, loop models, and S-matrices}, Nucl.\ Phys.\ B, vol.\ 777 
    (2007), pp.\ 263--315.
    \url{https://doi.org/10.1016/j.nuclphysb.2007.03.007}
    
    \bibitem{RS2} Read, N.\ and Saleur, H.: {\it Associative-algebraic approach 
    to logarithmic conformal field theories}, Nucl.\ Phys.\ B, vol.\ 777 
    (2007), pp.\ 316--351.
    \url{https://doi.org/10.1016/j.nuclphysb.2007.03.033}
    
    \bibitem{RLL} R\o rdam, M., Larsen, F.\ and Laustsen, N.\ J.: ``An 
    Introduction to $K$-theory for $C$*-Algebras'', London Mathematical Society 
    Student Texts, Vol.\ 49, Cambridge University Press, Cambridge, 2000.
    
    \bibitem{TL} Temperley, H.\ N.\ V.\ and Lieb, E.\ H.: {\it Relations 
    between 
    the `percolation' and `colouring' problem and other graph-theoretical 
    problems associated with regular planar lattices: Some exact results for 
    the `percolation' problem}, Proc.\ R.\ Soc.\ A, vol.\ 322 (1971), pp.\ 
    251--280.
    
    \bibitem{Tur} Turaev, V.\ G.: ``Quantum Invariants of Knots and 
    3-Manifolds'', de Gruyter Studies in Mathematics, Vol.\ 18, Walter de 
    Gruyter \& Co., Berlin, 1994.
    
    \bibitem{Wa2} Wassermann, A.\ J.: {\it Operator algebras and conformal 
    field theory. III. Fusion of positive energy representations of 
    $L\mathrm{SU}(N)$ using bounded operators}, Invent.\ Math., vol.\ 133 
    (1998), pp.\ 467--538.
    
    \bibitem{We} Wenzl, H.: {\it On sequences of projections}, C.\ R.\ 
    Math.\ Rep.\ Acad.\ Sci.\ Canada, vol.\ 9 (1987), pp.\ 5--9.
        
    \bibitem{We2} Wenzl, H.: {\it C* tensor categories from quantum 
    groups}, J.\ Amer.\ Math.\ Soc., vol.\ 11 (1998), pp.\ 261--282.
        
    \bibitem{Xu} Xu, F.: {\it Standard $\lambda$-lattices from quantum 
    groups}, Invent.\ Math., vol.\ 134 (1998), pp.\ 455--487.
        
    \bibitem{Yam} Yamagami, S.: {\it A categorical and diagrammatical 
    approach to Temperley--Lieb algebras}, preprint (2004),   
    \href{https://arxiv.org/abs/math/0405267}{\texttt{arXiv:math/0405267}} 
    [math.QA].
        
    \bibitem{Yuan} Yuan, W.: {\it Rigid C*-tensor categories and their 
    realizations as Hilbert C*-bimodules}, Proc.\ Edinburgh Math.\ Soc., vol.\ 
    62 (2019), pp.\ 367--393. \url{https://doi.org/10.1017/S0013091518000524}

\end{thebibliography}
\end{document}